\documentclass[11pt]{article}
\usepackage{amsthm}
\usepackage{amsmath}
\usepackage{dsfont}
\usepackage{amssymb}
\usepackage{graphicx}
\usepackage{hyperref}
\usepackage{tikz}
\usepackage{cases}
\usepackage{empheq}
\usepackage{enumerate}
\definecolor{darkgreen}{rgb}{0.1,0.5,0.1}
\hypersetup{colorlinks,linkcolor={red},citecolor={darkgreen},urlcolor={red}} 

\usetikzlibrary{arrows,positioning,fit,shapes,calc}

\newtheorem{Thm}{Theorem}
\newtheorem{Lem}[Thm]{Lemma}
\newtheorem{Def}[Thm]{Definition}

\newtheorem{Pro}[Thm]{Proposition}
\newtheorem{Cor}[Thm]{Corollary}
\newtheorem{Rem}[Thm]{Remark}
\numberwithin{equation}{section}

\newcommand{\beqa}{\begin{eqnarray}}
\newcommand{\eeqa}[1]{\label{#1}\end{eqnarray}}
\newcommand{\beq}{\begin{equation}}
\newcommand{\eeq}[1]{\label{#1}\end{equation}}

\newcommand{\Imag}{\mathop{\rm Im}\nolimits}

\newcommand{\calF}{\mathcal{F}}

\newcommand{\rmd}{{\mathrm{ d}}}
\newcommand{\rme}{{\mathrm{ e}}}
\newcommand{\rmi}{{\mathrm{ i}}}
\newcommand{\rmD}{{\mathrm{ D}}}
\newcommand{\R}{{\mathbb{ R}}}

\newcommand{\bbR}{{\mathbb{R}}}
\newcommand{\bbC}{{\mathbb{C}}}
\newcommand{\bbA}{{\mathbb{A}}}

\newcommand{\Hxy}{\boldsymbol{\mathcal{H}}}
\newcommand{\Hxydiv}{\boldsymbol{\mathcal{H}}_{\rm div0}}
\newcommand{\Hps}{\boldsymbol{\mathcal{H}}_{s}}
\newcommand{\Hms}{\boldsymbol{\mathcal{H}}_{-s}}
\newcommand{\hatH}{\boldsymbol{\hat{\mathcal{H}}}}
\newcommand{\indHx}{{\rm\scriptscriptstyle 1D}}
\newcommand{\Hx}{\boldsymbol{\mathcal{H}}_\indHx}

\newcommand{\eps}{\varepsilon}
\newcommand{\Om}{\Omega_{\rm m}}
\newcommand{\Oe}{\Omega_{\rm e}}
\newcommand{\Op}{\Omega_{\rm p}}
\newcommand{\Oc}{\Omega_{\rm c}}
\newcommand{\bk}{\mathbf{k}}

\newcommand{\bh}{\mathbf{h}}

\newcommand{\bm}{\mathbf{m}}
\newcommand{\bx}{\mathbf{x}}

\newcommand{\bu}{\mathbf{u}}

\newcommand{\bE}{\mathbf{E}}
\newcommand{\bH}{\mathbf{H}}
\newcommand{\bD}{\mathbf{D}}
\newcommand{\bB}{\mathbf{B}}
\newcommand{\bP}{\mathbf{P}}
\newcommand{\bM}{\mathbf{M}}
\newcommand{\bJ}{\mathbf{J}}

\newcommand{\bG}{\mathbf{G}}
\newcommand{\bU}{\mathbf{U}}
\newcommand{\bV}{\mathbf{V}}

\newcommand{\bbE}{\mathbb{E}}
\newcommand{\bbF}{\mathbb{F}}
\newcommand{\bbG}{\mathbb{G}}
\newcommand{\bbH}{\mathbb{H}}
\newcommand{\bbP}{\mathbb{P}}
\newcommand{\bbM}{\mathbb{M}}
\newcommand{\bbJ}{\mathbb{J}}

\newcommand{\bbU}{\mathbb{U}}
\newcommand{\bbV}{\mathbb{V}}

\newcommand{\bbW}{\mathbb{W}}
\newcommand{\bbI}{\mathbb{I}}
\newcommand{\calO}{{\mathcal{O}}}

\newcommand{\calZ}{{\mathcal{Z}}}

\def\wlkj{\boldsymbol{W}_{\!\! k,\omega,j}}
\def\wlkz{\boldsymbol{W}_{\!\! k,\omega,0}}
\def\wlk{\boldsymbol{W}_{\!\! k,\omega}}

\newcommand{\curl}{\operatorname{curl}}
\newcommand{\bcurl}{\operatorname{\bf curl}}
\def\div{\operatorname{div}}

\newcommand{\curlk}{\operatorname{curl_{\mathit k}}}
\newcommand{\bcurlk}{\operatorname{{\bf curl}_{\mathit k}}}

\newcommand{\bPi}{\boldsymbol{\Pi}}
\def\Rop{\mathrm{R}}
\def\bR{\boldsymbol{\rm{R}}}

\def\scD{\textsc{d}}
\def\scE{\textsc{e}}
\def\scI{\textsc{i}}
\def\scZ{\textsc{z}}
\def\DD{\textsc{dd}}
\def\DE{\textsc{de}}
\def\EI{\textsc{ei}}
\def\DI{\textsc{di}}
\def\EE{\textsc{ee}}
\def\zDD{\Lambda_\DD}
\def\zDE{\Lambda_\DE}
\def\zEI{\Lambda_\EI}
\def\zDI{\Lambda_\DI}
\def\zEE{\Lambda_\EE}
\def\zZ{\Lambda_\scZ}

\newcommand{\kE}{k_{\scE}}

\newcommand{\kc}{\kappa_{\rm c}}

\def\JacE{\mathcal{J}_{\scE}}
\newcommand{\calW}{\mathcal{W}}
\newcommand{\sgn}{\operatorname{sgn}}
\newcommand{\bhatU}{\boldsymbol{\hat{U}}}

\DeclareMathOperator*{\slim}{s-lim}

\def\Xint#1{\mathchoice
	{\XXint\displaystyle\textstyle{#1}}%
	{\XXint\textstyle\scriptstyle{#1}}%
	{\XXint\scriptstyle\scriptscriptstyle{#1}}%
	{\XXint\scriptscriptstyle\scriptscriptstyle{#1}}%
	\!\int}
\def\XXint#1#2#3{{\setbox0=\hbox{$#1{#2#3}{\int}$}
		\vcenter{\hbox{$#2#3$}}\kern-.5\wd0}}

\def\dashint{\Xint-}

\begin{document}
\vspace{-1in}

\title{An operator approach to the analysis of electromagnetic wave propagation in dispersive media. Part 2: transmission problems.}

\author{Maxence Cassier$^{a}$ and  Patrick Joly$^{b}$ \\ \\
{
\footnotesize $^a$ Aix Marseille Univ, CNRS,  Centrale Med, Institut Fresnel, Marseille, France}\\ 
{\footnotesize  $^b$ POEMS$^1$, CNRS, INRIA, ENSTA Paris, Institut Polytechnique de Paris, 91120 Palaiseau, France}\\ 
{\footnotesize (maxence.cassier@fresnel.fr, patrick.joly@inria.fr)}}

	\maketitle
	\begin{abstract}
In this second chapter, we analyse  transmission problems between a dielectric and a dispersive negative material. 
In the first part,  we consider  a transmission problem between two half-spaces,  filled respectively   by the vacuum and  a  Drude material, and separated by a planar interface. In this setting, we answer to the following  question: does this medium satisfy a limiting amplitude principle? This principle defines the stationary regime as the large time asymptotic behavior of a system subject to a periodic excitation. In the second part, we consider the transmission  problem  of  an infinite strip of Drude material  embedded in the vacuum and analyse the existence and dispersive properties of guided waves. In both problems,  our spectral analysis  enlighten new and  unusual  physical  phenomena  for the considered transmission problems due to the presence of the dispersive negative material. In particular,  we  prove the existence of an  interface  resonance   in the first part and the existence of slow light phenomena  for guiding waves in the second part.
	\end{abstract}	   

{\noindent \bf Keywords:} Maxwell's equations,  transmission problems between dielectrics and  passive metamaterials,  plasmonic waves, spectral theory, resonances, limiting absorption and limiting amplitude principles, guided waves, slow light phenomenon.

\section{Introduction}

This chapter is a second of two chapters devoted to the analysis of Maxwell's equations in dispersive media. In the first chapter, we essentially treat homogeneous media and in this second chapter we analyse transmission problems between dispersive media and non-dispersive media. This chapter is  based on  two articles  \cite{Cas-Haz-Jol-17,Cas-Haz-Jol-22} and two PhD theses \cite{Cas-14,Ros-23}. \\[6pt]
\noindent 
The present study is motivated by the modelling of wave  propagation  for transmission problems between dielectric media and metamaterials.
 Metamaterials are artificially microscopic structures whose macroscopic effective electromagnetic behavior is  dispersive materials.  Thus, the effective   electric permittivity $\eps$ and/or negative magnetic permeability $\mu$  is a function of the frequency and in addition, one or both of these functions  can become negative within some frequency range : in  such cases one says that the  material is negative  (see e.g. \cite{Gra-10,Gra-20}). \\[6pt]
 In the last twenty-five years,   transmission problems between dielectric media and metamaterials   have generated a huge interest among communities of physicists and mathematicians, owing to their extraordinary properties such as negative refraction \cite{Ves-68}, allowing the design of spectacular devices like the perfect flat lens \cite{Pen-00} or the cylindrical cloak in \cite{Nir-94}. Thanks to these negative electromagnetic coefficients, waves can propagate at the interface between such a negative material and a usual dielectric material \cite{Gra-12}.  Such a phenomenon can also be observed in metals in the optical frequency range \cite{Mai-07}. These waves, often called \emph{surface plasmons}, are localized near the interface and allow then to propagate signals in the same way as in an optical fiber, which may lead to numerous physical applications. 
Such transmission problems  have raised new questions in physics and mathematics. In particular, it is now well understood that in the case of a smooth interface between a dielectric and a negative material (both assumed non-dissipative), the time-harmonic Maxwell's equations become ill-posed if both ratios of $\eps$ and $\mu$ across the interface are equal to $-1$ (see e.g. \cite{Cos-Ste-85,Bon-14,Bon-14(2),Ngu-16,Bon-22}), which is precisely the conditions required for the perfect lens in \cite{Pen-00}. This result can be seen as the starting point of the present chapter where we in particular  characterize the situations where it is possible to consider  the time-harmonic solution at a frequency $\omega$ as the long-time behavior of the time-dependent equations subjected to a periodic forcing  source at $\omega$.
\\[6pt]
\noindent The rest of this chapter  is made of two sections:  the first section 2  treats the transmission problem between two homogeneous half-spaces separated by a planar interface: one half space is occupied by a non dispersive  dielectric (the vacuum) and the second one by a  Drude material. In section 3, we consider  an infinite strip of Drude material  embedded in the vacuum  which is precisely the geometrical setting introduced for the perfect lens in \cite{Pen-00}.  To analyse our problem, we use in both cases an abstract reformulation  of the Maxwell's evolution system:
\begin{equation}\label{eq.schrobis}
\frac{\rmd \, \bU}{\rmd\, t} + \rmi\, \bbA \, \bU=\bG,  \quad \mbox{ where $\bbA$ is a self-adjoint operator.} 
\end{equation}

\noindent For each case, we treat one classical questions in the mathematical analysis of linear of wave propagation models. In section \ref{sect-transm}, we address the question of the long-time behaviour of the electromagnetic fields when the medium is solicited by a causal time harmonic source  $\bG=\bbG \mathrm{e}^{-\rmi \omega t}$ with the existence of the so-called limiting amplitude principle. In section 3, we investigate the existence and dispersive properties of  waves guided by the Drude layer medium.  \\[6pt]
\noindent The outline of section \ref{sect-transm} is as follows. We describe the half-space transmission problem in section \ref{sec-math-model} and present the main results in section \ref{sec-large-limit}.  We analyse  the large time-behavior  of the solution and in particular emphasize unusual resonance phenomena in section \ref{sec-ampl}.  These results are related to the limiting absorption principle presented in section \ref{sec-abs-princip}.  In section \ref{spec-dens}, we define the main theoretical tools for our analysis. We introduce the spectral density $\omega \mapsto \bbM(\omega)$ of the operator $\bbA$ that we define in \ref{spec-dens-notion} and give the main properties in section \ref{sec-main-results-tools}. In section \ref{sec-construction-sepc-density}, we developed the construction of this density based on a Fourier reduction of the operator $\bbA$ (in section \ref{sec-reduced-op}), the construction of generalized  eigenfunctions (in section \ref{sec-geneigen}) and the diagonalization of the operator $\bbA$ via  a generalized  Fourier transform (in section \ref{sec-diag-th}). The above diagonalzation allows to construct the spectral density $ \bbM(\omega)$ in section \ref{sec-spec-density}.
Finally in section \ref{sec-proofs},  we prove  the main theorems of section \ref{sec-abs-princip}  (Theorem \ref{thm.limabs}) and  section  \ref{sec-ampl} (Theorem \ref{th.ampllim}).\\[6pt] The outline of section \ref{sec-slab} is as follows. Section \ref{sec-slab-model} describes the three layered transmission  problem. Section \ref{sec-common-prop}  recalls  common notations and properties  with  the problem  
analysed in section \ref{sect-transm}. In section  \ref{sec-guided-waves-def}, we define the notion of guided modes at frequency $\omega$ and wave-number $k$  (in the interface direction).  Moreover, one  shows that the existence of such guiding modes  is determined by the existence  of non-trivial odd or even solutions of a scalar dispersive  Sturm-Liouville equation. For the rest of the chapter, for length purpose,  we limit our analysis to the existence of even solutions of this  equation.  In section \ref{sec-disp-rel}, we  prove that the existence of non-trivial even solutions  is characterized by an implicit equation in the variable $(k,\omega)$ refereed to as the dispersion relation. In section \ref{sec-disp-cas}, we give the mathematical properties of  the solutions of the dispersion relation, named the dispersion curves. In section \ref{sec-disp-curves}, we  comment and illustrate (by numerical simulations)  the mathematical properties enlightened in section \ref{sec-disp-cas}.  In particular in section \ref{sec-disp-curves-genral-comments}, we compare these results to the ones obtained by the analysis   of the ``classical'' transmission problem involving  a  slab of dielectric medium (see appendix section \ref{classic-case}). Some comparisons are also made with the transmission problem analysed in section \ref{sect-transm}.
Finally, in section \ref{sec-disp-curves-slow-light}, we insist on  the existence of a slow light phenomenon due to presence of critical points on the dispersion curves where the group velocity  vanishes.

\section{Transmission problem involving  a dispersive media}\label{sect-transm}
\subsection{Description of the mathematical model}\label{sec-math-model}
We analyse here the propagation of transverse electric waves for a transmission problem between a dielectric (the vaccum) and a metamaterial (a Drude  dispersive medium). 
More precisely, we consider the Transverse Electric (TE) Maxwell's equations  in a bi-layered medium  composed by two materials: the vacuum and a Drude medium which fill  respectively two half-planes 
\begin{equation*} 
\bbR^2_- := \{\bx=(x,y)\in \bbR^2\mid x<0\} 
\quad\text{and}\quad
\bbR^2_+ := \{\bx=(x,y)\in \bbR^2\mid x>0\},
\end{equation*}
separated by a planar interface at $x=0$ (see figure  \ref{fig.med}).
\begin{figure}[h!]
\centering
 \includegraphics[width=0.4\textwidth]{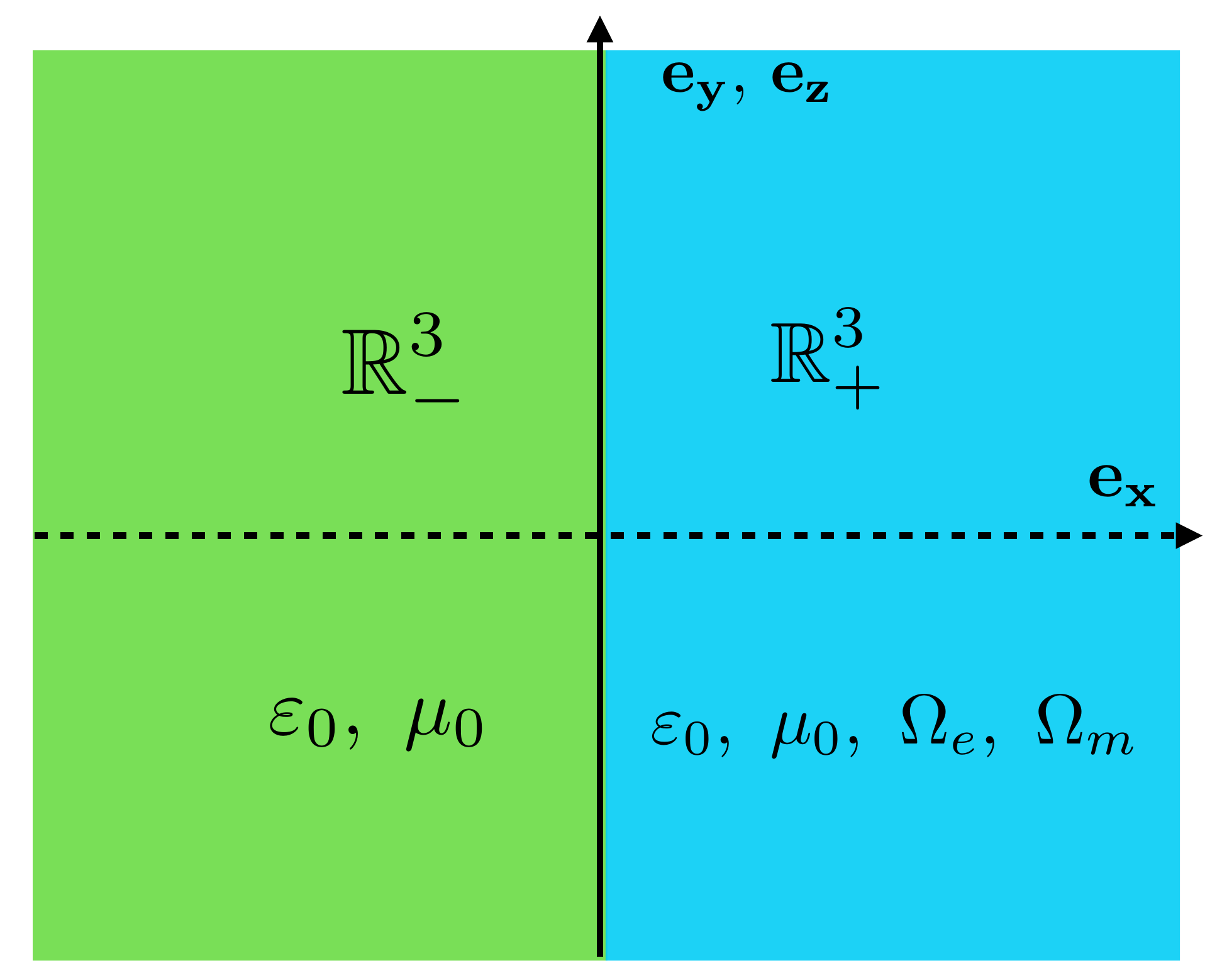}
 \hspace{0.75cm}
  \includegraphics[width=0.42\textwidth]{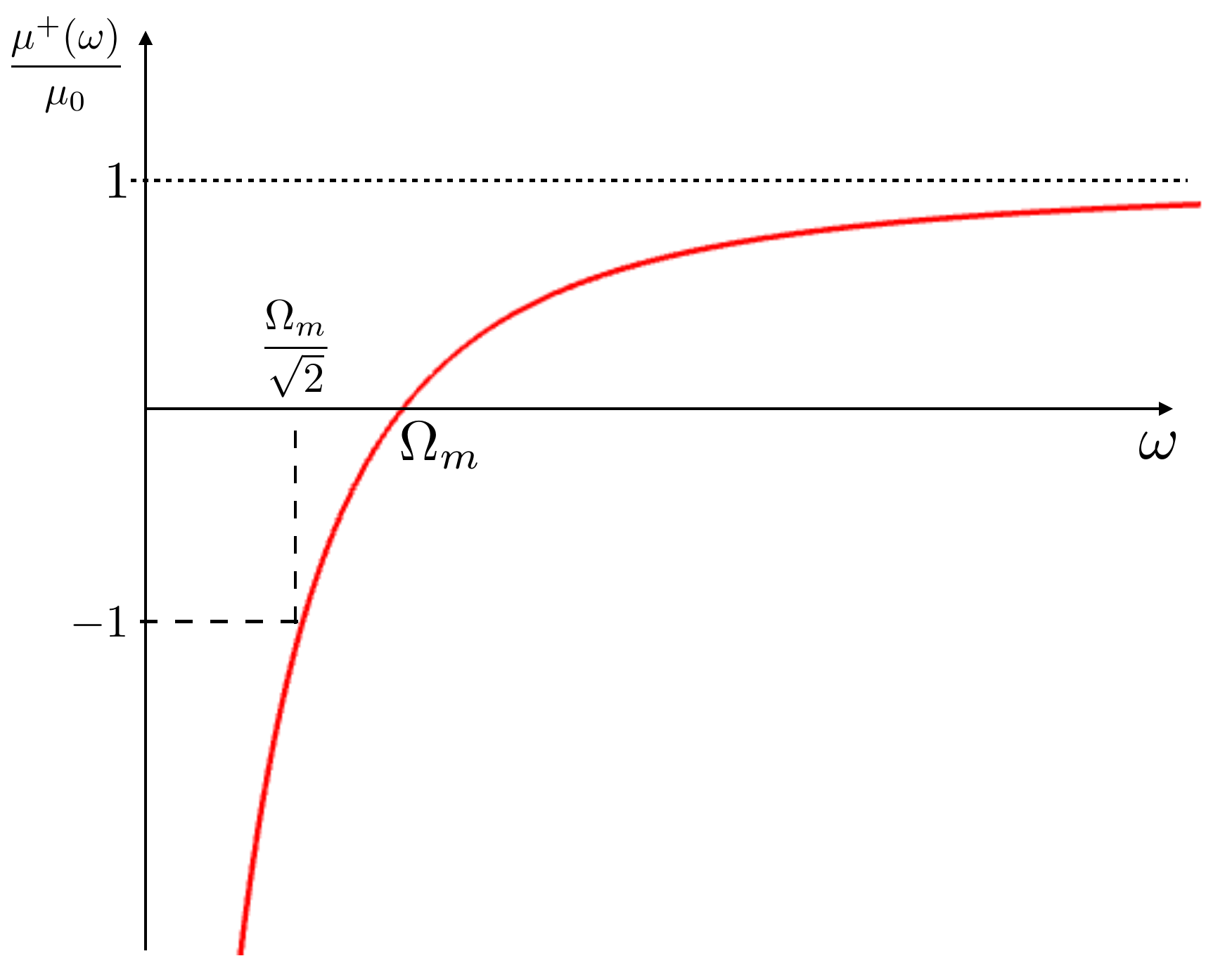}
 \caption{Left figure: Transmission problem between  the vacuum and a non-dissipative Drude  filling respectively $\mathbb{R}^3_-$   and  $\mathbb{R}^3_+$.
 Right figure:  Relative permeability function $\omega \mapsto \mu^{+}(\omega)/\mu_0$.}
 \label{fig.med}
\end{figure}
\\[4pt]
\noindent In this setting, the unknowns are the transverse electric induction and the magnetic induction: 
$$\bD(\bx,t) =D(\bx,t)\, {\bf e}_z \quad  \mbox{ and } \quad{\bf B}(\bx,t)=\big(B_x(\bx,t),B_y(\bx,t)\big)^{\top},$$  
and the transverse electric field and the magnetic  field:
$$\bE(\bx,t) =E(\bx,t) \,{\bf e}_z \quad  \mbox{ and } \quad{\bf H}(\bx,t)=\big(H_x(\bx,t),H_y(\bx,t)\big)^{\top}.$$  

\noindent In a presence of a current source term $\bJ_{\mathrm{s}}(\bx,t) =J_{\mathrm{s}}(\bx,t) \,{\bf e}_z $,  the evolution of $(D,{\bf B},E, \bH )$ is given by the following two dimensional  Transverse Electric (TE) Maxwell's equation:
\begin{equation}\label{eq.MaxwellTE}
\partial_t D -\curl \bH = -J_\mathrm{s} \quad  \mbox{ and } \quad \partial_t \bB +\bcurl  E = 0,
\end{equation}
where the 2D curl operators $\curl$ and  $\bcurl$ are respectively defined by
\begin{equation*}
\bcurl u := (\partial_y u, -\partial_x u)^{\top}
\quad\mbox{and}\quad 
\curl \bu := \partial_x u_y -\partial_y u_x \mbox{ for } \bu=(u_x,u_y)^{\top}.
\end{equation*}
For the derivation of the TE model  from the Maxwell's equations, we refer to  \cite{Cas-Haz-Jol-17}, section 2.2.\\[6pt]
The Maxwell's equations  have to be supplemented by the constitutive laws of each material:
\begin{equation}\label{eq.CL1}
\begin{array}{lllll}
& D=\varepsilon_0\,  E \quad   &\mbox{ and } &\quad \bB=\mu_0 \, \bH & \quad   \mbox{ in $\bbR^2_-$  \quad (i.e.  in the  vacuum)},\\
&  D=\varepsilon_0\,  E+P \quad  &   \mbox{ and } & \quad \bB=\mu_0 \, \bH +\bM &\quad   \mbox{ in $\bbR^2_+$  \quad (i.e.  in the  Drude material)}
\end{array}
\end{equation}
where the polarization  $P$ and magnetisation  $\bM$ satisfy the following ODE's equations in the Drude medium (i.e in $\bbR^2_+$):
\begin{equation}\label{eq.CL2}
\partial_t^2 P =\varepsilon_0 \, \Omega_e^2\, E  \  \ \mbox{ and }  \  \  \partial_t^2\bM= \mu_0 \, \Omega_m^2\,\bH ,
\end{equation}
 (where we point out that the two above equations can be obtained  from equations (6.4) and (6.6) of \cite{Cas-Jol-25} when  $N_e=N_m=1$ and $\omega_{e,1}=\omega_{m,1}=0$.)\\[6pt]
\noindent We define the two unknowns $\dot{P}:=\partial_t \, P$ and $\dot{\bM}:= \partial_t \bM$, referred in physics as the induced electric and magnetic currents. Eliminating the induction $(D,\bB)$ in 
(\ref{eq.MaxwellTE},\ref{eq.CL1},\ref{eq.CL2})  leads
to the following first order PDE's system on  $(E, \bH, \dot P, \dot \bM)$:
\begin{equation}\label{TE}
\left\{ \begin{array}{ll}
\varepsilon_0 \:\partial_t E -\curl \bH + \Pi \, \dot{P} = - J_{\mathrm{s} } & \mbox{in }  \bbR^2, \\[5pt]
\mu_0\: \partial_t \bH + \bcurl E + \bPi\, \dot{\bM}= 0 & \mbox{in } \bbR^2,\\[5pt]
\partial_t \dot{P} = \varepsilon_0  \Omega_e^2 \, \Rop \, E & \mbox{in } \bbR_+^2, \\[5pt]
\partial_t \dot{\bM}= \mu_0 \Omega_m^2 \, \bR \, \bH & \mbox{in } \bbR_+^2.
\end{array} \right.
\end{equation}
where the operator $\Pi$ (resp. $\bPi$) denotes the extension by $0$ of a scalar function (resp. a 2D vector field) defined on $\bbR^2_+$ to $\bbR^2$, whereas  $\Rop$ (resp., $\bR$) is the restriction operator to  $\bbR^2_+$ of a scalar function (respectively, a 2D vector field) defined on $\bbR^2$. 
The system \eqref{TE} is interpreted in the sense of distributions, thus  it contains implicitly  the transmission conditions 
\begin{equation}\label{eq.transmission}
[ \bE]_{x=0}=0 \quad\mbox{and}\quad [ \bH_y]_{x=0}=0,
\end{equation}
where $[f]_{x=0}$ denotes the gap of $f$ across the line $x=0.$   \eqref{eq.transmission} 
 expresses the continuity  of the tangential electric and magnetic fields at the interface $x=0$.\\[4pt]
\noindent Looking for time-harmonic solutions  of  \eqref{TE} at the frequency $\omega_{\mathrm{s}}\in \mathbb{R}$, i.e. solutions of the form $$(E(\bx,t),\bH(\bx,t), \dot P(\bx,t),\dot \bM(\bx,t)) =(\bbE(\bx),  \bbH(\bx) ,  \dot{\bbP}(\bx),  \dot \bbM(\bx))   \  \rme^{-\rmi \, \omega_{\mathrm{s}} \, t}$$   in a presence of  a time-harmonic source current $\bJ_s=\mathbb{J}_{\mathrm{s}}  \  \rme^{-\rmi \, \omega_{\mathrm{s}} t}$ yields the following time-harmonic Maxwell's equation in the medium (at the frequency $\omega$):
\begin{equation}\label{TE-harm}
\left\{ \begin{array}{ll}
-\rmi \, \varepsilon_0\,  \omega_{\mathrm{s}} \: \bbE -\curl \bbH + \Pi \, \dot \bbP = - \bbJ_{\mathrm{s} } & \mbox{in }  \bbR^2, \\[5pt]
-\rmi \, \mu_0\: \omega_{\mathrm{s}}\,   \bbH + \bcurl \bbE + \bPi\,\dot{ \bbM}= 0 & \mbox{in } \bbR^2,\\[5pt]
-\rmi \,  \omega_{\mathrm{s}} \, \dot{\bbP} = \varepsilon_0  \Omega_e^2 \, \Rop \, \bbE & \mbox{in } \bbR_+^2, \\[5pt]
-\rmi  \,  \omega_{\mathrm{s}} \,  \dot{\bbM}= \mu_0 \Omega_m^2 \, \bR \, \bbH & \mbox{in } \bbR_+^2.
\end{array} \right.
\end{equation}
After eliminating $(\dot{\bbP}, \dot{\bbM})$, it leads to the following time-harmonic equation in $(\bbE, \bbH)$:
\begin{equation*}\label{eq.Harmoniceq}
\rmi \, \omega_{\mathrm{s}} \, \eps(\omega_{\mathrm{s}},\cdot) \, \bbE + \bcurl \bbH= \bbJ_{{\rm s},\omega_{\mathrm{s}}}
\quad\mbox{and}\quad 
-\rmi \, \omega_{\mathrm{s}} \, \mu(\omega_{\mathrm{s}},\cdot) \,\bbH + \bcurl\bbE= 0
\quad\mbox{in }\bbR^2,
\end{equation*}
where the electric $\eps(\omega,\cdot)$ permittivity and magnetic permeability $\mu(\omega,\cdot)$ functions are given by
\begin{equation*}\label{eq.defepsmu}
\eps(\omega,\bx)  :=  \left\lbrace\begin{array}{ll}
         \eps^{-}(\omega):=\eps_{0} & \mbox{ if } x<0,\\[2pt]
       \displaystyle   \eps^{+}(\omega)=\eps_0 \Big( 1-\frac{ \Oe^2}{\omega^2}\Big) & \mbox{ if } x>0,\\[4pt]
       \end{array}\right.   \mbox{and }
         \mu(\omega, \bx)  :=  \left\lbrace\begin{array}{ll}
         \mu^{-}(\omega):=\mu_{0} & \mbox{ if } x<0,\\[2pt]
       \displaystyle   \mu^{+}(\omega):=\mu_0 \Big( 1-\frac{ \Om^2}{\omega^2}\Big) & \mbox{ if } x>0.\\[4pt]
       \end{array}\right.
\end{equation*}
The rational functions $\omega \mapsto \eps^{+}(\omega), \mu^{+}(\omega)$  (see figure \ref{fig.med}) characterize the dispersion law of the Drude material   (see example 30 of \cite{Cas-Jol-25}). They satisfy  the (HF) principle.  Moreover, one has
\begin{equation*}
\eps^{+}(\omega)<0 \mbox{ for } |\omega| \in (0,\Omega_e) \quad\mbox{and}\quad \mu^+(\omega)<0 \mbox{ for } |\omega| \in (0,\Omega_m).
\end{equation*}
Thus, the Drude medium behaves as  a negative index refractive material  with both  $\eps^{+}(\omega)<0$ and  $\mu^{+}(\omega)<0$ when $|\omega|\in \big(0,\min (\Oe, \Om)\big)$.
Furthermore, when $\Omega_e\neq\Omega_m$, there is a frequency gap $\big(\min(\Omega_e, \Omega_m), \max(\Omega_e, \Omega_m)\big)$  where  $\eps^{+}(\omega)$ and $\mu^{+}(\omega)$ have opposite signs. In this gap,  waves cannot propagate in the material bulk: they are evanescent with respect to the interface $x=0$. It is precisely what  happens for plasmonic waves in metals at optical frequencies \cite{Mai-07}. 
Finally there exists a unique frequency for which the relative permittivity $\eps^{+}(\omega) / \eps_0$ (respectively the relative permeability  $\mu^{+}(\omega) / \mu_0$) is equal to $-1$:
\begin{equation}\label{eq.ratio}
\frac{ \eps^{+}(\omega)}{\eps_0} = -1  \mbox{ if } |\omega| = \frac{\Omega_e}{\sqrt{2}}
\quad\mbox{and}\quad
\frac{\mu^{+}(\omega): }{\mu_0} = -1  \mbox{ if } |\omega| = \frac{\Omega_m}{\sqrt{2}}.
\end{equation}
\begin{Rem}[Critical case]\label{Rem-critica-case}
Note that in particular that both ratios in \eqref{eq.ratio} can be simultaneously equal to $-1$ at the same frequency if and only if $\Omega_e = \Omega_m$. This situation will be referred as the critical case in the following.
\end{Rem}
\noindent The evolution system \eqref{TE}  can be rewritten as a conservative evolution equation
\begin{equation}\label{eq.schro}
\frac{\rmd \, \bU}{\rmd\, t} + \rmi\, \bbA \, \bU=\bG,
\end{equation}
in the Hilbert space
\begin{equation}\label{eq.defHxy}
\mathcal{H} := L^2(\bbR^2) \times L^2(\bbR^2)^2 \times L^2(\bbR^2_+) \times L^2(\bbR^2_+)^2
\end{equation}
whose inner product is defined for all $\bU:=(E, \bH, \dot{\bP}, \dot{\bM})^{\top}$ and $\bU':=(E^{\prime}, \bH^{\prime},  \dot{\bP}^{\prime},  \dot{\bM}^{\prime})^{\top}\in \mathcal{H}$ by
\begin{equation}\label{eq.innerproduct}
(\bU,\bU')_{\mathcal{H}} := 
\int_{\bbR^2}\left(\eps_0\,E\,\overline{E^{\prime}} + \mu_0\,\bH\cdot\overline{\bH^{\prime}}\right)\rmd \bx 
+ \int_{\bbR^2_+}\left(\eps_0^{-1} \Oe^{-2}\,\dot{\bP} \,\overline{\dot{\bP}^{\prime}} + \mu_0^{-1} \Om^{-2}\, \bM\cdot\overline{\dot{\bM}^{\prime}}\right)\rmd \bx.
\end{equation}
The propagative operator  $\bbA$ in  \eqref{eq.schro} is the unbounded \emph{self-adjoint} operator on $\mathcal{H}$ defined by
\begin{align*}
\bbA & := \ \rmi\, \begin{pmatrix}
0 &\eps_0^{-1}\,\curl & -\eps_0^{-1} \, \Pi & 0\\
- \mu_0^{-1}\,\bcurl& 0 &0 & - \mu_0^{-1} \,\bPi \\
\eps_0 \Oe^2 \, \Rop & 0 & 0 &0 \\
0 & \mu_0 \Om^2\, \bR & 0 & 0
\end{pmatrix},
\end{align*}
where its domain $\rmD(\bbA)$ is given by
\begin{equation*}
\rmD(\bbA)  := H^{1}(\bbR^2) \times \bH_{\!\curl}(\bbR^2) \times  L^2(\bbR^2_+) \times L^2(\bbR^2_+)^2 \subset \mathcal{H}
\end{equation*}
with $\bH_{\!\curl}(\bbR^2) := \{ \bu\in  L^2(\bbR^2)^2 \mid \curl \bu \in L^2(\bbR^2)\}$.  We notice that a proof of the self-adjointness of $\bbA$ is given  in [\cite{Cas-Haz-Jol-17}, proposition 2] and in [\cite{Cas-14}, proposition 5.2.1]. 
Finally the source term $\bG$ in \eqref{eq.schro} is given by $\bG(t) := (- \eps_0^{-1} \, J_{\rm s}(\cdot,t) \,,0\,, \,0 ,\, 0)^{\top} \in \mathcal{H} $.\\[6pt]
\noindent For the sequel, we shall denote by $\sigma(\bbA)$  the spectrum of $\bbA$ and one has
 $\sigma(\bbA)\subset  \bbR$ (since $\bbA$ is self-adjoint). Accordingly, we introduce the corresponding  resolvent of $\bbA$: 
\begin{equation*}
R(\zeta):=(\bbA-\zeta\,\bbI)^{-1} \in B(\mathcal{H})\quad \text{for } \zeta \in \bbC\setminus\sigma(\bbA),
\end{equation*}
where $B(\Hxy)$ denotes the Banach algebra of bounded linear operators in $\Hxy$.\\[6pt]
\noindent We consider here for simplicity zero initial conditions (i.e. $\bU(0,t)=0$) in \eqref{eq.schro}). As $\bbA$ is  self-adjoint, the Hille-Yosida  theorem (see e.g. \cite{Bre-10,Paz-83}) implies that if $\bG\in C^1(\mathbb{R}^+, \mathcal{H})$ then  the evolution equation \eqref{eq.schro} admits a unique solution $\bU(t)\in  C^{1}(\bbR^{+},\mathcal{H})\cap C^{0}(\bbR^{+},\rmD(\bbA))$  given by
\begin{equation}\label{eq.duhamel}
\bU(t)=\int_{0}^{t} \rme^{-\rmi \, \bbA\, (t-\tau)}\, \bG(\tau) \,\rmd \tau,\quad \forall \, t \geq 0, 
\end{equation}
where $(\rme^{-\rmi\, \bbA\, t})_{t\in \bbR}$ is the group of unitary operators generated by  $\bbA$.
\begin{Rem}
As $\rme^{-i \, \bbA\, (t-\tau)}$ is unitary,   the Duhamel formula (\ref{eq.duhamel} implies that if $t \mapsto \|\bG(t)\|_{\mathcal{H}}$ is bounded on $\bbR^+$ (as e.g. for a causal time-harmonic source), then $\|\bU(t)\|_{\mathcal{H}}$ is linearly bounded:
\begin{equation}\label{eq.incrlint}
\|\bU(t)\|_{\mathcal{H}} \leq t \ \sup_{\tau \in \bbR^+} \|\bG(\tau)\|_{\mathcal{H}},\quad \forall\,  t \geq 0.
\end{equation}
\end{Rem}
\subsection{Large time behavior and limiting absorption principle}\label{sec-large-limit}
In this section, we recall the main results obtained in \cite{Cas-Haz-Jol-17} concerning the limiting amplitudes and limiting absorption principles 
for the evolution equation \eqref{eq.schro}. These principles have a long history  in scattering theory and more generally in mathematical physics. The pioneering ideas seem to date back to the early 1900's with the work of  W. von Ignatowsky \cite{Igna-05}. These principles were first proved rigorously  to our knowledge by  C. Morawetz \cite{Mora-62} for the propagation of waves in presence of sound soft obstacles in a homogeneous medium via energy techniques. Then D. Eidus  \cite{Eid-65,Eid-69} constructed an abstract proof via a spectral approach  and applied it to a class of acoustic media that are locally inhomogeneous. Eidus' approach was then developed by  S. Agmon \cite{Agm-75}, C. Wilcox \cite{Wil-84}, Y. Dermanjian, and J-C. Guillot  \cite{Der-86} and R. Weder \cite{Wed-91} for acoustic and electromagnetic stratified media using the notion of spectral decomposition of the underlying operator. Finally, it was extended  to other structures  such as waveguides \cite{Mor-89}, periodic media \cite{Rad-15}  \ldots  and  to other  waves equations: elastic waves \cite{Der-88,San-89}, water waves \cite{Vul-87,Haz-07}, \ldots. More  recently, the limiting amplitude principle, which allows to connect the  time-dependent  and  harmonic equations in wave propagation phenomena, has been  used  for numerical purpose in acoustic inhomogeneous media    to solve the Helmholtz equation in the high frequency regime, see e.g. \cite{Arn-24}.\\[6pt]
The method we use here is inspired from Eidus' spectral approach and its extension to stratified media. It is applied for the first time in the context of dispersive Maxwell's equations and  metamaterials which complicates significantly the analysis. 
\\[6pt]
\noindent We are interested here  to analyse the long-time behavior of the solution $\bU(t)$ given by \eqref{eq.duhamel} 
 when the source term $\bG$  is a casual periodic forcing  source term  at the source frequency $\omega_{\mathrm{s}}\in \mathbb{R}$:
 $$
 \bG(t)= \mathbb{G}\, \rme^{-\rmi \omega_s t} , \forall t\geq 0 \mbox{ with } \mathbb{G}\in \mathcal{H} .
 $$
For such source terms,  the formula \eqref{eq.duhamel} can be rewritten (using the functional calculus of self-adjoint operator) equivalently as 
\begin{equation*}
\bU(t) = \phi_{\omega_s,t}(\bbA)\,\bG,
\end{equation*}
where $\omega \mapsto \phi_{\omega_s,t}(\omega)$ the bounded continuous function defined by for all $t\geq 0$ by
\begin{equation}\label{eq.phiduhamel2}
\phi_{\omega_s,t}(\omega) := 
\rme^{-\rmi \omega\, t} \int_{0}^{t} \rme^{\rmi(\omega- \omega_{\mathrm{s}})\,\tau}\,\rmd \tau
= \left\lbrace\begin{array}{ll}
\displaystyle \rmi \,\frac{\rme^{-\rmi \omega\, t} -\rme^{-\rmi \omega_{\mathrm{s}} \, t}}{\omega-\omega_{\mathrm{s}}} & \mbox{ if } \omega \neq \omega_{\mathrm{s}} , \\[10pt]
t\, \rme^{-\rmi \omega_{\mathrm{s}}\,t} & \mbox{ if } \omega =\omega_{\mathrm{s}}.
\end{array}
\right.
\end{equation}
Usually, after a transient regime, the solution $\bU$  reaches a permanent regime and thus behave asymptotically for long time as a time-harmonic wave $\mathbb{U}_{ \omega_{\mathrm{s}}}^{+}=(\mathbb{E}_{\omega_{\mathrm{s}}}^{+},\mathbb{H}_{\omega_{\mathrm{s}}}^{+}, \dot{\bbP}_{\omega_{\mathrm{s}}}^{+},\dot{\bbM}_{\omega_{\mathrm{s}}}^{+})^\top$:
\begin{equation}
\bU(t)\sim   \, \mathbb{U}_{\omega_{\mathrm{s}}}^{+} \, \rme^{-\rmi \, \omega_{\mathrm{s}}\, t} \quad\text{as } t \to +\infty.
\label{eq.LAmP-imprecis}
\end{equation} 
The so-called \emph{limiting amplitude principle}  is closely related to the \emph{limiting absorption principle}, which  links  the definition of the field $\bbU_{\omega_{\mathrm{s}}}^{+}$ to the resolvent of $\bbA$ in \eqref{eq.LAmP-imprecis}  by (see remark \ref{Prem_LAP})
\begin{equation} \label{eq.refabsprincip}
\bbU_{\omega_{\mathrm{s}}}^{+} =- \rmi \, \lim_{\eta \searrow 0} R(\omega_{\mathrm{s}}+\rmi \eta) \ \mathbb{G} 
  \end{equation} 
assuming that the above limit exists (in a sense to be defined).
\begin{Rem}[Physical justification]  \label{Prem_LAP} Adding absorption to the equation \eqref{eq.schro} with  $\bG(t)= \mathbb{G}\, \rme^{-\rmi \omega_s t}$ consists in considering for $\eta > 0$ the approximate equation 
	\begin{equation}\label{eq.schro_abs}
		\frac{\rmd \, \bU_{\eta}}{\rmd\, t} + \eta \, \bU_\eta + \rmi\, \bbA \, \bU_\eta= \mathbb{G}\, \rme^{-\rmi \omega_s t},
	\end{equation} 
	Then looking for a time harmonic  solution  $\bU
	_\eta (t)= \mathbb{U}_{\omega_s,\eta} \, \rme^{-\rmi \omega_s t}$ of \eqref{eq.schro_abs} leads
	to 
	$$
	-\rmi  \, (\omega_s + \rmi \eta) \,  \mathbb{U}_{\omega_s,\eta}   +  \rmi\, \bbA \,  \mathbb{U}_{\omega_s,\eta} =  \mathbb{G}
	$$ 
that is to say $\mathbb{U}_{\omega_s,\eta} =  - \,  \rmi \, R(\omega_{\mathrm{s}}+\rmi \eta) \ \mathbb{G}$. Then looking for the limit of $\mathbb{U}_{\omega_s,\eta}$ when $ \eta \searrow 0$ leads (formally) to \eqref{eq.refabsprincip}.
\end{Rem}
\noindent Our aim with the following results   is to define a mathematical framework for a rigorous statement of these two principles and to make precise when these principles hold true or not. 

\subsubsection{The limiting absorption principle}\label{sec-abs-princip}
The proof of the limiting absorption principle strongly relies on   $\sigma(\bbA)$, the spectrum of $\bbA$, and its related properties.
First, we notice  that the limit \eqref{eq.refabsprincip} clearly exists in $\Hxy$ when $\omega_s$ belongs to the resolvent set $\bbC\setminus \sigma(\bbA)$  since 
$R(\omega_{\mathrm{s}}+\rmi \eta) $ converge to $R(\omega_{\mathrm{s}})$ in $B(\Hxy)$. 
Oppositely, if $\omega_{\mathrm{s}}\in \sigma(\bbA)$, the  limit   in $\Hxy$   cannot exist in general since for any $\bbG$ since $R(\omega_{\mathrm{s}}+ \rmi \eta)$ blows as $\eta^{-1}$ in $B(\Hxy)$.
Moreover, if $\sigma_p(\bbA)$ (the point spectrum of $\bbA$)  is not empty and $\omega_{\mathrm{s}}\in \sigma_p(\bbA)$, this is even worse since one needs to restrict the space of  $\bbG$   to expect the existence  of the limit  \eqref{eq.refabsprincip} even  in a weaker topology. Indeed, if $\omega_{\mathrm{s}}$ is an eigenvalue and $\bbG$ an associated eigenfunction  then
\begin{equation}\label{eq.expanresolvsp}
R(\omega_{\mathrm{s}}+\rmi \eta )= - \rmi  \, \frac {\bbG} {\eta }.
\end{equation}
The above observation leads us to first identify the spectrum and the point spectrum of $\bbA$. 
To this aim, we denote  by $\Op$ (``p'' for ``plasmonic'') and $\Oc$ (``c'' for ``cross point'')  the following particular frequencies:
\begin{equation}
\Op := \frac{\Om}{\sqrt{2}} \quad\text{and}\quad
\Oc := \frac{\Oe \, \Om}{\sqrt{\Oe^2+\Om^2}}
\label{eq.def-Op-Oc}
\end{equation}
Note that in the critical case (defined in Remark \ref{Rem-critica-case}), that is, when $\Oe = \Om$, we have $\Op = \Oc.$ It is also useful to introduce the following set of ``exceptional frequencies'' (whose role will be made clear later):
\begin{equation}
\sigma_{\rm exc} := 
\{0,\pm\, \Op,\pm \, \Om\}  \   \text{ if } \Oe \neq \Om \quad  \mbox{ and }  \quad  \sigma_{\rm exc} :=   \{0,\pm \,\Om\} \  \text{ if } \Oe = \Om.
\label{eq.def-sigma-exc}
\end{equation}
In addition, we also define the closed  sub-space of $\mathcal{H}$:
\begin{equation}
\Hxydiv :=\{ (\bE, \bH, \dot{\bP}, \dot{\bM})^\top \in \Hxy \mid \div \bH=0 \mbox{ in } \bbR^{2}_{\pm} \mbox{ and } \div \dot \bM=0 \mbox{ in } \bbR^2_+\}.
\label{eq.Hxydiv}
\end{equation}
We can now state our proposition on the spectrum  $\sigma(\bbA)$ and the point spectrum  $\sigma_{p}(\bbA)$  of $\bbA$   which gather various results  proved in \cite{Cas-Haz-Jol-17} (in the section 4.2.1 and the corollary 23).
\begin{Pro}\label{prop.spectrumA}
 $\sigma(\bbA)=\bbR$  and $\sigma_{p}(\bbA)$ is composed of eigenvalues of infinite multiplicity: 
\begin{equation}\label{eq.specpt}
\sigma_{p}(\bbA) = \{0,\pm\Om\}\ \text{ if }  \  \Oe \neq \Om \  \mbox{ and }  \  \sigma_{p}(\bbA)= \{0,\pm \,\Op,\pm \,\Om\} \ \text{ if } \Oe = \Om.
\end{equation}	
The eigenspaces $\ker(\bbA)$ and $\ker(\bbA\pm \, \Om)$ are respectively given by:
	\begin{align}
		\ker(\bbA) & =  \{ (0, \widetilde{\bPi}\, \nabla \phi, 0 , 0)^{\top} \ \mid \ \phi\in  W_0^1(\bbR^2_-)\},  \label{eq.kernel} \\
		\ker(\bbA \mp \Om \,\bbI) & = \left\{ \left(0, \,\bPi \,\nabla \phi, \,0 , \pm \rmi\mu_0\Om \,\nabla \phi \right)^{\top} \ \mid \ \phi\in W_0^1(\bbR^2_+)\right\},\label{eq.eigenspaces}
	\end{align}
	where $\widetilde{\bPi}$ is the extension operator by $0$ of a 2D vector fields defined on $\bbR^2_-$ to  $\bbR^2$ and $W_0^1(\bbR^2_{\pm})$ is  the Beppo-Levi space 
	$W_0^1(\bbR^2_{\pm}) := \{ \phi \in L^2_{\rm loc}(\bbR^2_{\pm})\ \mid  \nabla \phi  \in L^2(\bbR_{\pm}^2)^2 \mbox{ and } \phi|_{x=0} = 0 \}.$ Moreover,  the orthogonal complement of the direct sum of these eigenspaces is the space $\Hxydiv$, i.e.
	\begin{equation}\label{eq.Hxydiv2}
	\Hxydiv = \big( \ker \bbA \oplus \ker(\bbA+\Om \,\bbI)\oplus \ker(\bbA-\Om \,\bbI) \big)^\perp.
	\end{equation}
\end{Pro}
\noindent $\Hxydiv$ is   the Hilbert space  of propagative waves, i.e. the orthogonal of standing waves defined via a gradient field in the vacuum or the Drude medium.
We denote by  $\bbP_{\rm div0}$ the orthogonal projection on the subspace $\Hxydiv$ of $\Hxy$, by $\bbP_{\rm p}$ the orthogonal projection on the point subspace $\mathcal{H}_{\rm p}$ of $\bbA$,  that is, the direct sum of the eigenspaces associated with the eigenvalues of $\bbA$, and by $\bbP_{\pm\Op}$ the orthogonal projection on the eigenspace associated with $\pm \, \Op$. Finally we introduce
\begin{equation}\label{eq.def-Pac}
\bbP_{\rm ac} := \bbI - \bbP_{\rm p}
= \left\{\begin{array}{ll}
\bbP_{\rm div0} & \text{if } \Oe \neq \Om, \\[5pt] 
\bbP_{\rm div0}  - \bbP_{-\Op} - \bbP_{+\Op} & \text{if } \Oe = \Om,
\end{array}\right. 
\end{equation}
where the last equality follows from Proposition \ref{prop.spectrumA}. We use the index ``ac'' since we will see that $\bbP_{\rm ac}$ is  the orthogonal projection on the \emph{absolutely continuous} subspace $\mathcal{H}_{ac}$ associated with $\bbA$.\\[6pt]
The limiting absorption principle explores the behavior of the resolvent of $\bbA$, $R(\zeta)$
near the absolute continuous spectrum of $\bbA$. More precisely, we study the existence of the one-sided limits of the \emph{absolutely continuous} part of the resolvent near some $\omega_{\mathrm{s}} \in \bbR,$ that is,
\begin{equation}\label{eq.def-Rac}
R^{\pm}_{\rm ac}(\omega_{\mathrm{s}} ) := \lim_{\eta \searrow 0} R_{\rm ac}(\omega_{\mathrm{s}}  \pm \, \rmi \eta) 
\quad\text{where}\quad
R_{\rm ac}(\zeta) := R(\zeta)\,\bbP_{\rm ac} = \bbP_{\rm ac}\,R(\zeta)
\text{ for }\zeta \in \bbC\setminus\bbR.
\end{equation}
where the presence of the projector $\bbP_{\rm ac}$ is here for avoiding the point spectrum.\\[6pt]
\noindent The resolvent $R(\zeta)$ is an analytic function of $\zeta$ in $\bbC\setminus \bbR$ valued in $B(\Hxy)$. Thus if the limit \eqref{eq.def-Rac} exists for $\omega \in \sigma(\bbA)$, it has to be in a weaker topology than the topology of $B(\Hxy)$ (otherwise, it would imply that $\omega$ belongs to the resolvent set).
This weaker topology is defined  here via the weighted $L^2-$spaces:
\begin{equation*}
\Hps := L^2_{s}(\bbR^2) \times L^2_{s}(\bbR^2)^2 \times L^2_{s}(\bbR^2_+) \times L^2_{s}(\bbR^2_+)^2,
\end{equation*}
where $L^2_{s}(\calO) := \{ u \in L^2_{\rm loc}(\calO) \mid  \eta_{s}\,u \in L^2(\calO)\}$ for $\calO=\bbR^2$ or $\bbR^2_+$, and the weight $\eta_s$ is given by
\begin{equation*}
\eta_s(x,y) := (1 + x^2)^{s/2}\,(1 + y^2)^{s/2}.
\end{equation*}
The space $L^2_{s}(\calO)$ is naturally endowed with the norm
\begin{equation*}
\|u\|_{L^2_{s}(\calO)}^2 := \|\eta_s\,u\|_{L^2(\calO)}^2 = \int_\calO |\eta_s\,u|^2\,\rmd\bx.
\end{equation*}
Similarly, the Hilbert space $\Hps$ is equipped with the norm
$
\|\bU\|_{\Hps}  := \|\eta_s \,\bU\|_{\Hxy}.
$
For $s>0,$ the spaces $\Hps$ and $\Hms$ are dual to each other if $\Hxy$ is identified with its own dual space, which yields the continuous embeddings $\Hps \subset \Hxy \subset \Hms.$ The notation $\langle \cdot\,,\cdot\rangle_{s}$ stands for the duality product between them. This duality product extends the inner product of $\Hxy$ since
\begin{equation}\label{eq.innerproductbis}
\langle \bU \,, \bU' \rangle_{s} = (\bU \,, \bU')_{\Hxy} \quad\mbox{if }\bU \in \Hps \mbox{ and }\bU' \in \Hxy.
\end{equation}
For the topology for the limits \eqref{eq.def-Rac}, we choose the operator norm  $\| \cdot\|_{\Hps,\Hms}$ of $B(\Hps,\Hms)$, the Banach algebra  of bounded linear operators from $\Hps$ to $\Hms$. The proof of the following limiting amplitude principle  theorem  will use the H\"{o}lder regularity of these limits. This theorem aims at providing an optimal result with respect to this topology. 

\begin{Thm}[Limiting absorption principle]\label{thm.limabs}
Let $s>1/2$. For all $\omega_{\mathrm{s}}  \in \bbR \setminus \sigma_{\rm exc},$ the absolutely continuous part of the resolvent 
$R_{\rm ac}(\zeta)$ has one-sided limits 
$R^{\pm}_{\rm ac}(\omega_{\mathrm{s}} ) := \lim_{\eta \searrow 0} R_{\rm ac}(\omega_{\mathrm{s}}  \pm  \, \rmi \eta)$ for the operator norm of $B(\Hps,\Hms)$. Moreover, by denoting
\begin{equation}\label{eq.def-Rac-pm}
R^{\pm}_{\rm ac}(\zeta) := 
R_{\rm ac}(\zeta) \quad \text{if } \zeta \in \bbC^\pm := \{ \zeta\in\bbC \mid \pm \,  \Imag\zeta > 0 \},
\end{equation}
the function  $\zeta \mapsto R^{\pm}_{\rm ac}(\zeta) \in B(\Hps,\Hms)$ is analytic in $\bbC^\pm$ and locally H\"{o}lder continuous in $\overline{\bbC^\pm} \setminus \sigma_{\rm exc}.$ More precisely, for any compact set $K \subset \overline{\bbC^\pm} \setminus \sigma_{\rm exc}$, there exists a set $\Gamma_K \subset (0,1)$ of H\"{o}lder exponents such that for any $\gamma \in \Gamma_K,$ there exists $C_{K,\gamma}>0$ such that
\begin{equation*}
\forall (\zeta,\zeta') \in K \times K, \quad 
\Big\| R^{\pm}_{\rm ac}(\zeta') - R^{\pm}_{\rm ac}(\zeta) \Big\|_{\Hps,\Hms} 
\leq C_{K,\gamma} \ |\zeta'-\zeta|^{\gamma}.
\end{equation*}
The set $\Gamma_K$ is defined as follows: 
\begin{equation}
\Gamma_K := \left\{
\begin{array}{ll}
\big(0,\min(s-1/2,1)\big) & \text{if } K \cap \{ \pm\, \Oe,\pm \, \Oc \} = \varnothing, \\[5pt] 
\big(0,\min(s-1/2,1/2)\big) & \text{if } K \cap \{ \pm \,\Oe,\pm \, \Oc \} \neq \varnothing.
\end{array}\right.
\label{eq.def-Gamma-K}
\end{equation}
\end{Thm}
\noindent Theorem \ref{thm.limabs}  ensures the existence of the  limits $R^{\pm}_{\rm ac}(\omega)$, but the explicit spectral decomposition of these operators is  given in section  \ref{thm.limabs}  (see  \ref{eq.lim-res}). For the physical meaning of these limits,  we refer to the Remark \ref{eq.refabsprincip} and the following corollary. 
\begin{Cor}
Let $s>1/2$,  $\omega_{\mathrm{s}}  \in \bbR \setminus \sigma_{\rm exc}$  and $\bbG\in \Hps$ 
 then $\mathbb{U}_{\omega_s}^{\pm}$ defined via the Theorem \ref{thm.limabs} by 
\begin{equation*}
\mathbb{U}_{\omega_s}^{\pm}=  - \,  \rmi \, R_{\rm ac}^{\pm}(\omega_{\mathrm{s}} ) \ \mathbb{G}=   - \,  \rmi \lim_{\eta \searrow 0} R_{\rm ac}^{\pm}(\omega_{\mathrm{s}}  \pm \rmi \eta) \ \bbG
\end{equation*}
satisfies the time-harmonic Maxwell's equations   
\begin{equation}\label{eq.Harmonicabs}
\rmi \big(\bbA \, \mathbb{U}_{\omega_s}^{\pm}  -\omega_{\mathrm{s}} \mathbb{U}_{\omega_s}^{\pm}  \big)  =\bbP_{\mathrm{ac}} \bbG  \qquad  \mbox{ in } \Hms
\end{equation} which coincides in particular with   \eqref{TE-harm} for physical source terms of the form $\bbG=(\bbJ_{\mathrm{s}}, 0, 0, 0)^{\top}$ in the range of the orthogonal projector $\bbP_{\mathrm{ac}}$ (where $\bbP_{\mathrm{ac}}$ is given by \eqref{eq.def-Pac}).
\end{Cor}
\begin{proof}
We described the sketch of the proof of this corollary. \\[4pt] Let $\zeta=\omega_{\mathrm{s}}  \pm \rmi \eta$ with $\eta>0$. By  the definition of the resolvent $R_{\rm ac}(\zeta) $, one has $(\bbA-\zeta) R_{\rm ac}^{\pm}(\zeta) =  \bbP_{\rm ac}$  and therefore, it leads to the following identity
\begin{equation}\label{eq.reslvrelation}
\bbA R_{\rm ac}^{\pm}(\zeta) -  \zeta R_{\rm ac}^{\pm}(\zeta)= \bbP_{\rm ac}.
\end{equation}
The idea is to use Theorem \ref{thm.limabs} to pass the limit $\eta\searrow 0$ in the relation \eqref{eq.reslvrelation}.
To this aim, one defines, for $s>1/2$,  an Hilbert  space  $\rmD(\bbA)_{-s}$  by $\rmD(\bbA)_{-s}=\{\bU \in \Hms \mid \bbA  \, \bU \in \Hms \} $ (where $\bbA\, \bU$ is defined in the sense of distribution)  endowed with the norm
$$
  \|\bU\|_{\rmD(\bbA)_{-s}}=(\|\bU\|^2_{\Hms}+\| \bbA\bU\|^2_{\Hms})^{\frac{1}{2}}.
$$ 
One then  proves using Cauchy sequences  and the convergence of $R_{\rm ac}^{\pm}(\zeta)$ to $R_{\rm ac}^{\pm}(\omega_{\mathrm{s}} )$  in $B(\Hps,\Hms)$ when $ \eta\searrow 0$  that one can take the limit $ \eta\searrow 0$ in \eqref{eq.reslvrelation} and get  $\bbA R_{\rm ac}^{\pm}(\omega_{\mathrm{s}})-  \omega_{\mathrm{s}} R_{\rm ac}^{\pm}(\omega_{\mathrm{s}} )= \bbP_{\rm ac}$.
 In other words, $ R_{\rm ac}(\zeta)$ converges to $R_{\rm ac}^{\pm}(\omega_{\mathrm{s}} )$ in $B(\Hps, \rmD(\bbA)_{-s})$.  This  clearly implies \eqref{eq.Harmonicabs}.
\end{proof}
\noindent Thus, $\bbU_{\omega_{\mathrm{s}}}^+$ and  $\bbU_{\omega_{\mathrm{s}}}^-$  represent both time-harmonic solutions of our Maxwell equations \eqref{TE}. Their difference can be explained  via the \emph{limiting amplitude principle} which states that $\bbU_{\omega_{\mathrm{s}}}^+$ is \emph{outgoing}, in the sense of \eqref{eq.LAmP-imprecis}. Similarly, one shows that $\bbU_{\omega_{\mathrm{s}}}^-$ is \emph{incoming}, by considering the behavior as $t \to -\infty$ of anti-causal solutions.
\subsubsection{Limiting amplitude principle and interface resonance phenomenon}\label{sec-ampl}
We state now our  result on the long time behavior of the solution $\bU(t)$ of the evolution equation \eqref{eq.schro} when $\bU(0)=0$ and  $\bG(t) = \mathbb{G}\, \rme^{-\rmi\omega\, t}$ is  a time-harmonic excitation starting at $t=0$. 
\begin{Thm}\label{th.ampllim}
Let $s > 1/2$ and $\omega_{\mathrm{s}} \in \bbR \setminus \sigma_{\rm exc}$ (see \eqref{eq.def-sigma-exc}) and $\bbG \in \Hps \cap \Hxydiv$.\\[6pt]
\noindent On the one hand,  in the non-critical case $\Oe\neq \Om$, the limiting amplitude principle holds true, i.e. the solution $\bU(t)$ to \eqref{eq.schro} with $\bG(t) = \mathbb{G}\, \rme^{-\rmi\, \omega_{\mathrm{s}} \, t}$ for $t\geq 0$ and $\bU(0)=0$ satisfies
\begin{equation}
\lim_{t\to+\infty}\Big\| \bU(t) -\, \bbU_{\omega_{\mathrm{s}} }^{+} \, \rme^{-\rmi\omega_{\mathrm{s}}  t} \Big\|_{\Hms} = 0,
\label{eq.lim-ampl}
\end{equation}
where $\bbU_{\omega_{\mathrm{s}} }^{+} := -  \rmi \,  R^{+}_{\rm ac}(\omega_{\mathrm{s}} ) \, \mathbb{G}\in \Hms$ is given by the limiting absorption principle. \\[6pt]
\noindent  On the other hand, in the critical case $\Oe = \Om,$ one has
\begin{equation}
\lim_{t\to+\infty}\Big\| \bU(t) 
- \Big(\,\bbU_{\omega_{\mathrm{s}} }^{+} \, \rme^{-\rmi \omega_{\mathrm{s}}  t} + \sum_{\pm} \bbP_{\pm\Op} \mathbb{G}\ \phi_{\omega_{\mathrm{s}} ,t}(\pm \, \Op)\Big) \Big\|_{\Hms} = 0,
\label{eq.lim-ampl-reson}
\end{equation}
where $\bbP_{\pm\Op}$ is the orthogonal projection on the infinite dimensional eigenspace associated with $\pm \, \Op$, $\phi_{\omega,t}$ is defined in \eqref{eq.phiduhamel2} and $\bbU_{\omega_{\mathrm{s}} }^{+} :=   -  \rmi  \,R^{+}_{\rm ac}(\omega_{\mathrm{s}} ) \, \mathbb{G}$ as above.
\end{Thm}
\noindent Theorem  \ref{th.ampllim} shows that in the non-critical case $\Oe \neq \Om$, the limiting amplitude principle holds true for any $\bbG\in \Hxydiv \cap \Hps$ (since $\Hxydiv$ is here  the range of $\bbP_{\rm ac}$, see \eqref{eq.def-Pac}). The assumption $\bbG\in \Hxydiv$ forces the solution to remain orthogonal to the eigenspaces associated with the point spectrum $\{0,\pm \, \Om\}$. It is a natural physical assumption since we analyse the  propagative part of the solution which is orthogonal  $\ker(\bbA )$  and $\ker(\bbA \mp \Om )$ (see Proposition \ref{prop.spectrumA}). $0,$ $\pm \, \Om$ and $\pm \, \Op$ are excluded from Theorem  \ref{th.ampllim} . However, the limiting amplitude principle holds also for the frequencies $\omega_{\mathrm{s}}=\pm\Op$ but  they require a special treatment (see Remark \ref{rem-asymp}). \\[6pt]
\noindent In the critical case $\Oe=\Om$, the validity of the limiting amplitude principle depends on the spectral content of the source $\bbG \in \Hxydiv \cap \Hps$ which can be decomposed, using \eqref{eq.def-Pac}, as
\begin{equation*}
\bbG = \bbP_{\rm ac}\bbG + \bbP_{-\Op}\bbG + \bbP_{+\Op}\bbG.
\end{equation*}
If $\bbP_{-\Op}\bbG = \bbP_{+\Op}\bbG= 0$ (i.e when $\bbG$ belongs to the range of $\bbP_{\rm ac}$), the principle holds true for any $\omega_{\mathrm{s}} \in \bbR \setminus \{0,\pm \, \Om\},$ which includes in particular the frequencies $\omega = \pm \, \Op.$ But if $\bbP_{-\Op}\bbG \neq 0$ or $\bbP_{+\Op}\bbG \neq 0$, the behavior of $\bU(t)$ is no longer time-harmonic at the frequency $\omega$. Two situations may occur. Firstly, if $\omega \in \bbR \setminus \{0,\pm \, \Op,\pm \, \Om\},$ the solution $\bU(t)$ remains bounded in time but oscillates at the two frequencies $\omega_{\mathrm{s}}$ and $\Op$ (see the expression \eqref{eq.phiduhamel2} of $\phi_{\omega,t}(\pm \Op)$ for $\omega\neq \pm \, \Op$): it is a beat phenomenon. Secondly, if $\omega_{\mathrm{s}} = \pm \, \Op,$ there is no stationary regime at all, since $\bU$ blows up linearly in time (since by  \eqref{eq.phiduhamel2}: $\phi_{\omega,t}(\pm \, \Op)=t\, \rme^{\mp \rmi \Op t}$ for $\omega=\pm \Op$). 
This  corresponds  to a resonance phenomena. Such a phenomenon is classical for vibration problems in bounded domains but quite unusual for unbounded domains. Here, the fields $\bbP_{\pm \,\Op}\bbG$ are trapped waves which belongs to $\Hxy$ and are defined as (continuous) superpositions of functions  which are exponentially decaying with the distance to the interface (\textit{i.e.}, \emph{plasmonic waves}): they give birth to an \emph{interface resonance phenomenon} (which  has been enlighten in the physical literature in \cite{Gra-12} and whose existence has been first proved  mathematically in \cite{Cas-14,Cas-Haz-Jol-17,Cas-Haz-Jol-22}). The linear behavior in time is characteristic to a resonance due to an eigenvalue for a self-adjoint operator. Here, the eigenvalues $\pm\, \Op$ are of infinite multiplicities and embedded in the continuous spectrum of $\bbA$. This  interesting resonance phenomenon  does not occur  in a  stratified media made of standard dielectric materials (see \cite{Wed-91}) since such non-zero eigenvalue of the Maxwell operator does not exist.
This conclusion confirms the strong ill-posedness of the time-harmonic transmission problems described in \cite{Bon-14,Bon-14(2),Ngu-16}  when the relative permittivity   $\eps^{+}(\omega) / \eps_0$ and  permeability $\mu^{+}(\omega) / \mu_0$ are simultaneously equal to $-1$. 

\subsection{Spectral density of the propagative operator}\label{spec-dens}

\subsubsection{The  notion of spectral density}\label{spec-dens-notion}
A scalar Borel measure is  absolutely continuous  if  it is ``proportional'' to the Lebesgue's measure via a $L^1_{\mathrm{loc}}$ density function. 
We extend  this notion  for Bochner integrals of  operator-valued function in $ B(\Hps, \Hms)$. Roughly speaking, the spectral density is here  an operator-valued function  in $ B(\Hps, \Hms)$ associated with the functional calculus of  the absolute continuous part of  $\bbA$. It allows to rewrite operators $f(\bbA)\,\bbP_{\rm ac}$ (for any bounded measurable  $f:\bbR\to \bbC$)  as
\begin{equation}\label{eq.specdensity}
f(\bbA)\,\bbP_{\rm ac} = \int_{\bbR} f(\omega) \, \bbM_{\omega} \,\rmd \omega, \ \mbox{ where }\omega \mapsto \bbM_{\omega}\in B(\Hps, \Hms) \mbox{ is locally integrable}.
\end{equation}
\noindent  In particular, applying formula \eqref{eq.Bochner}  to $f=\mathds{1}_S$ (the indicator function of a Borel set $S$), one can show that the spectral measure $\bbE$ (see \cite[\S 2.3]{Cas-Haz-Jol-17} for a brief reminder about this notion) of  $\bbA$ is \emph{absolutely continuous} in the range of $\bbP_{ac}$. Namely, the spectral measure of $\bbA$ is ``proportional'' to the Lebesgue measure in the orthogonal complement of the point subspace $\mathcal{H}_{pt}$.  This explains the terminology \emph{spectral density} for $\bbM_{\omega},$ and the notation $\bbP_{\rm ac}.$ This is the object of Corollary \ref{cor.specmes} which connect the spectral density to the spectral measure of  $\bbA$ and  shows in particular  that outside the eigenvalues of $\bbA,$ the spectrum of $\bbA$ is absolutely continuous.\\[6pt]
\noindent As we will see, the existence and the regularity  of the spectral density are essential  for the proofs of Theorem \ref{thm.limabs} and \ref{th.ampllim}. Indeed, on the one hand, for the limiting absorption principle at a given frequency $\omega \in \bbR,$ we have to consider  functions $r_{\omega\pm\, \rmi\eta}: \bbR \to \bbC$ for $\eta > 0$ defined by
\begin{equation}
r_{\omega_{\mathrm{s}}\pm \, \rmi\eta}(\omega) := \frac{1}{\omega- (\omega_{\mathrm{s}}   \pm \,  \rmi\eta)}.
\label{eq.def-r-eta}
\end{equation}
On the other hand, for the limiting amplitude principle, we will examine the behavior of $\phi_{\omega,t}(\bbA)$ as $t \to +\infty,$ where $\phi_{\omega,t}(\cdot)$ is defined in \eqref{eq.phiduhamel2}. Both limiting processes are intimately connected. We focus here on the former to explain the motivation of the notion of spectral density. Using \eqref{eq.specdensity} for the function $f=r_{\omega_{\mathrm{s}}  \pm \, \rmi\eta}$ defined in \eqref{eq.def-r-eta}, the absolutely continuous part of the resolvent of $\bbA$ (see \eqref{eq.def-Rac}) appears as a Cauchy integral
\begin{equation*}
R_{\rm ac}(\omega_{\mathrm{s}} \pm \rmi\eta) := R(\omega_{\mathrm{s}} \pm \, \rmi\eta)\, \bbP_{\rm ac} = 
\int_{\bbR} \frac{\bbM_{\omega}}{\omega - (\omega_{\mathrm{s}} \pm\,  \rmi\eta)} \,\rmd \omega,
\end{equation*}
whose limits as $\eta \searrow 0$ will be given by a suitable version of the well-known Sokhotski--Plemelj formula \cite{Hen-86}, provided that $\omega \mapsto \bbM_{\omega}$ is locally H\"{o}lder continuous. This is actually the  objectives of the next  section which states the existence of a spectral density (for which formula \eqref{eq.specdensity} holds)  and gives the local H\"{o}lder regualrity  of $\bbM_{\omega}.$

\subsubsection{Main results on the spectral density}\label{sec-main-results-tools}
We state now the two main theorems and a corollary  (proved in \cite{Cas-Haz-Jol-22})   on the spectral density.

\begin{Thm}\label{Thm.specdensity}
	Let $s > 1/2.$ There exists a   spectral density $\omega \in \bbR \setminus \sigma_{\rm exc} \mapsto \bbM_{\omega}\in B(\Hms, \Hps)$,  locally integrable on $ \bbR \setminus \sigma_{\rm exc}$, such that
for any bounded function $f: \bbR \to \bbC$ with a compact support which does not intersect $\sigma_{\rm exc}$ (see \eqref{eq.def-sigma-exc}), the operator $f(\bbA)\,\bbP_{\rm ac}$ is given by
	\begin{equation}\label{eq.Bochner}
	f(\bbA)\,\bbP_{\rm ac} = \int_{\bbR} f(\omega) \, \bbM_{\omega} \,\rmd \omega .
	\end{equation}
\begin{Rem}
We make here some comments on Theorem \ref{Thm.specdensity}.
First, we  point out that the existence of such spectral density at an operator level  can not be deduced in a general framework directly  from the spectral theorem.
Then, one shows easily that formula \eqref{eq.Bochner} defines the spectral density uniquely  almost everywhere on $\bbR$ (in the sense of the Lebesgue's measure). We construct  the function $\omega\mapsto \bbM_{\omega}$  in section \ref{sec-spec-density} and gives its explicit formula  for the non-critical case  $\Oe \neq \Om$ in \eqref{eq.density-non-crit}  and for the critical case in  \eqref{eq.density-crit} if $\Oe = \Om$.  Finally, we  give  an expression of $f(\bbA)\,\bbP_{\rm ac}$ when $f:\bbR \to \bbC$ is a  bounded function whose support $S$  is not compact and/or intersects $\sigma_{\rm exc},$   via a  limiting process using \eqref{eq.Bochner}, see section \ref{sec-spec-density},  formula \eqref{eq.calc-fonct-ac-lim}.
\end{Rem}	
\end{Thm}

\begin{Cor}\label{cor.specmes}
	$\bbE(\cdot)$, the spectral measure of the self-adjoint operator $\bbA$, satisfies
	\begin{equation}\label{eq.absolutecontinuousmes2}
	\forall \bU, \bV \in \Hps,\quad 
	\rmd\big(\bbE(\omega)\bbP_{\rm ac} \bU, \bbP_{\rm ac} \bV\big)_{\Hxy} =
	\rmd\big(\bbE(\omega)\bbP_{\rm ac} \bU, \bV\big)_{\Hxy} =  \langle\bbM_{\omega}\bU,\bV\rangle_s \,\rmd\omega.
	\end{equation}
	Moreover, for any Borel set $S \subset \bbR$ (bounded or not), we have
	\begin{equation}\label{eq.norm-Eac}
	\forall \bU \in \Hps,\quad 
	\left\| \bbE(S)\,\bbP_{\rm ac} \, \bU \right\|_{\Hxy}^2 
	= \int_\bbR\mathds{1}_{S}(\omega) \, \langle\bbM_{\omega} \bU, \bU \rangle_s\,\rmd\omega,
	\end{equation}
	where the function $\omega \mapsto \langle\bbM_{\omega} \bU, \bU \rangle_s$ is non-negative and integrable on $\bbR$.
\end{Cor}

\begin{proof}
The first equality of \eqref{eq.absolutecontinuousmes2} simply follows from the fact that $ \bbP_{\rm ac}=\bbE(\bbR \setminus \sigma_{p}(\bbA))$ is an  orthogonal projection which commutes with $\bbE(S)$ for any Borel set $S \subset \bbR$.\\[6pt]
We prove now the second equality of of \eqref{eq.absolutecontinuousmes2}.  As $\bbE(S)=\mathds{1}_S(\bbA)$, by virtue of \eqref{eq.Bochner} apply to $f=\mathds{1}_S$, we get that  for any bounded Borel set $S \subset \bbR$ such that  $\overline{S} \cap \sigma_{\rm exc} = \varnothing$: 
	\begin{equation*}
	\forall \bU, \bV \in \Hps,\quad 
	\big( \bbE(S)\,\bbP_{\rm ac} \, \bU, \bV \big)_{\Hxy} = 
	\Big\langle \int_\bbR\mathds{1}_{S}(\omega) \, \bbM_\omega \bU \,\rmd\omega\,, \bV \Big\rangle_s.
	\end{equation*}
	As the above integral is Bochner, we can permute it with the duality product \cite[Theorem 3.7.12]{Hil-96}, which yields
	\begin{equation}
	\forall \bU, \bV \in \Hps,\quad 
	\big( \bbE(S)\,\bbP_{\rm ac} \, \bU, \bV \big)_{\Hxy} = 
	\int_\bbR\mathds{1}_{S}(\omega) \, \langle \bbM_\omega \bU, \bV \rangle_s \,\rmd\omega.
	\label{eq.expres-Eac}
	\end{equation}
Besides, as  $\sigma_{\rm exc}$ is composed of a finite number of elements (see  \eqref{eq.def-sigma-exc}),   it is clear that  $\bbE(\sigma_{\rm exc})\,\bbP_{\rm ac} =\bbE(\sigma_{\rm exc}\setminus \sigma_{p}(\bbA))= 0$ and that   $\sigma_{\rm exc}$ has zero Lebesgue's measure.  Thus, using  the sigma-additivity of the spectral measure $\bbE(\cdot)$, one first extends \eqref{eq.expres-Eac}   to any  bounded  Borel set $S$ (even it intersects $\sigma_{\rm exc}$) and then in a second time to any Borel set of $\bbR$. This leads to the second equality of \eqref{eq.absolutecontinuousmes2}.\\[6pt]
\noindent Finally, if we choose $\bV = \bU$ in \eqref{eq.expres-Eac}, we obtain \eqref{eq.norm-Eac} for $S$ bounded. The fact that it holds true for unbounded $S$ follows from the spectral theorem which ensures that for all $\bU \in \Hxy,$ the map $S \mapsto (E(S)\bU,\bU)_{\Hxy}$ defines a non-negative finite Borel measure.
\end{proof}

\begin{Thm}\label{th.Holder-dens-spec}
	Let $s > 1/2.$ The spectral density $\omega \mapsto \bbM_\omega \in B(\Hps,\Hms)$ is locally H\"{o}lder-continuous on $\bbR \setminus \sigma_{\rm exc}$. More precisely, let  $[a,b]\subset \bbR \setminus \sigma_{\rm exc}$ and $\Gamma_{[a,b]} \subset (0,1)$ be the set of H\"{o}lder exponents defined by  \eqref{eq.def-Gamma-K} for  $K=[a,b]$. Then, one has 
	\begin{equation}\label{eq.holderestim-specdensity}
\forall \, \gamma \in \Gamma_{[a,b]}, \ \ \exists \, C_{a,b}^{\gamma}>0 \ \  \mid	\forall  \, \omega', \, \omega \in [a,b], \ \
	\big\| \bbM_{\omega'} - \bbM_\omega \big\|_{\Hps,\Hms} \leq C_{a,b}^{\gamma} \ |\omega'-\omega|^{\gamma}.
	\end{equation}
\end{Thm}
\begin{Rem}
The proof of  Theorem \ref{th.Holder-dens-spec} is very  technical and thus not detailed here for shortness purposes. Indeed, it corresponds to  the sections 3.3 and 3.4 which occupy  20 pages of \cite{Cas-Haz-Jol-22}. It is based on precise  $\Hms$-H{\"o}lder type estimates in $\omega$   of the generalised eigenfunctions $\bbW_{k,\omega,j}$ introduced in the next paragraph. These estimates have to be done meticulously  since they make
the values of the local H{\"o}lder exponent  $\gamma$ dependent  of   the weight index $s$ of the space $\Hms$.
\end{Rem}

\subsection{Construction of the spectral density}\label{sec-construction-sepc-density}
The construction of the spectral density function $\omega \mapsto \bbM_{\omega}$ rely on the spectral analysis of the self-adjoint operator $\bbA$ and the explicit construction of a generalized Fourier transform $\bbF$ performed in \cite{Cas-Haz-Jol-17}. $\bbF$  diagonalizes  $\bbA$, in the sense that it is a unitary transformation from the \emph{physical} space $\Hxy$ (defined in \eqref{eq.defHxy}) into a second Hilbert space $\hatH$, named the \emph{spectral} space, in which  $\bbA$ takes a diagonal form. 
Therefore, the goal of this section is to recall all the results and notations from \cite{Cas-Haz-Jol-17} that are necessary to introduce the operator $\bbF$.

\subsubsection{Reduced operator $\bbA_{\bk}$}\label{sec-reduced-op}
To construct the generalized Fourier transform $\bbF$, one exploits the invariance of the medium in the $y-$direction to reduce the problem's dimension.  It allows to  decompose  the operator  $\bbA$ into a family of operators $(\bbA)_{k\in \bbR}$ which acts on a Hilbert space $\Hx$ of functions depending only on the variable $x$. \\[6pt]
To this aim, one introduces  $\calF$ the Fourier transform in the $y$-direction defined by
\begin{equation}\label{eq.deffour}
\calF u(k) := \frac{1}{\sqrt{2\pi}} \int_{\bbR} u(y)\, \rme^{-\rmi k \,y}\, \rmd y  \quad \forall u \in L^1(\bbR)\cap L^2(\bbR),
\end{equation}
which extends to a unitary transformation from $L^2(\bbR_y)$ to $L^2(\bbR_k).$ For functions of both variables $x$ and $y,$ we denote also  by $\calF$ be the partial Fourier transform in the $y$-direction. In particular, for an element $\bU \in \Hxy$, one has
\begin{equation}\label{eq.defH1D}
\calF \bU(\cdot,k) \in \Hx := L^2(\bbR) \times L^2(\bbR)^2 \times L^2(\bbR_+) \times L^2(\bbR_+)^2 \quad \mbox{for a.e. } k \in \bbR,
\end{equation}
where  $\bbR_\pm=\{ x\in \bbR \mid \pm x>0\}$ and  the Hilbert space $\Hx$ is endowed with the inner product $(\cdot\,,\cdot)_{\indHx}$ defined as  $(\cdot\,,\cdot)_{\Hxy}$ in \eqref{eq.innerproduct} except that  $L^2$ inner products are now defined on one-dimensional domains.\\[6pt]
\noindent Applying $\calF$, one can decompose the operator $\bbA$ as a direct  integral of self-adjoint operators $\bbA_{\bk}$:
\begin{equation}\left\{
\begin{array}{l}
\label{eq.AtoAk}
 \displaystyle  \bbA=\calF^{*}\ \bbA^{\oplus}  \calF \quad  \mbox{ where }\quad  \bbA^{\oplus} = \int_{\bbR }^{\oplus} \bbA_{k} \, \rmd k, \quad \mbox{meaning that } \nonumber\\[15pt]
   \forall \; \bU \in \Hxy, \   \calF(\bbA \bU)(\cdot\,,k) = \bbA_k \, \calF \bU(\cdot\,,k)  \ \mbox{for a.e. } k  \in \bbR,
\end{array}\right.
\end{equation}
where $\bbA_k: D(\bbA_k)\subset \Hx \to \Hx$ is deduced from the definition of $\bbA$ by replacing  $y$-derivative by  product by $\rmi \, k$. Namely, one has
\begin{equation}\label{eq.opAk}
\bbA_k := \rmi\ \begin{pmatrix}
0 & \eps_0^{-1}\,\curlk & -\eps_0^{-1} \, \Pi & 0 \\[4pt]
-\mu_0^{-1}\,\bcurlk & 0 & 0 & - \mu_0^{-1} \,\bPi \\[4pt]
\eps_0 \Oe^2 \,\Rop & 0 & 0 &0 \\[4pt]
0 & \mu_0 \Om^2\,\bR & 0 & 0
\end{pmatrix},
\end{equation} 
\begin{equation*}
\bcurlk u := \left(\rmi k u, -\frac{\rmd u}{\rmd x}\right)^{\top}, \quad \ \curlk \bu := \frac{\rmd u_y}{\rmd x}-\rmi k u_x \mbox{ for }   \bu :=(u_x,u_y)^{\top},
\end{equation*} 
and the operators $\Pi$, $\bPi$, $\Rop$ and $\bR$ are defined as in \eqref{TE} but for functions of the variable $x$ only.  It is defined on the dense domain $D(\bbA_k)$ in $\Hx$ given by 
\begin{eqnarray*}
 & &\rmD(\bbA_k) := H^{1}(\bbR) \times \bH_{\curlk}(\bbR) \times  L^2(\bbR_+) \times L^2(\bbR_+)^2,  \ \mbox{ where } \\[4pt]
& &   \bH_{\curlk}(\bbR) := \{\bu\in L^2(\bbR)^2 \mid \curlk \bu \in L^2(\bbR)\} = L^2(\bbR) \times H^1(\bbR).\end{eqnarray*}

\subsubsection{The generalized eigenfunctions}\label{sec-geneigen}

\subsubsection*{A formal approach}

The generalized Fourier transform $\bbF$ is expressed via a family of time-harmonic solutions of the evolution equation \eqref{eq.schro}, referred to as \emph{generalized eigenfunctions} or generalized eigenmodes.   Such  modes are by definition non-zero bounded solutions of the equation 
\begin{equation}\label{eq.genralized}
\bbA \, \mathbb{W}=\omega\, \mathbb{W}  \quad \mbox{ for  non stationary frequency } \omega\in \bR\setminus \{ 0, \pm \, \Om\}.
\end{equation}
This equation has to be understood in distributional sense since these solutions do not  belong to $\Hxy$. As the medium  is stratified, they are expressed as separable functions of the variables $x$ and $y.$ Indeed, they appear as superpositions of planes waves on each side of the interface $x=0$. Thus, one looks of bounded non trivial solutions of  \eqref{eq.genralized}  of the form 
\begin{equation}  \label{eq.genralizedbis} 
\bbW_{k,\omega}(x,y)= \wlk(x) \,  \rme^{\rmi\,  k  \, y}, 
\end{equation}
where the  vector-valued function is of the form
$$\wlk=(e_{k,\omega}, \bh_{k,\omega}, \dot{p}_{k,\omega}, \dot \bm_{k,\omega})^{\top}\in L^{\infty}(\R)\times   L^{\infty}(\R)^2 \times  L^{\infty}(\R_+) \times  L^{\infty}(\R_+)^2.$$
Thus, it is equivalent to find a bounded function $\wlk$ solution in the distributional sense of
\begin{equation}\label{eq.systemeigenx}
\bbA_{k}  \wlk= \omega \wlk .
\end{equation}
After an elimination of the unknowns $\bh_{k,\omega}, \dot{p}_{k,\omega}, \dot \bm_{k,\omega}$ (left to the reader), one expresses $\wlk$ via   $\bbV_{k,\omega}$  a ``vectorizator'' operator:
\begin{equation}\label{eq.def-Vk}
\bbV_{k,\omega} \, w :=  
\left(  w \,,\, -\frac{\rmi}{\mu_{\omega}\, \omega}\,\bcurlk w \,,\, \frac{\rmi\, \eps_0 \,\Oe^2 }{\omega}\,\Rop \,w \,,\, \frac{\mu_0\, \Om^2 }{\mu_{\omega}^+\, \omega^2}\,\bR \,\bcurlk w \right)^{\top},
\end{equation}
which defines each  $\wlk$ in term of its first scalar component $e_{k,\omega}$ (the component associated with the electrical field).
One finally obtains that the system  \eqref{eq.systemeigenx} is equivalent to 
\begin{equation}\label{eq.funcproprgenk}
  \wlk=\bbV_{k,\omega} \,e_{k,\omega}
\end{equation}
with  the  bounded function $e_{k,\omega}$  solution of  the following scalar Sturm-Liouville equation:
\begin{equation}\label{eq.E}
\displaystyle - \frac{\rmd}{\rmd x}\left( \frac{1}{\mu(\omega,\cdot)}\frac{\rmd e_{k,\omega}}{\rmd x}\right)+\frac{ \mathcal{D}_{k,\omega} }{\mu(\omega,\cdot )}\, e_{k,\omega} =0, \  \mbox{ with }   \mathcal{D}_{k,\omega}(x) := k^2- =\varepsilon(\omega,x)\, \mu(\omega, x)\, \omega^2,
\end{equation}
where $\varepsilon(\omega,x)$ and $\mu(\omega, x)$  are defined as their two dimensional version $\varepsilon(\omega,\bx)$ and $\mu(\omega, \bx)$ in  \eqref{eq.defepsmu} by simply  replacing $\bx=(x,y)$ by $x$.
As this scalar Sturm-Liouville equation is taken  in sens of distributions, it contains implicitly  the following transmission conditions:
\begin{equation}\label{eq.transmissionconditionbis}
[e_{k,\omega}]_{x=0}=0 \quad  \mbox{ and } \quad  \Big[  \frac{1}{\mu(\omega,\cdot)}\frac{\rmd e_{k,\omega}}{\rmd x}\Big]_{x=0}=0.
\end{equation}

\noindent Thanks to \eqref{eq.genralizedbis} and \eqref{eq.funcproprgenk}, the dimension of the  space of generalized functions $\bbW_{k,\omega}$ associated to $(k,\omega)$  is  isomorphic to the space of bounded solutions  of \eqref{eq.E} whose dimension  is $0$, $1$ or $2$. This dimension is by definition the spectral  multiplicity of $\omega$ for the operator $\bbA_k$, that we shall  call  for simplicity the spectral multiplicity of $(k,\omega)$, which is the object of  in the next section. 


\subsubsection*{Definition of the spectral zones and spectral multiplicity}
\emph{The spectral zones}  will be constructed as region of the $(k,\omega)$-plane for the spectral  multiplicity of $(k,\omega)$ is constant and positive.
The characterization of these spectral zones is linked to the sign of the piecewise-constant function $\mathcal{D}_{k,\omega}$.  From \eqref{eq.defepsmu}, we have more explicitly
\begin{equation*}\label{eq.defTheta}
\mathcal{D}_{k,\omega}(x) = \left\lbrace
\begin{array}{ll} 
\mathcal{D}_{k,\omega}^{-} := k^2-\eps_0\,\mu_0 \,\omega^2 & \mbox{if }x < 0,\\[4pt]
\mathcal{D}_{k,\omega}^{+} := k^2-\eps^+(\omega) \, \mu^+(\omega) \, \omega^2& \mbox{if }x > 0.
\end{array}
\right.
\end{equation*}
Physically $\mathcal{D}_{k,\omega}^{\pm}$ represents the square of the wavenumber in the $x$-direction inside $\mathbb{R}^2_\pm$, for a plane wave of frequency $\omega$ whose wavenumber in the $y$-direction is $k$. At fixed $(k,\omega)$, the sign of $\mathcal{D}_{k,\omega}^{\pm}$ in each medium  determines if the generalized eigenfunction is an oscillating solution (and thus \emph{propagative}) of \eqref{eq.E}   or an exponential decreasing solution (and thus \emph{evanescent})  in the considered medium. As $\mathcal{D}_{k,\omega}^{\pm} =\mathcal{D}_{|k|,|\omega|}^{\pm}$ for all $(k,\omega) \in \bbR^2$  and $\mu(\omega, \cdot)$ is even in $\omega$, we can restrict ourselves to the quadrant $k \geq 0$ and $\omega \geq 0.$ In this quadrant, there are three curves through which the sign of $\mathcal{D}_{k,\omega}^{-}$ or $\mathcal{D}_{k,\omega}^{+}$ changes. 
More precisely, one has
\begin{equation} \label{defk0DI}
\begin{array}{ll}
\mathcal{D}_{k,\omega}^{-} = 0 & \Longleftrightarrow \quad |k| = k_0 (\omega) :=\sqrt{ \eps_0 \, \mu_0} \, |\omega|, \\ [8pt]
\mathcal{D}_{k,\omega}^{+} = 0 & \Longleftrightarrow \quad |k| = \left\{ \begin{array}{ll}
k_{\scD} (\omega) := \sqrt{\eps^+ (\omega)\, \mu^+(\omega)} \, |\omega|
& \text{if }  |\omega| \geq \max(\Oe,\Om)\\
\quad\text{or} &  \\ 
k_{\scI} (\omega) := \sqrt{\eps^+(\omega) \, \mu^+(\omega)} \, |\omega|
& \text{if }  0 < |\omega| \leq \min(\Oe,\Om).
\end{array} \right.
\end{array}
\end{equation}
The spectral cuts are represented in Figure \ref{fig.speczones1} in the cases $\Oe < \Om$ and $\Oe = \Om$. They delimit different areas in the positive quadrant. More precisely, 
the orange areas represent the parts of the positive quadrant where $\mathcal{D}_{k,\omega}^{-}  < 0$ and $\mathcal{D}_{k,\omega}^{+}  < 0$. It corresponds to a propagative regime along the $x$-direction in both media. Then, the green area corresponds to a propagative regime (again in the $x$-direction) in the vacuum (since $\mathcal{D}_{k,\omega}^{-}  < 0$) and  an evanescent regime (since $\mathcal{D}_{k,\omega}^{+}  > 0$) in the Drude medium. Oppositely, in the blue area, the regime is evanescent in the vacuum (since $\mathcal{D}_{k,\omega}^{-} >0$) and  propagative  in the Drude material (since $\mathcal{D}_{k,\omega}^{+}  < 0$).\\[6pt]
Moreover, as the Drude material is a dispersive negative material, in the  region where it is propagative (i.e. where $\eps^+(\omega)$ and $\mu^+(\omega)$ have the same sign), the propagation can be direct or inverse (see section 6.3 of \cite{Cas-Jol-25} or  section 3.3.2 of \cite{Cas-Haz-Jol-17}). Indeed, in the areas, where it is propagative and both $\eps^+(\omega)$ and $\mu^+(\omega)$ are  positive,  the group and phase velocities of a plane wave have the same direction, as in vacuum and one says that it is \emph{direct propagative}. On the other hand, in the area,  where the Drude material is propagative but where $\eps^+(\omega)$ and $\mu^+(\omega)$ are negative, the propagation is called \emph{inverse}, since the group and phase velocities  point in opposite  directions. This justifies the use of the indices $\scD,$ $\scI$ and $\scE$, meaning respectively \emph{direct}, \emph{inverse} and \emph{evanescent}, to name the various spectral zones. Each of them is actually indexed by a pair of indices: the first one indicates the behavior in the vacuum ($\scD$ or $\scE$) and the second one, in the Drude material ($\scD$, $\scI$ or $\scE$). 
We thus define
\begin{equation*}
\begin{array}{ll}
\zDD & := \ \left\{(k,\omega) \in \bbR^2 \mid \  
|\omega| > \max(\Oe,\Om) \text{ and } |k| < k_\scD(\omega) \, \right\},\\[4pt]
\zDI & := \ \left\{(k,\omega) \in \bbR^2 \mid \ 
0 < |\omega| < \min(\Oe,\Om) \text{ and } |k| < \min\big( k_0(\omega),k_\scI(\omega) \big) \, \right\}, \\[4pt]
\zEI & := \ \left\{(k,\omega) \in \bbR^2 \mid \ 
0 < |\omega| < \min(\Oe,\Om),\ k_0(\omega) < |k| < k_\scI(\omega) \, \right\}, \\[4pt]
\zDE & := \ \left\{(k,\omega) \in \bbR^2 \mid \ 
|\omega| \neq \Om  \text{ and } |k| < k_0(\omega)\, \right\} 
\setminus \overline{\zDD \cup \zDI}.
\end{array}
\end{equation*}
In the following, the above sets will be referred as  {\em surfacic spectral zones}. The parts of these spectral zones located in the quadrant $\bbR^+ \times \bbR^+$ are represented in Figure \ref{fig.speczones1}. 
\begin{figure}[h!]
\centering
 \includegraphics[width=0.47\textwidth]{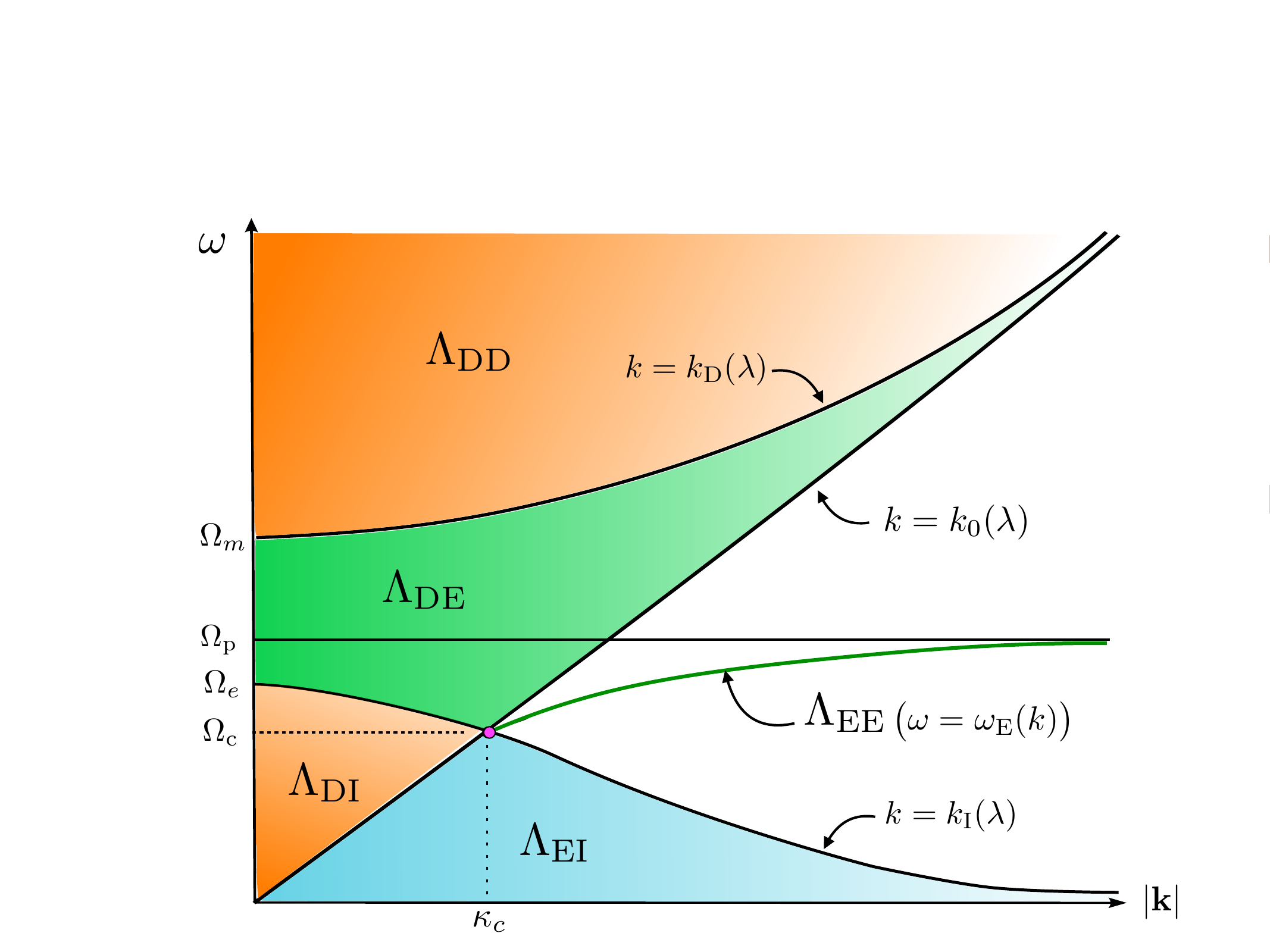}
  \includegraphics[width=0.515\textwidth]{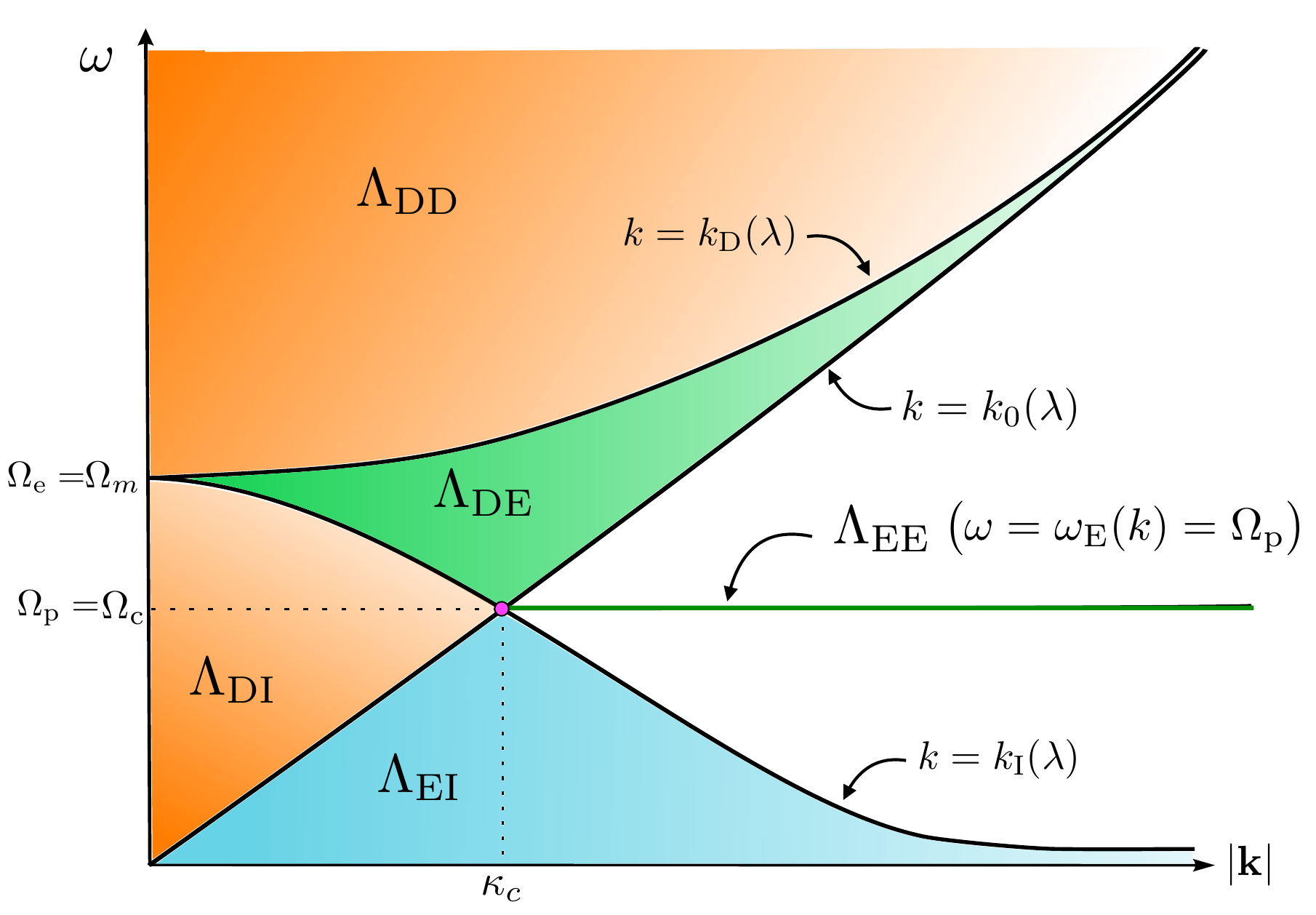}
 \caption{Spectral zones represented on  $\bbR^+ \times \bbR^+$ for $\Oe<\Om$ (left) and  $\Oe=\Om$ (right).}
\label{fig.speczones1}
\end{figure}
The expression of the generalized eigenfunctions given below involves an appropriate square root $\xi_{k,\omega}^\pm$ of $\mathcal{D}_{k,\omega}^{\pm} $  that has the property to be either purely imaginary or positive real (the choice of the square root is justified by a limiting absorption process \cite[\S 3.3.1]{Cas-Haz-Jol-17}). We thus define
\begin{eqnarray}
\xi_{k,\omega}(x) & := &\xi_{k,\omega}^\pm \quad\text{if }\pm \, x > 0 \quad\text{where} \label{eq.defthetax} \\[8pt]
\xi_{k,\omega}^-  & := & \left\lbrace \begin{array}{ll} 
- \rmi \,\sgn(\omega)\,|\mathcal{D}_{k,\omega}^-|^{1/2} & \mbox{if } (k,\omega) \in \zDI \cup \zDE \cup \zDD,\\[4pt]
|\mathcal{D}_{k,\omega}^{-} |^{1/2} & \mbox{otherwise},
\end{array}\right.
\label{eq.def-thetam} \\[8pt]
\xi_{k,\omega}^+ & := & \left\lbrace \begin{array}{ll} 
+ \rmi \,\sgn(\omega)\,|\mathcal{D}_{k,\omega}^{+} |^{1/2} & \mbox{if } (k,\omega) \in \zEI \cup \zDI,\\[4pt]
- \rmi \,\sgn(\omega)\,|\mathcal{D}_{k,\omega}^{+} |^{1/2} & \mbox{if } (k,\omega) \in \zDD,\\[4pt]
|\mathcal{D}_{k,\omega}^+|^{1/2} & \mbox{otherwise}.
\end{array}\right.
\label{eq.def-thetap}
\end{eqnarray} 

\noindent We  introduce now a last spectral zone $\zEE $,  associated to \emph{plasmonic waves}, {\it i.e.}, guided modes that are evanescent in both media (since it is localized in the white area in figure \ref{fig.speczones1} where  $\mathcal{D}_{k,\omega}^{-}  > 0$ and  $\mathcal{D}_{k,\omega}^{+}  > 0$). Thus, these modes are  localized and propagates alongside the interface between both media  \cite{Mai-07}. Unlike the four other spectral zones which are surface areas, $\zEE$ is composed of four curves which originate at the intersection points of the spectral cuts, called here the \emph{cross points}. These are the points where $\mathcal{D}_{k,\omega}^{-} = \mathcal{D}_{k,\omega}^{+} = 0,$ that is, the four points $(k,\omega)$ such that $|k| = \kc$ and $|\omega| = \Oc$, where $\kc =  k_0 (\Oc) = k_{\scI} (\Oc),$ which yields the definition \eqref{eq.def-Op-Oc} of $\Oc$ and $$
\kc = \sqrt{\eps_0 \mu_0} \, \Oc.$$
The spectral zone $\zEE $ is composed of the solutions $(k,\omega)$ of the following \emph{dispersion equation}:
\begin{equation}\label{eq.disp}
\calW_{k,\omega} = 0 
\quad\text{where}\quad
\calW_{k,\omega} := \frac{\xi_{k,\omega}^-}{\mu^{-}(\omega)}+ \frac{\xi_{k,\omega}^+}{\mu^{+}(\omega)} = 0.
\end{equation}
We know from \cite[Lemma 13]{Cas-Haz-Jol-17} that for a given $k$, this equation admits no solution if $|k|<\kc$, and two opposite solutions  $\pm \, \omega_\scE(k)$  if $|k|\geq \kc$, where
\begin{equation}\label{eq.expressionlambdae}
\omega_\scE(k) := 
\left\{ \begin{array}{ll}
\displaystyle \Om \ \sqrt{\frac{1}{2} + \frac{k^2}{K} - \sgn(K)\,\sqrt{\frac{1}{4} + \frac{k^4}{K^2}}} 
& \quad  \text{if } \Om \neq \Oe, \\[10pt]
\displaystyle \Om / \sqrt{2} 
& \quad  \text{if } \Om = \Oe,
\end{array}\right.
\end{equation}
with $K := \eps_0 \mu_0 \,(\Om^2-\Oe^2).$ The function $k \mapsto \omega_\scE(k)$ is strictly decreasing on $[\kc,+\infty)$ if $\Om < \Oe$ and strictly increasing if $\Om > \Oe$. Moreover $\omega_\scE(k) = \Om / \sqrt{2} + O(k^{-2})$ as $|k| \to +\infty.$ In the case where $\Om \neq \Oe,$ we denote by $\kE$  the inverse of $\omega_\scE$, originally defined for positive $\omega$ and $k$ and extended to negative $\omega$ by setting $\kE(-\omega)=\kE(\omega)$, that is,
\begin{equation} \label{defkE}
|\omega| = \omega_\scE(k) \quad \Longleftrightarrow \quad |k| = \kE (\omega) 
\quad\text{if } |k| \in [\kc, +\infty) \text{ and } |\omega| \in \omega_\scE\big([\kc, +\infty)\big),
\end{equation}
where $\omega_\scE\big((\kc, +\infty)\big) = 
\big(\min(\Op,\Oc),\max(\Op,\Oc)\,\big)$. We finally define 
\begin{equation*}\label{eq.defZEE}
\zEE := \left\{ (k,\omega)\in \bbR^2 \mid |k| > \kc \text{ and } |\omega| = \omega_\scE(k) \right\}.
\end{equation*}
$\zEE$ is the union of four curves, so we name it \emph{the lineic spectral zone}.
We have excluded the cross points from its definition, although they are  solutions to \eqref{eq.disp}. Thus, $\overline{\zEE}$ yields all the solutions to \eqref{eq.disp}. Figure \ref{fig.speczones1} shows the location of $\zEE$  when $\Oe<\Om$, $\Oe=\Om$ and $\Oe>\Om$.

The generalized eigenfunctions functions will be indeed denoted  here  by  $\bbW_{k,\omega,j}$ and indexed by three variables. The two first ones, already presented, are real-valued: $k$ is  the wavenumber in the $y$-direction and $\omega$ is a spectral parameter. The last one $j$ is an integer is related to the multiplicity of $(k,\omega) \in \zZ$ with $\scZ \in \{ \DD, \DE, \DI, \EI\}$.
Summing up, the set $J_\scZ$ of possible values of $j$ when $(k,\omega) \in \zZ$ with $\scZ \in \{ \DD, \DE, \DI, \EI\}$ is
\begin{equation}\label{eq.def-Jz}
J_\scZ :=
\left\lbrace\begin{array}{ll}
\{ -1,+1 \} & \mbox{ if } \scZ = \DD \mbox{ or }\DI, \\
\{ +1 \}    & \mbox{ if } \scZ = \DE,\\
\{ -1 \}    & \mbox{ if } \scZ = \EI, \\
\{ 0 \}    & \mbox{ if } \scZ = \EE,
\end{array}
\right.
\end{equation}
and 
$
 m_{\scZ}=\operatorname{card}\, J_\scZ  \in \{1,2\} \mbox{ is the constant multpilicity of } (k,\omega)  \mbox{ in  } \scZ \in \{ \DD, \DE, \DI, \EI\}.
 $
\\[6pt]
\noindent Before giving the expression of the generalized eigenfunctions $\bbW_{k,\omega,j}$, let us discuss their physical interpretation, which make clear our choice of possible values for the index $j$. Consider first the case of the surface zones, that is, $\scZ \in \calZ \setminus \{ \EE \}$. Each $\bbW_{k,\omega,j}$ represents here an incident plane wave which scatters on the interface between both media and produces a reflected plane wave and a transmitted wave. In the half-plane where both incident and reflected waves coexist, the regime of vibration is necessarily propagative (direct or inverse) in the $x$-direction. Whereas in the half-plane where the transmitted wave occurs, the regime can be propagative or evanescent. This explains that for a given pair $(k,\omega)$ in the spectral zones $\zDD$ and $\zDI$ where both half-planes are propagative, two generalized eigenfunctions $\bbW_{k,\omega,j}$ are considered: they are indexed by $j=\pm \, 1$ which indicates the half-plane $\bbR^2_\pm$ where the transmitted wave takes place. Following the same interpretation, for a given pair $(k,\omega)$ in the spectral zones $\zEI$ and $\zDE$, only one $\bbW_{k,\omega,j}$ is considered, with $j=-1$ in $\zEI$ and $j=+1$ in $\zDE$. On the other hand, for the one-dimensional spectral zone $\zEE$, the regime is evanescent in both media. For a given pair $(k,\omega) \in \zEE,$ only one $\bbW_{k,\omega,j}$ which  represents now a guided wave that propagates along the interface is considered. Since there is no longer transmitted wave, we use here  the index $j=0$.

\subsubsection*{Expression of the generalized eigenfunctions}
We  introduce now the generalized eigenfunctions $\bbW_{k,\omega,j}$ related to the spectral zones $\zZ$ for
\begin{equation*}
\scZ \in \calZ := \{ \DD,\DE,\DI,\EI,\EE \},
\end{equation*}
defined by
\begin{equation}\label{eq.def-W}
\forall \scZ \in \calZ,\ \forall (k,\omega) \in \zZ,\ \forall j\in J_\scZ,\quad 
\bbW_{k,\omega,j} = \wlkj \,  \rme^{\rmi\,  k  \, y} \  \mbox{ with }  \ \wlkj:= \bbV_{k,\omega}\ e_{k,\omega,j},
\end{equation}
where $\bbV_{k,\omega}$ is the  ``vectorizator''  operator defined by \eqref{eq.def-Vk} and  $e_{k,\omega,j}$ is  the scalar function:
\begin{equation}\label{eq.def-w}
e_{k,\omega,j}(x) := A_{k,\omega,j}\ \psi_{k,\omega,j}(x)  \quad \forall x \in \bbR,
\end{equation}
where the expressions of $A_{k,\omega,j}$ and $\psi_{k,\omega,j}(x)$ depend on the spectral zones.\\[6pt] 
	
\noindent In the surface spectral zones $\zDD$, $\zDE$, $\zDI$ and $\zEI$,  the coefficient   $A_{k,\omega,j}$ is given by
\begin{align}
A_{k,\omega,\pm 1} & := \frac{1}{\pi\,\left| \calW_{k,\omega} \right|} \,
\left| \frac{\omega}{2}\, \xi_{k,\omega}^\mp / \mu_\omega^\mp \right|^{1/2} 
\quad\text{and} \label{def-A-gen} 
\end{align}
where $\calW_{k,\omega}$ and $\xi_{k,\omega}(x)$ are defined respectively in \eqref{eq.disp} and \eqref{eq.defthetax} and the function
\begin{equation}\label{eq-def-phi}
\psi_{k,\omega,\pm 1}(x) = 
\left\{ \begin{array}{ll}
\displaystyle \cosh\big( \xi_{k,\omega}^\mp\, x \big) \mp \frac{ \xi_{k,\omega}^\pm /\mu^\pm(\omega) }{ \xi_{k,\omega}^\mp / \mu^\mp(\omega)}\ \sinh\big(\xi_{k,\omega}^\mp\, x \big) & \text{if } \pm \, x \leq 0, \\
\exp\big( \mp \xi_{k,\omega}^\pm\, x \big) & \text{if } \pm \, x \geq 0,
\end{array} \right.
\end{equation}
which justifies the above-mentioned physical interpretation of the $\bbW_{k,\omega,j}$. \\[6pt]
\noindent In the plasmonic spectral zone $\zEE,$ we have
\begin{equation}\label{def-A-plasm} 
A_{k,\omega,0}  := \frac{\omega^2\, \left|\mu^{+}(\omega) \, \xi_{k,\omega}^{+}\right|^{1/2}}{\sqrt{2\pi}\,\Om \big( 4k^4 +(\eps_0\mu_0)^2(\Oe^2-\Om^2)^2 \big)^{1/4}} \ \text{ and }  \
\psi_{k,\omega,0}(x) := \exp\big( - \xi_{k,\omega}(x) \, |x| \big),  
\end{equation}
which shows clearly that $\bbW_{k,\omega,0}$ is a guided wave localized near the interface $x=0$.\\[6pt]
\noindent We point out that in each spectral zones  $\Lambda_\scZ $,  for a fixed $(k,\omega)\in \Lambda_\scZ$,  the functions $\psi_{k,\omega,j}$ for $j\in J_{\scZ}$ form a basis of the space of bounded solutions of the Sturm-Liouville  equation \eqref{eq.E}. This basis of functions satisfies the transmission  conditions \eqref{eq.transmissionconditionbis} and is chosen such that 
$$\psi_{k,\omega,j}(0)=1 \quad \mbox{ and } \quad  \Big(\frac{1}{\mu(\omega,\cdot)}\frac{\rmd \psi_{k,\omega,j}}{\rmd x}\Big)(0)=1.
$$

\begin{Rem}[On the choice of the normalization coefficient $A_{k,\omega,j}$]
As the generalized eigenfunction $\bbW_{k,\lambda,j}$ is defined up to a complex coefficient (depending on the spectral parameters $(k,\omega)$ and the index $j$),  the renormalization  coefficient $A_{k,\omega,j}$  given by   \eqref{def-A-gen} and \eqref{def-A-plasm} can not be deduced by the formal approach presented here. However,  the value of  the coefficient  $A_{k,\omega,j}$  is important.  Indeed, as for the diagonalization of  a Hermitian matrix where one normalizes its eigenvectors to construct a unitary matrix which diagonalizes it,  one needs here to find   ``the good weight''  associated to $\bbW_{k,\lambda,j}$ which allows to define in the following paragraph a unitary map: the generalized Fourier transform which diagonalizes the operator $\bbA$. The problem is that except for $A_{k,\omega,0} $ that can be deduced  (up to a $1/\sqrt{2 \pi}$ factor due to the Fourier transform $\mathcal{F}$ in the $y-$direction) by the normalization  of  $\wlkz\in \Hx$,   the other coefficients $A_{k,\omega,j}$  can not be deduced by a normalization process since $\wlkj\notin  \Hx$ (as it is just bounded). \\[4pt]
Indeed, one  uses a rigorous approach (see \cite{Cas-Haz-Jol-17}  for the details) to obtain the complete expression (given by \eqref{eq.def-W}, \eqref{eq.def-w}, \eqref{def-A-gen} and \eqref{def-A-plasm}) of the generalized eigenfunctions $\bbW_{k,\omega,j}$. This approach is based on the spectral theorem and the Stone formula  by computing for   $k\in \bbR$ the spectral measure $\bbE_k(\cdot)$ of the reduced operators  $\bbA_k$ via the  the limit of the imaginary part of the resolvent $R_k(\omega)=(\bbA_k-\omega \mathrm{I})^{-1}$  when  it approaches the real axis from above. The spectral measure $\bbE_k(\cdot)$  is defined in terms of the generalized eigenfunctions $\wlkj$ (involved in \eqref{eq.def-W})  and is the key tool to construct via the spectral theorem  the generalized Fourier transform. 
\end{Rem}

\subsubsection{The diagonalization Theorem}\label{sec-diag-th}
We introduce now the spectral space 
\begin{equation*}
\hatH := \bigoplus \limits_{\scZ \in \calZ} L^{2}(\zZ)^{\operatorname{card}(J_{\scZ})}=L^2(\zDD)^2 \oplus L^2(\zDE)\oplus L^2(\zDI)^2 \oplus  L^2(\zEI)\oplus   L^2(\zEE),
\label{eq.ident-hatH}
\end{equation*}
in which the action of $\bbA$ reduces to a simple multiplication by the spectral variable $\omega$. $\hatH$ is a direct sum of $L^2$ spaces of each spectral zone.  More precisely, each $L^2(\zZ)$ for $\scZ \in \calZ$  is repeated $\operatorname{card}(J_{\scZ})$ times, that is, the number of generalized eigenfunctions associated to the spectral zone $\zZ$. As for the $\bbW_{k,\omega,j}$'s, we denote by $\bhatU(k,\omega,j)$ the fields of $\hatH$, where it is understood that the set $J_{\scZ}$ of possible values for $j$ depends on the spectral zone $\zZ$ to which the pair $(k,\omega)$ belongs. Using these notations, the Hilbert space $\hatH$ is endowed with the following norm:
$$
\| \bhatU \|_{\hatH}^2 := \sum_{\scZ \in \calZ\setminus\{\EE\}}\sum_{j\in J_\scZ} \int_{ \zZ} |\bhatU(k,\omega,j) |^2 \,\rmd \omega \,\rmd k+  \sum_{\pm }\int_{|k|>k_{\rm c}}| \bhatU(k,\pm \omega_\scE(k),0) |^2  \,\rmd k.
$$

Theorem \ref{th.diagA} below gathers the results of Theorem 20 and Proposition 21 in \cite{Cas-Haz-Jol-17}. It defines the generalized Fourier transform $\bbF$ and its adjoint $\bbF^{*}$.  $\bbF$ is  a ``decomposition'' operator on the family of generalized eigenfunctions $(\bbW_{k,\omega,j})$, whereas  $\bbF^*$ is a ``recomposition'' operator in the sense that its ``recomposes'' a function $\bU\in \Hxy$ from its spectral components $\bhatU(k,\omega,j)\in \hatH$ which appear as ``coordinates'' on the ``generalized spectral basis'' $(\bbW_{k,\omega,j})$. Both operators are (partial) isometries and thus bounded, so it is sufficient to know   their expression  on the dense subspaces $\Hps$ (with $s>1/2$) of $\mathcal{H}$  for $\bbF $ and $\hatH_{\rm comp}$ of $\hatH$ introduced below for $\bbF^*$.

\begin{Thm}[Diagonalization Theorem, cf. \cite{Cas-Haz-Jol-17}]\label{th.diagA}
	Let $s > 1/2.$ 
	
	{\rm (i)} The generalized Fourier transform $\bbF: \Hxy \mapsto \hatH$ is a partial isometry, defined   by  
	\begin{equation}\label{eq.Four-gen}
	\forall \bU \in \Hps, \ \forall \scZ \in \calZ,\ \forall (k,\omega) \in \zZ,\ \forall j\in J_{\scZ}, \quad
	\bbF\bU(k,\omega,j) =\langle  \bU, \bbW_{k,\omega,j}  \rangle_{s},   
	\end{equation}
	where the $\bbW_{k,\omega,j}$'s are defined in {\rm (\ref{eq.def-W})}.
	
	{\rm (ii)} Let $\hatH_{\rm comp}$ be the dense subspace of $\hatH$ composed of compactly supported functions whose supports do not intersect the boundaries of $\zZ$ for $\scZ \in \calZ \setminus \{ \EE\}$ (\textit{i.e.}, the spectral  cuts and the three lines $\bbR \times \{0,\pm\Om \}$). Then $\bbF^*: \hatH \mapsto \Hxy$  of $\bbF$ is an isometry defined for all $\bhatU \in \hatH_{\rm comp}$ by
	\begin{equation}\label{eq.adj-Four-gen}
	\bbF^{*}\bhatU = \sum_{\scZ \in \calZ \setminus \{\EE \} } \sum_{j\in J_{\scZ}} \int_{\zZ } \bhatU(k,\omega,j)\, \bbW_{k,\omega, j} \,\rmd \omega\, \rmd k 
	+ \sum_{\pm} \int_{|k|>k_{\rm c} } \bhatU(k,\pm \omega_\scE(k) ,0)\, \bbW_{ k,\pm \omega_\scE(k) ,0} \, \rmd k,
	\end{equation}
	where the integrals are understood as Bochner integrals with values in $\Hms$. 	
	
	{\rm (iii)} Moreover, we have $ \bbF\,\bbF^{*} =\mathrm{Id}_{\hatH},
	$ while $ \bbF^{*} \bbF\,=\bbP_{\rm div0}$ where  $\bbP_{\rm div0}$ is the orthogonal projector in $\Hxy$ onto $\Hxydiv$ (see \eqref{eq.Hxydiv}).
	Thus, the restriction of $\bbF$ to $\Hxydiv$ is a unitary operator. Furthermore $\bbF$ diagonalizes $\bbA$ in the sense that for any measurable function $f:\bbR \to \bbC$,
	\begin{equation}\label{eq.diagA}
	f(\bbA)\bbP_{\rm div0}=\bbP_{\rm div0} f(\bbA)=\bbF^{*}\,f(\omega)\,\bbF \ \mbox{ in } \rmD(f(\bbA)).
	\end{equation}
\end{Thm}

\begin{Rem}\label{rem.F}
{\rm (i)}  To define $\bbF$ in \eqref{eq.Four-gen}, we use  the duality product $\langle \cdot,\cdot \rangle_{s}$ (which extends the inner product of $\Hxy$, see \eqref{eq.innerproductbis}) because  $\bbW_{k,\omega,j}\notin  \Hxy$ (this is why it is called a \emph{generalized eigenfunctions})  since their norm does not decay at infinity. But as $\bbW_{k,\omega,j}$ is bounded, one has $\bbW_{k,\omega,j}\in \Hms$ for $s > 1/2$.

{\rm (ii)} In \eqref{eq.adj-Four-gen}, we restrict ourselves to functions of $\hatH_{\rm comp}$ since one easily checks that the $\Hms$-norm of $\bbW_{k,\omega,j}$ remains uniformly bounded if $(k,\omega)$ is restricted to vary in a compact set of $\bbR^2$ that does not intersect the boundaries of the spectral zones. Hence, for $\bhatU\in \hatH_{\rm comp}$, the integrals considered in \eqref{eq.adj-Four-gen}, whose integrands are valued in $\Hms$, are Bochner integrals {\rm \cite{Hil-96}} in $\Hms$. However, as $\bbF^{*}$ is bounded from $\hatH$ to $\Hxy,$ the values of these integrals belongs to $\Hxy$. 
	
	{\rm (iii)}
	 If $\bhatU$ does not vanish near some part of the boundaries of the spectral zones, because of the singular behavior of some $\bbW_{k,\omega,j}$, the integrals in \eqref{eq.adj-Four-gen} can no longer be Bochner integrals in $\Hms$, but limits of Bochner integrals. Indeed, thanks to the density of $\hatH_{\rm comp}$ in $\hatH$, we can approximate $\bhatU$ by its restrictions to an increasing sequence of compact subsets of $\cup_{\scZ \in \calZ} \zZ$ as in the definition of $\hatH_{\rm comp}$, which yields an approximation of $\bbF^{*}\bhatU$. Of course, the limit we obtain belongs to $\Hxy$ and does not depend on the sequence. We indicate here this limiting process before each integral as follows: for all $\bhatU \in \hatH$,
	\begin{eqnarray}\label{eq.adj-F-lim}
	\bbF^{*}\bhatU= && \sum_{\scZ \in \calZ \setminus \{\EE \} } \sum_{j\in J_{\scZ}} \lim_{\Hxy}\int_{\zZ } \bhatU(k,\omega,j)\, \bbW_{k,\omega, j} \,\rmd \omega\, \rmd k \nonumber \\
	&& + \sum_{\pm} \lim_{\Hxy}\int_{|k|>k_{\rm c} } \bhatU(k,\pm \, \omega_\scE(k) ,0)\, \bbW_{ k,\pm \, \omega_\scE(k) ,0} \, \rmd k.
	\end{eqnarray}
\end{Rem}
\subsubsection{Construction of $\bbM_{\omega}$} \label{sec-spec-density}
 {\bf The non-critical case: $\Oe \neq \Om$.}
The orthogonal projection $\bbP_{\rm ac} $ coincides here with $\bbP_{\rm div0}$ (see \eqref{eq.def-Pac}). To prove \eqref{eq.specdensity}, we apply  the diagonalization Theorem \ref{th.diagA}  to the spectral measure of $\bbA:$ for any Borel set $S \subset \bbR,$ we have $\bbE(S) =\mathds{1}_{S}(\bbA)$ where $\mathds{1}_{S}$ denotes the indicator function of $S$. We assume here that  
$S$  is  bounded and $\overline{S} \cap \sigma_{\rm exc} = \varnothing$
where  $\sigma_{\rm exc} := \{0,\pm \, \Op,\pm \, \Om\}$  for $\Oe \neq \Om$ (see \eqref{eq.def-sigma-exc}). Thus, we exclude the eigenvalues $0$ and $\pm\Om,$ which implies that
\begin{equation}
\bbE(S) =\bbE(S)\,\bbP_{\rm div0} =  \bbE(S)\,\bbP_{\rm ac} 
\label{eq.Eac-non-crit}
\end{equation}
(since $\bbE(S)\,\bbP_{\rm div0} = \bbE(S)\,\bbE\big(\bbR \setminus \{ 0, \pm \Om \}\big) = \bbE\big(S \cap (\bbR \setminus \{ 0, \pm \, \Om \})\big) = \bbE(S)$), and also the plasmonic frequencies $\pm \, \Op$. Applying \eqref{eq.diagA} to $\mathds{1}_{S}(\bbA)$ then yields
\begin{equation*}
\bbE(S)\,\bbP_{\rm div0} = \bbF^{*}\,\mathds{1}_{S}(\omega)\,\bbF.
\end{equation*}
Using the expressions \eqref{eq.Four-gen} and \eqref{eq.adj-F-lim} of $\bbF$ and $\bbF^{*},$ this formula writes more explicitly as
\begin{equation}
\label{eq.spect-rep-E}
\begin{array}{lll}
\bbE(S)\,\bbP_{\rm div0} \, \bU &= &\displaystyle \sum_{\scZ \in \calZ\setminus\{\EE\}} \sum_{j \in J_{\scZ}} \lim_{\Hxy} \int_{\zZ }\mathds{1}_{S}(\omega) \, \langle  \bU, \bbW_{k,\omega,j} \rangle_{s}  \; \bbW_{k,\omega,j} \,\rmd\omega\,\rmd k 
\\[18pt]
&+& \displaystyle \sum_{\pm } \lim_{\Hxy} \int_{|k|>k_{\rm c}}\mathds{1}_{S}(\pm \omega_\scE(k))\, \langle  \bU,  \bbW_{k,\pm \omega_\scE(k),0}  \rangle_{s} \,   \bbW_{k,\pm \omega_\scE(k),0} \,\rmd k,
\end{array}
\end{equation}
for all $\bU \in \Hps$, where we recall that the limit (in $\Hxy$) is obtained by considering an increasing sequence of compact subsets of each $\zZ$ whose union covers $\zZ$. Indeed, using our assumptions on $S$ and a $\Hms$-estimate on the functions $\bbW_{k,\omega,j}$,  one can show that such a limiting process is useless and that one can apply the Fubini's theorem for the surface integrals on  $\zZ$ for $\scZ \in \calZ\setminus\{\EE\}$, as well as the change of variable $k = \pm k_\scE(\omega)$ in the last integral. Admitting this (see  section 3  of \cite{Cas-Haz-Jol-22} for the justification) and using \eqref{eq.Eac-non-crit}, we obtain that for all $\bU \in \Hps$:
\begin{equation}
\label{eq.expr-Eac-non-crit}
\bbE(S)\,\bbP_{\rm ac} \, \bU = \bbE(S) \, \bU = \int_\bbR\mathds{1}_{S}(\omega) \, \bbM_\omega \bU \,\rmd\omega
\quad\text{with}
\end{equation}
\begin{equation}
\label{eq.density-non-crit}
\bbM_\omega \bU := \sum_{\scZ \in \calZ\setminus\{\EE\}} \sum_{j \in J_{\scZ}} \int_{\zZ(\omega)} \langle \bU,\bbW_{k,\omega,j} \rangle_{s}  \; \bbW_{k,\omega,j} \,\rmd k 
+ \sum_{k \in \zEE(\omega)} \JacE(\omega)\ \langle  \bU,  \bbW_{k,\omega,0}  \rangle_{s} \,   \bbW_{k,\omega,0},
\end{equation}
for almost every $\omega \in \bbR,$ where $\JacE(\omega)$ is the Jacobian of the change of variable $k= \pm k_\scE(\omega)$: 
\begin{equation} \label{eq.jacob}
\JacE(\omega) := \big| \kE'(\omega)\big| = \Big|\frac{\rmd \omega_\scE }{\rmd k} \big(\kE(\omega)\big)\Big|^{-1} 
\end{equation}
and $\zZ(\omega)$ is the set of $k \in \bbR$ corresponding to the horizontal section of $\zZ$ at the ``height'' $\omega:$ 
\begin{equation}
\zZ(\omega) := \left\{ k \in \bbR \mid (k,\omega) \in \zZ \right\}.
\label{eq.def-section}
\end{equation}
Figure \ref{fig.speczones1} clearly shows that if $\scZ \in \calZ\setminus\{\EE\}$,  then $\zZ(\omega)$ is either empty (in this case the corresponding integral vanishes) or is a bounded set composed of one or two intervals. For instance, if $\omega > \max(\Oe,\Om),$ then 
$\zDD(\omega) = \big( -k_\scD(\omega),+k_\scD(\omega) \big)$. 
Moreover, we have
\begin{equation*}
\zEE(\omega) = \left\{
\begin{array}{ll}
\{ \pm\, k_\scE(\omega) \} & \text{if } |\omega| \in \omega_\scE\big((\kc, +\infty)\big) = \big(\min(\Op,\Oc),\max(\Op,\Oc)\big),  \\[5pt]
\varnothing & \text{otherwise,}
\end{array}
\right.
\end{equation*}
which shows that the last term in \eqref{eq.density-non-crit} appears only if $|\omega| \in \omega_\scE\big((\kc, +\infty)\big).$ \\

\noindent \textbf{The critical case: $\Oe = \Om$}. We keep  the same assumption for $S$, but now $\sigma_{\rm exc} := \{0,\pm \,\Om\}$ (see \eqref{eq.def-sigma-exc}), so that \eqref{eq.Eac-non-crit} is no longer true. From \eqref{eq.def-Pac}, it has to be replaced by
\begin{equation}
\bbE(S)\,\bbP_{\rm div0} = \bbE(S) = \bbE(S)\,\bbP_{\rm ac} + \bbE\big(S)\,\bbP_{-\Op} + \bbE\big(S)\,\bbP_{+\Op}.
\label{eq.Eac-crit}
\end{equation}
Formula \eqref{eq.spect-rep-E} is still valid. The difference with the non-critical case lies in the last term. 
Since $\omega_\scE(k) = \Op$ for all $|k| > k_{\rm c}$,  it  represents the quantities $\bbE\big(S)\,\bbP_{\pm \,\Op}\, \bU$ related to the eigenvalues $\pm \,\Op$ of infinite multiplicity. 
Thus,  formula \eqref{eq.Eac-crit}  shows that one has  to be subtract it form  \eqref{eq.spect-rep-E} to express $\bbE(S)\,\bbP_{\rm ac}.$ Using the same arguments as above for the surface integrals on $\zZ$ for $\scZ \in \calZ\setminus\{\EE\}$, we obtain instead of \eqref{eq.expr-Eac-non-crit}-\eqref{eq.density-non-crit}
\begin{equation}
\label{eq.expr-Eac-crit}
\bbE(S)\,\bbP_{\rm ac} \, \bU = \bbE(S\setminus\{\pm \, \Op\}) \, \bU = \int_\bbR\mathds{1}_{S}(\omega) \, \bbM_\omega \bU \,\rmd\omega
\quad\text{with}
\end{equation}
\begin{equation}
\label{eq.density-crit}
\bbM_\omega \bU := \sum_{\scZ \in \calZ\setminus\{\EE\}} \sum_{j \in J_{\scZ}} \int_{\zZ(\omega)} \langle \bU,\bbW_{k,\omega,j} \rangle_{s}  \; \bbW_{k,\omega,j} \,\rmd k.
\end{equation}

\noindent {\bf Extension of Theorem \ref{Thm.specdensity} to the case of any bounded function $f$.}\\[4Pt]
We  extend Theorem \ref{th.Holder-dens-spec} to the case of bounded function $f:\bbR\mapsto \bbC$ whose support $S$ is no longer compact and/or contains points of $\sigma_{\rm exc}$. However, the integral representation is not in general a Bochner integral in $B(\Hps,\Hms)$.  In this case, the expression of $f(\bbA)\,\bbP_{\rm ac}$ follows from Theorem  \ref{Thm.specdensity} by considering an increasing sequence $(S_n)$ of compacts subsets of $S \setminus \sigma_{\rm exc}$ whose union covers this set. Setting $f_n := f\,\mathds{1}_{S_n},$ Theorem \ref{th.diagA}  and the Lebesgue's  dominated convergence theorem show that
\begin{equation*}
\big\| \big(f(\bbA) - f_n(\bbA)\big)\,\bbP_{\rm ac}\bU \big\|_{\Hxy} =
\left\{ \begin{array}{ll}
\big\| \big(f(\omega) - f_n(\omega)\big)\,\bbF\bU \big\|_{\hatH} & \text{if }\Oe \neq \Om, \\[8pt]
\big\| \big(f(\omega)\mathds{1}_{\bbR\setminus\{\pm \,\Op\}} - f_n(\omega)\big)\,\bbF\bU \big\|_{\hatH} & \text{if }\Oe = \Om,
\end{array}
\right.
\end{equation*}
tends to $0$. Hence, with the same notations of \eqref{eq.adj-F-lim}, we have 
\begin{equation*}
\forall \bU \in \Hps, \quad
f(\bbA)\,\bbP_{\rm ac}\bU = \lim_{\Hxy} \int_{\bbR} f(\omega) \, \bbM_{\omega}\bU \,\rmd \omega,
\end{equation*}
that we rewrite in the condensed form (see \cite{Conway-19}, p. 256)
\begin{equation}
f(\bbA)\,\bbP_{\rm ac} = \slim_{B(\Hps,\Hxy)}\ \int_{\bbR} f(\omega) \, \bbM_{\omega} \,\rmd \omega,
\label{eq.calc-fonct-ac-lim}
\end{equation}
where ``$\slim$'' means that the limit is taken for the strong operator topology of $B(\Hps,\Hxy)$ .

%
%

\subsection{Proof of the  limiting absorption and  limiting amplitude principles}\label{sec-proofs}
The proofs of the  limiting absorption and  limiting amplitude principles: Theorems  \ref{thm.limabs} and \ref{th.ampllim} are based on the proofs of these results  done  in \cite{Cas-Haz-Jol-22} in sections 4.1 and 4.2. We recall them here.
\subsubsection*{Proof of Theorem \ref{thm.limabs}}
\textbf{The non-critical case: $\Omega_e\neq \Omega_m$}. 
We study the limit of $R_{\rm ac}(\zeta)=R(\zeta) \bbP_{\rm ac}$ when $\zeta\in \bbC\setminus\bbR\to \omega_{\mathrm{s}} \in \bbR\setminus\sigma_{\rm exc}$. By the functional calculus  formula \eqref{eq.diagA}, this limit involves the singularity of  $\omega \mapsto (\omega-\zeta)^{-1}$ at $\omega=\omega_{\mathrm{s}}$. To isolate the role of this singularity, we choose some $\rho > 0$ small enough so that the interval $J := [\omega_{\mathrm{s}}-\rho,\omega_{\mathrm{s}}+\rho]\cap \sigma_{\rm exc} $ and we decompose the latter function as
\begin{equation}
\frac{1}{\omega-\zeta} = f_\zeta^{\rm sin}(\omega) + f_\zeta^{\rm reg}(\omega)
\quad\text{where}\quad
f_\zeta^{\rm sin}(\omega) := \frac{\mathds{1}_{J}(\omega)}{\omega-\zeta}
\text{ and }
f_\zeta^{\rm reg}(\omega) := \frac{\mathds{1}_{\bbR\setminus J}(\omega)}{\omega-\zeta}.
\label{eq.def-f-sin-reg}
\end{equation}
This leads  to split $R_{\rm ac}$ into a ``singular part'' and a ``regular part'' via the formula \eqref{eq.diagA}:
\begin{equation*}
R_{\rm ac}(\zeta) = \bbF^{*}\,\frac{1}{\omega-\zeta}\ \bbF= R_{\rm sin}(\zeta) + R_{\rm reg}(\zeta)
\quad\text{where}\quad
\left\{\begin{array}{l}
R_{\rm sin}(\zeta) := \bbF^{*}\,f_\zeta^{\rm sin}(\omega)\,\bbF, \\[4pt]
R_{\rm reg}(\zeta) := \bbF^{*}\,f_\zeta^{\rm reg}(\omega)\,\bbF.
\end{array}\right.
\end{equation*}

On the one hand, the family of functions $\zeta \mapsto f_\zeta^{\rm reg}(\cdot)$ is differentiable in $\zeta$ in a vicinity of $\omega_{\mathrm{s}}$ uniformly with respect to $\omega \in \bbR$. Hence, in this vicinity, the operator of multiplication by $f_\zeta^{\rm reg}(\cdot)$ is a holomorphic function of $\zeta$ for the operator norm of $B(\hatH)$. As $\bbF$ and $\bbF^*$ are bounded, $\zeta \mapsto R_{\rm reg}(\zeta)$ is also holomorphic in this vicinity for the operator norm of $B(\Hxy)$, thus a fortiori for the one of $B(\Hps,\Hms)$  for $s>1/2$. Its limit value at $\omega_{\mathrm{s}}$ is simply given by 
\begin{equation}
R_{\rm reg}(\omega) = \bbF^{*}\,f_\omega^{\rm reg}(\omega)\,\bbF 
= \slim_{B(\Hps,\Hxy)} \ \int_{\bbR\setminus J} \frac{\bbM_{\omega}}{\omega-\omega_{\mathrm{s}}} \,\rmd \omega,
\label{eq.lim-res-reg}
\end{equation}
where the last equality is obtained via the formula \eqref{eq.calc-fonct-ac-lim} applied to $f=f_\omega^{\rm reg}$.\\

On the other hand, the ``singular part'' $R_{\rm sin}(\zeta)$ is no longer continuous on the real axis. We denote by $R_{\rm sin}^\pm$ the restrictions of $\zeta \mapsto R_{\rm sin}(\zeta)$ to $\bbC^\pm := \{ \zeta\in\bbC \mid \pm \, \Imag\zeta > 0 \},$ i.e.,
\begin{equation*}
\forall \zeta \in \bbC^\pm, \quad
R_{\rm sin}^\pm(\zeta) = \bbF^{*}\,f_\zeta^{\rm sin}(\omega)\,\bbF 
= \int_{J} \frac{\bbM_{\omega}}{\omega-\zeta} \,\rmd \omega.
\end{equation*}
Note that the ``$\slim$'' symbol is removed here since $\omega \mapsto f_\zeta^{\rm sin}(\omega)$ is bounded and compactly supported in $\bbR \setminus \sigma_{\rm exc}$, so that Theorem \ref{Thm.specdensity} applies: the  latter integral is a Bochner integral valued in $B(\Hms,\Hms)$.  Then, by  Theorem \ref{th.Holder-dens-spec}, as the spectral density $\omega \mapsto \bbM_{\omega}$ is locally H\"{o}lder continuous,   one can use the Sokhotski--Plemelj formula \cite[theorem 14.1.c, p. 94]{Hen-86} which ensures the existence of the one-sided limits of $R_{\rm sin}^\pm(\zeta)$ when $\bbC^{\pm} \ni \zeta \to \omega $ for the norm of $B(\Hps,\Hms)$ for $s>1/2$. This formula gives also an explicit expression of these limits:
\begin{equation}
R_{\rm sin}^\pm(\omega) 
= \dashint_{J} \frac{\bbM_{\omega}}{\omega-\omega_{\mathrm{s}}}\, \rmd\omega\  \pm\ \rmi \pi \, \bbM_{\omega_{\mathrm{s}}}  \ \in B(\Hps,\Hms),
\label{eq.lim-res-sin}
\end{equation}
where the dashed integral denotes a Cauchy principal value at $\omega=\omega_{\mathrm{s}}$. Moreover (see \cite{Hen-86,McL-88}), the local H\"{o}lder regularity of the spectral density $\omega \mapsto \bbM_{\omega}$ also ensures that $\zeta \mapsto R_{\rm sin}^\pm(\zeta)$ is locally H\"{o}lder continuous on $\overline{\bbC^{\pm}} \setminus \sigma_{\rm exc}$ for the operator norm of $B(\Hps,\Hms)$, with the same H\"{o}lder exponents $\gamma \in (0,1)$ as those of $\omega \mapsto \bbM_{\omega}$. We point out that even if $\omega \mapsto \bbM_{\omega}$ was locally Lipschitz continuous (i.e., $\gamma = 1$), the Sokhotski--Plemelj theorem would not ensure that so is $\zeta \mapsto R_{\rm sin}^\pm(\zeta)$. This explains why the  value $\gamma = 1$ has not been considered in Theorem \ref{th.Holder-dens-spec}.\\
\noindent Furthermore Combining \eqref{eq.lim-res-reg} and \eqref{eq.lim-res-sin} yields the following spectral representation of $R^{\pm}_{\rm ac}$:

\begin{equation}
R^{\pm}_{\rm ac}(\omega) = \slim_{B(\Hps,\Hms)}\ 
\dashint_{\mathbb{R}} \frac{\bbM_{\omega}}{\omega-\omega_{\mathrm{s}}}\, \rmd \omega\  \pm \, \rmi\pi\, \bbM_{\omega_{\mathrm{s}}}.
\label{eq.lim-res}
\end{equation}

\begin{Rem}\label{rem-Principalvalue}
The ``$\slim$'' symbol and the dashed integral  in  \eqref{eq.lim-res} involve two limit processes that can be considered independently by isolating a vicinity $J$ of $\omega$, exactly as we did above. The principal value denotes by a dashed integral means that we remove from $J$ a symmetric neighborhood of $\omega$, i.e.,  by considering $J_\delta := J \setminus (\omega-\delta,\omega+\delta)$, and take the limit of the integral on $J_\delta$ as $\delta \searrow 0$ in the operator norm of $B(\Hps,\Hms)$.  For the ``$\slim$'', one introduces an increasing sequence of compact subsets of $\bbR \setminus (\sigma_{\rm exc}\cup J)$ whose union covers this set, and take the limit of the integral on these compacts subsets for the strong operator topology of $B(\Hps,\Hms)$. In both limit processes, the integrals on compact sets are Bochner integrals valued in $B(\Hps,\Hms)$.

We point out that to gather the terms \eqref{eq.lim-res-reg} and \eqref{eq.lim-res-sin} to obtain \eqref{eq.lim-res}, we replace  the $\slim_{B(\Hps,\Hxy)}$ by the   $\slim_{B(\Hps,\Hms)}$  to ensure the existence of the principal value in \eqref{eq.lim-res-sin} (this is justified since the existence of the  $\slim_{B(\Hps,\Hxy)}$ in \eqref{eq.lim-res-reg} implies  a fortiori the existence of the $\slim_{B(\Hps,\Hms)}$ of this term as the $\Hxy$-norm dominates  the $\Hms$-norm).
\end{Rem}

\begin{Rem}\label{rem-asymp}
Theorem \ref{thm.limabs}  excludes  the values $\omega_{\mathrm{s}} = \pm \, \Op\in \sigma_{\rm exc}$ for $\Oe \neq \Om$ even if  $ \pm \, \Op \notin \sigma_{p}(\bbA)$. Indeed, they require a special study as  the Jacobian $\JacE$ in \eqref{eq.density-non-crit} is singular  (see \eqref{eq.jacob}) at $\pm \, \Op$ since $\pm \, \omega_\scE(\cdot)$ has an  horizontal asymptote (see \eqref{eq.expressionlambdae} and the first  figure of \ref{fig.speczones1}). In  section  5 of \cite{Cas-Haz-Jol-22}, we prove a limiting absorption and limiting amplitude principle in a weaker topology at $\pm \, \Op$.
\end{Rem}

\noindent \textbf{The critical case $\Oe = \Om$}. In this case, $\bbP_{\rm ac}$ and $\bbP_{\rm div0}$ no longer coincide. They actually differ from the sum of the eigenprojection associated to $\pm\Op$ (see \eqref{eq.def-Pac}). Hence, the spectral representation $R_{\rm ac}(\zeta)$  given by  Theorem \ref{th.diagA} is now
\begin{equation*}
R_{\rm ac}(\zeta) = \bbF^{*}\,\frac{\mathds{1}_{\bbR\setminus \{\pm \, \Op\}}(\omega)}{\omega-\zeta}\ \bbF
\quad\text{for }\zeta \in \bbC\setminus\bbR.
\end{equation*}
Thus, the  proof remains valid if we simply replace the definition \eqref{eq.def-f-sin-reg} of $f_\zeta^{\rm sin}$ and $f_\zeta^{\rm reg}$ by
\begin{equation*}
f_\zeta^{\rm sin}(\omega) := \mathds{1}_{\bbR\setminus \{\pm \, \Op\}}(\omega)\,\frac{\mathds{1}_{J}(\omega)}{\omega-\zeta}
\ \mbox{ and } \
f_\zeta^{\rm reg}(\omega) := \mathds{1}_{\bbR\setminus \{\pm \, \Op\}}(\omega)\,\frac{\mathds{1}_{\bbR\setminus J}(\omega)}{\omega-\zeta}.
\end{equation*}
where   $ \mathds{1}_{\bbR\setminus \{\pm \, \Op\}}(\omega)\,\mathds{1}_{J}(\omega)=\,\mathds{1}_{J}(\omega)$ (since  $\pm \Op \in \sigma_{\rm exc}$ and thus $J\subset \bbR\setminus \sigma_{\rm exc}\subset \bbR\setminus \{\pm \, \Op\}$).

\subsection*{Proof of Theorem \ref{th.ampllim}}
{\bf The non critical case $\Oe\neq \Om$}. We prove here \eqref{eq.lim-ampl} assuming that $\bbG \in \Hps\cap \Hxydiv $. As $\bU(t) = \phi_{\omega_{\mathrm{s}},t}(\bbA) \,\bbG$ where   $\phi_{\omega_{\mathrm{s}},t}(\cdot)$ given by \eqref{eq.phiduhamel2} is bounded and  $\bbG=\bbP_{\rm ac}\bbG$ (since $ \Hxydiv$ is the range of $\bbP_{\rm ac}$ for $\Oe\neq \Om$, see \eqref{eq.def-Pac}), we can use formula \eqref{eq.calc-fonct-ac-lim} (for $f=\phi_{\omega_{\mathrm{s}},t}(\cdot)$) to get
\begin{equation}\label{eq.decompUt}
\forall t \geq 0, \quad
\bU(t)=\phi_{\omega,t}(\bbA) \,  \bbP_{\rm ac} \,\bbG= \lim_{\Hxy}\int_{\mathbb{R}} \rmi \,\frac{\rme^{-\rmi\omega\, t} -\rme^{-\rmi \omega_{\mathrm{s}} \, t}}{\omega-\omega_{\mathrm{s}}} \,\bbM_{\omega}\bbG  \ \rmd \omega.
\end{equation}
The idea is to rewrite \eqref{eq.decompUt} as follows:
\begin{equation}\label{eq.intampl1}
\bU(t) = \rme^{-\rmi \omega_{\mathrm{s}} \, t} \lim_{\Hxy} \int_{\mathbb{R}}  \left( \frac{ \rmi \,\rme^{-\rmi \, (\omega-\omega_{\mathrm{s}}) \,  t}  }{\omega-\omega_{\mathrm{s}}}\, \bbM_{\omega} \bbG  - \rmi  \frac{\bbM_\omega\bbG}{\omega-\omega_{\mathrm{s}}} \right) \rmd \omega
\end{equation}
 to  make appear the time-harmonic solution $\bbU_{\omega_{\mathrm{s}}}^{+} :=  - \rmi \,R^{+}_{\rm ac}(\omega_{\mathrm{s}}) \, \bbG\in \Hms$ given by the limiting absorption principle for $s>1/2$ via the formula \eqref{eq.lim-res}.\\
Then, we split the integral \eqref{eq.intampl1}, by integrating separately both functions inside the parentheses. It  has to be done carefully, since each of them is singular at $\omega=\omega_{\mathrm{s}}$, while $\phi_{\omega_{\mathrm{s}},t}(\omega)$ is not. To do this, one introduces two Cauchy principal values at $\omega=\omega_{\mathrm{s}}$ defined in $\Hms$, i.e.,
\begin{equation*}
\bU(t) = \rme^{- \rmi \, \omega_{\mathrm{s}} \, t} 
\left( \lim_{\Hms} \dashint_{\mathbb{R}}  \frac{   \rmi \, \rme^{-\rmi \, (\omega-\omega_{\mathrm{s}}) \,t}}{\omega-\omega_{\mathrm{s}}}\, \bbM_{\omega}\bbG  \,\rmd \omega 
 - \rmi  \lim_{\Hms} \dashint_{\mathbb{R}} \frac{\bbM_{\omega}  \bbG}{\omega-\omega_{\mathrm{s}}} \, \rmd \omega\right).
\end{equation*}
As the second Cauchy principal value is  precisely the one involved in  \eqref{eq.lim-res}, we obtain
\begin{equation}
\bU(t)  = \rme^{-\rmi \omega_{\mathrm{s}}\, t} \,\Big( \bbU_{\omega_{\mathrm{s}}}^{+} + \bV(t) - \pi\, \bbM_{\omega_{\mathrm{s}}}\,\bbG\Big)
\ \text{ where } \bV(t)  := 
\lim_{\Hms} \dashint_{\mathbb{R}}   \frac{ \rmi \, \rme^{-\rmi \, (\omega-\omega_{\mathrm{s}}) \,t} }{\omega-\omega_{\mathrm{s}}}\, \bbM_{\omega}\bbG \, \rmd \omega.
\label{eq.def-V}
\end{equation}
Then the proof of \eqref{eq.lim-ampl} will be complete once we have proved the following lemma.

\begin{Lem}\label{lem-lim-ampl}
Let $s>1/2$, $\omega \in \bbR\setminus\sigma_{\rm exc}$ and  $\bV(\cdot)$ be defined in \eqref{eq.def-V} for $t\geq 0$.	Then,  we have
\begin{equation}\label{eq.ampl-2}
\lim_{t\to +\infty} \big\| \bV(t) - \pi\, \bbM_{\omega}\,\bbG\big\|_{\Hms} = 0, \quad   \forall \bbG \in \Hps.
\end{equation}
\end{Lem}

\begin{proof}
As in the proof of the limiting absorption principle, we  separate the Cauchy principal value at $\omega_{\mathrm{s}}$ from the limit in $\Hms$ in the definition of $\bV(t)$. Again, we choose some $\rho > 0$ small enough so $J := [\omega-\rho,\omega+\rho]$ does not intersect  $\sigma_{\rm exc}$. It leads us to decompose $\bV(t)$ in the form
\begin{align}
\bV(t) & = \rmi \, \big( \bV^{\rm sin}(t) + \bV^{\rm reg}(t) \big)
\quad\text{where}\quad \nonumber \\[4pt]
\bV^{\rm sin}(t) := \dashint_{J}   \frac{\rme^{-\rmi \, (\omega-\omega_{\mathrm{s}}) \,t}}{\omega-\omega_{\mathrm{s}}} \, \bbM_{\omega}\bbG\,\rmd \omega \quad &\text{and}\quad 
\bV^{\rm reg}(t) := \lim_{\Hms} \int_{\mathbb{R} \setminus J}   \frac{\rme^{-\rmi \, (\omega-\omega_{\mathrm{s}}) \,t}   }{\omega-\omega_{\mathrm{s}}}\, \bbM_{\omega}\bbG\,\rmd \omega.
\label{eq.def-V-reg}
\end{align}
Using the decomposition  \eqref{eq.def-V-reg}, we  prove \eqref{eq.ampl-2} by showing  successively that
\begin{equation}
\lim_{t\to +\infty} \big\| \bV^{\rm sin}(t) +  \rmi\pi\, \bbM_{\omega_{\mathrm{s}}}\,\bbG\big\|_{\Hms} = 0
\quad\text{and}\quad
\lim_{t\to +\infty} \big\| \bV^{\rm reg}(t) \big\|_{\Hms} = 0.
\label{eq.two-limits}
\end{equation}

\textbf{(i)} Let us first consider $\bV^{\rm sin}(t)$ that we rewrite as
\begin{align*}
\bV^{\rm sin}(t) = v^{\rm sin}(t) \ \bbM_{\omega_{\mathrm{s}}}\bbG \ + \  \rme^{\rmi\omega_{\mathrm{s}}\,t}\ \widetilde{\bV}^{\rm sin}(t) 
\quad & \text{where}\quad
v^{\rm sin}(t) := \dashint_{J}  \frac{\rme^{-\rmi \, (\omega-\omega) \,t}}{\omega-\omega_{\mathrm{s}}}\,\rmd \omega \quad \text{and}\\[4pt]
\widetilde{\bV}^{\rm sin}(t) := \int_J \rme^{-\rmi\omega\,t} \,\widetilde{\bV}_\omega \,\rmd\omega 
\quad & \text{with}\quad
\widetilde{\bV}_\omega := \frac{\big(\bbM_{\omega}-\bbM_{\omega_{\mathrm{s}}}\big)\bbG}{\omega-\omega_{\mathrm{s}}}.
\end{align*}
The latter integral is no longer a Cauchy principal value since $J \ni \omega \mapsto \widetilde{\bV}_\omega \in \Hms$ is Bochner integrable. Indeed,  by  H\"older continuity of $\bbM_{\omega}$ (Theorem \ref{th.Holder-dens-spec}), for all $\gamma\in\Gamma_J$,   there exists $C^\gamma_J>0$ such that
\begin{equation*}
\forall \omega \in J\setminus\{\omega\},\quad
\big\| \widetilde{\bV}_\omega \big\|_{\Hms} 
\leq C^\gamma_J \ |\omega-\omega_{\mathrm{s}} |^{-1+\gamma}\, \| \bbG\|_{\Hps}.
\end{equation*}
Thus, the Riemann-Lebesgue theorem (applied to $\Hms$-valued Bochner integrals) gives us
\begin{equation}
\lim_{t\to +\infty} \big\| \widetilde{\bV}^{\rm sin}(t) \big\|_{\Hms}= 0.
\label{eq.lim-tildeV}
\end{equation}
Besides, using the change of variable $\xi=(\omega-\omega_{\mathrm{s}})t$, we derive the limit of  $v^{\rm sin}(t) $ as $t\to+\infty$:
\begin{equation*}\label{eq.Vsing}
v^{\rm sin}(t) := \dashint_{-\rho t}^{+\rho t} \frac{\rme^{-\rmi \, \xi}}{\xi}\,\rmd \xi  \rightarrow   -\rmi \pi \quad  \mbox{ when } t\to +\infty,
\end{equation*}
where   the Cauchy principal value is at $\xi = 0$ and the limit can be shown via standard complex analysis  (see section 6.5 of \cite{Saff-02}).
Together with \eqref{eq.lim-tildeV}, this yields the first statement of \eqref{eq.two-limits}.

\textbf{(ii)} Consider now $\bV^{\rm reg}(t)$ defined in \eqref{eq.def-V-reg}. In view of formula \eqref{eq.calc-fonct-ac-lim}, we can rewrite it as
\begin{equation*}
\bV^{\rm reg}(t) = f^{\rm reg}_t(\bbA) \bbP_{\rm ac}\bbG
\quad\text{where}\quad
f^{\rm reg}_t(\omega) := \mathds{1}_{\mathbb{R} \setminus J}(\omega) \    \frac{\rme^{-\rmi \, (\omega-\omega_{\mathrm{s}}) \,t}   }{\omega-\omega_{\mathrm{s}}},
\end{equation*}
since $\omega \mapsto f^{\rm reg}_t(\omega)$ is a bounded function on $\bbR.$ This shows  that $\bV^{\rm reg}(t)$  belongs to $\Hxy$ and that the limit in \eqref{eq.def-V-reg} can be taken in $\Hxy$ instead of $\Hms$. This limit is constructed via an increasing sequence $(S_n)$ of compact subsets of $S := \mathbb{R} \setminus (J\cup\sigma_{\rm exc})$ whose union covers $S$, so that
\begin{equation}
\bV^{\rm reg}(t) = \lim_{n\to\infty} \bV^{\rm reg}_n(t)
\quad\text{where}\quad
\bV^{\rm reg}_n(t) := \bbE(S_n) \bV^{\rm reg}(t) = \int_{S_n}   \frac{\rme^{-\rmi \, (\omega-\omega_{\mathrm{s}}) \,t}   }{\omega-\omega_{\mathrm{s}}}\, \bbM_{\omega}\bbG \,\rmd \omega.
\label{eq.def-Vnreg}
\end{equation}
From the above definitions of $\bV^{\rm reg}(t)$ and $\bV^{\rm reg}_n(t)$, we have
\begin{equation*}
\bV^{\rm reg}(t) - \bV^{\rm reg}_n(t) = 
f^{\rm reg}_t(\bbA) \bbE(S \setminus S_n)\bbP_{\rm ac}\bbG,
\end{equation*}
from which we deduce that
\begin{equation*}
\left\| \bV^{\rm reg}(t) - \bV^{\rm reg}_n(t) \right\|_{\Hms} 
\leq \left\| \bV^{\rm reg}(t) - \bV^{\rm reg}_n(t) \right\|_{\Hxy}
\leq \left\|f^{\rm reg}_t\right\|_\infty \left\|\bbE(S \setminus S_n)\bbP_{\rm ac}\bbG\right\|_{\Hxy},
\end{equation*}
where $\left\|f^{\rm reg}_t\right\|_\infty = \rho^{-1}.$ Moreover we know from \eqref{eq.norm-Eac} that
\begin{eqnarray}
 \left\| \bbE(S \setminus S_n)\bbP_{\rm ac}\bbG \right\|_{\Hxy}^2=\int_\bbR \mathds{1}_{S \setminus S_n}(\omega) \, \langle\bbM_{\omega}\,  \bbG, \bbG \rangle_s\,\rmd\omega, \label{eq.intSn}
\end{eqnarray}
where   $\omega\mapsto \langle\bbM_{\omega} \bbG, \bbG\rangle_s \in L^1(\bbR)$ (see Corollary \ref{cor.specmes}). 
Hence,  by definition of $(S_n)$, \eqref{eq.intSn} tends to 0  (independently of $t$) as $n\to +\infty$ (by the  Lebesgue's dominated convergence Theorem). Thus the convergence of $\bV^{\rm reg}_n(t)$ to $\bV^{\rm reg}(t)$ is uniform in $t$ and for any given $\delta>0$,  it exists $n_\delta\in\mathbb{N}\setminus \{0\}$ such that
\begin{equation*}
\forall t \geq 0,\quad 
\left\| \bV^{\rm reg}(t) - \bV^{\rm reg}_{n_\delta}(t) \right\|_{\Hms} 
\leq \delta/2.
\end{equation*}
Moreover, as $S_{n_\delta}$ is bounded,   the Riemann-Lebesgue theorem on Bochner integrals  implies  (as in (i)) that $\bV^{\rm reg}_{n_\delta}(t)$ in \eqref{eq.def-Vnreg}  tends to 0 in $\Hms$ as $t \to +\infty.$ Thus, it exists  $T_\delta > 0$ such that
\begin{equation*}
\forall t \geq T_\delta,\quad 
\left\| \bV^{\rm reg}_{n_\delta}(t) \right\|_{\Hms} 
\leq \delta/2.
\end{equation*}
By the triangle inequality, we get that for any $\delta>0$,  it exists $T_\delta > 0$ such that $\| \bV^{\rm reg}(t) \|_{\Hms} \leq \delta$, $\forall t \geq T_\delta.$ Thus, we have proved the second statement of \eqref{eq.two-limits} and it concludes  the proof.
\end{proof}

\noindent {\bf The critical case}. We assume now that $\Oe=\Om$ and prove the asymptotic behavior \eqref{eq.lim-ampl-reson} for an excitation $\bbG \in \Hps \cap \Hxydiv$. From \eqref{eq.def-Pac}, we see that $\bbG$ can be decomposed as
\begin{equation}\label{eq.decompPac-PP}
\bbG = \bbP_{\rm ac}\bbG+ \bbP_{-\Op}\bbG + \bbP_{+\Op}\bbG.
\end{equation}
Hence the solution $\bU(t) = \phi_{\omega,t}(\bbA) \,\bbG$ to  \eqref{eq.schro} can be decomposed accordingly:
\begin{align*}
& \bU(t) = \bU_{\rm ac}(t) + \bU_{-\Op}(t) + \bU_{+\Op}(t)
\quad\text{where} \\[4pt]
& \bU_{\rm ac}(t) := \phi_{\omega,t}(\bbA)\,\bbP_{\rm ac}\bbG
\quad\text{and}\quad
\bU_{\pm \, \Op}(t) := \phi_{\omega,t}(\bbA)\,\bbP_{\pm\, \Op}\bbG.
\end{align*}

On the one hand, the asymptotic behaviour of $\bU_{\rm ac}(t)$ results from the previous lines since  \eqref{eq.decompUt} holds true by replacing  $\bU(t)$  by $\bU_{\rm ac}(t)$. We obtain
\begin{equation*}
\lim_{t\to+\infty}\Big\| \bU_{\rm ac}(t) + \rmi \, R^{+}_{\rm ac}(\omega) \, \bbG \, \rme^{-\rmi \omega t} \Big\|_{\Hms} = 0.
\end{equation*}
On the other hand, Theorem \ref{th.diagA} tells us that the operator $\phi_{\omega,t}(\bbA)$ is a multiplication by $\phi_{\omega,t}(\pm\, \Op)$ in the range of the spectral projection $ \bbP_{\pm \, \Op}$ associated to the eigenvalues $\pm \Op$. Hence, one has
\begin{equation*}
\bU_{\pm\, \Op}(t) = \phi_{\omega,t}(\pm \, \Op)\,\bbP_{\pm \, \Op}\bbG \quad  \mbox{ and this concludes the proof of  \eqref{eq.lim-ampl-reson}.}
\end{equation*}
\section{The case of a slab of dispersive media}\label{sec-slab}
\subsection{Description of the model}\label{sec-slab-model}
We consider in this section the case of two interfaces for a transmission problem between a dielectric and a metamaterial. This study  was done during the PhD thesis \cite{Ros-23} co-supervised by both authors and is the object of an article  in preparation ``M. Cassier and P. Joly and L. A. Rosas Mart\'inez, On guided waves by a slab of Drude material embedded in the vacuum''.\\[6pt]
\noindent More precisely,  we analyse  the (TE) Maxwell's equations in a medium composed of a layer of  a Drude non-dissipative material 
$$
\mathcal{L}=\{ \bx=(x,y)\in \bbR^2 \mid -L<x<L   \}
$$
embedded in the vacuum which fills the complementary open set $\bbR^2\setminus \overline{\mathcal{L}} $.\\
\begin{figure}[h!]
\centering
 \includegraphics[width=0.52\textwidth]{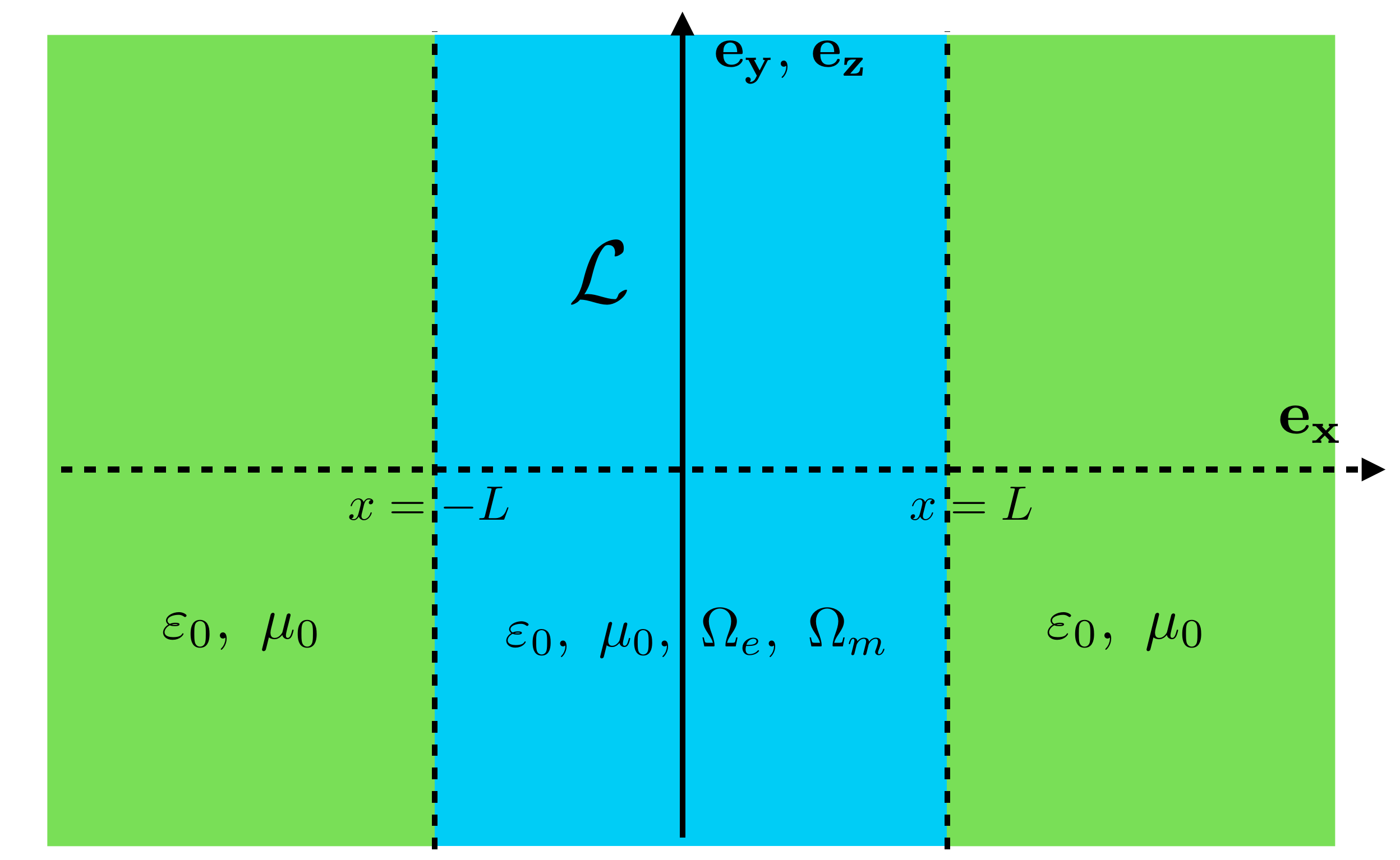}
 \caption{Slab of Drude non-dissipative material $\mathcal{L}$ of width $2L$ embedded in the vacuum.}
 \label{fig.medslab}
\end{figure}
\subsection{Common properties with the transmission problem studied in section \ref{sect-transm}}\label{sec-common-prop}
\noindent The difference with the transmission problem analysed in section   \ref{sect-transm}  relies on the fact that the non-dissipative Drude  material and the vacuum fill  respectively $\mathcal{L}$  and $\bbR^2\setminus \overline{\mathcal{L}} $ instead of the half planes
$\mathbb{R}_+^2$ and $\mathbb{R}_-^2$. Therefore,  the (TE) Maxwell's system \eqref{TE}, the evolution equation \eqref{eq.schro}, the associated Hilbert space $\mathcal{H}$ \eqref{eq.defHxy}, the  propagative self-adjoint operator $\bbA$, its domain $\rmD(\bbA)$, $\ldots$ are defined similarly  if one replaces by $\mathbb{R}_+^2$ by  $\mathcal{L}$  and   $\mathbb{R}_-^2$  by $\bbR^2\setminus \overline{\mathcal{L}} $.\\[6pt]
A similar  result as the Proposition  \ref{prop.spectrumA} holds in this new geometry.  Namely,
$$
 \{ -\Om, 0, \Om\}\in \sigma_p(\bbA),
$$
and these eigenvalues are of infinite multiplicity. They are also associated to non-propagating waves whose eigenspaces,  $\ker(\bbA)$ and $\ker(\bbA\, \pm \, \Om \bbI)$ are given by gradients supported respectively in the vacuum and  in the Drude the medium. Indeed, the formula for these eigenspaces  \eqref{eq.kernel} and  \eqref{eq.eigenspaces} and for the orthogonal of their direct sum, the space $\Hxydiv$ defined by \eqref{eq.Hxydiv} and \eqref{eq.Hxydiv2}, hold by replacing $\mathbb{R}_+^2$ by  $\mathcal{L}$  and   $\mathbb{R}_-^2$  by $\bbR^2\setminus \overline{\mathcal{L}} $.
\\[6pt]
\noindent Finally,  using the invariance of the medium  in the  $y-$direction, one reduces the dimension of the problem by applying the partial Fourier transform  \eqref{eq.deffour}. It allows to decompose the operator $\bbA$ via  \eqref{eq.AtoAk} as a direct integral of  reduced operators $(\bbA_k)_{k\in \bbR }$.  Thus,   the mathematical quantities introduced in section \ref{sec-reduced-op}:  the Hilbert space  $\Hx$,  the self-adjoints operators $\bbA_k$ for $k\in \bbR$,  their domain $\rmD(\bbA_k)$ are defined similarly  by substituting  $\bbR_+$ by $(-L,L)$. 
\subsection{Guided waves: definition and equations of propagation}\label{sec-guided-waves-def}

\noindent  In this separable geometry, we are interested to the existence of guided waves, which are waves localized in  the $x$-direction and propagating in the $y$-direction. These waves are non trivial  solutions (in the distributional sense) of  the (TE) Maxwell's system \eqref{TE}
of the form
\begin{equation}  \label{eq.guideswaves} 
\bU_{k,\omega}(x,y,t)= \wlk(x) \,  \rme^{\rmi\,  (k  \, y-\omega t)}  \  \mbox{ with } \  \wlk=(e_{k,\omega}, \bh_{k,\omega}, \dot{p}_{k,\omega}, \dot \bm_{k,\omega})^{\top}\in \Hx,
\end{equation}
for  a propagative frequency $\omega \in \bbR \setminus  \{ \pm\,  \Om, 0\}$ and a  wavenumber $k\in \bbR$ (in the $y$-direction).
One easily sees that  it is   equivalent to find an eigenfunction  $\wlk\in \ker(\bbA_k-\omega \mathrm{I})$ associated to the eigenvalue $\omega$, namely 
\begin{equation}\label{eq.systemeigenguided}
\bbA_{k}  \wlk= \omega \wlk  \ \mbox{ with }  \ \wlk\in \rmD(\bbA_k).
\end{equation}
In other words, the existence of a guided wave  at the propagative frequency $\omega\in \bbR \setminus  \{ \pm \, \Om, 0\}$ for the wavenumber $k\in \bbR$ is equivalent to the fact that $\omega\in \sigma_p(\bbA_k)\setminus  \{ \pm\, \Om, 0\}$.\\[6pt]
By eliminating the unknowns $\bh_{k,\omega}, \dot{p}_{k,\omega}, \dot \bm_{k,\omega}$,   one shows (as in section  
\ref{sec-geneigen}) that solutions $\wlk$ of \eqref{eq.systemeigenguided} are given by the  ``vectorizator'' operator $\bbV_{\omega,k}$ introduced in \eqref{eq.def-Vk}:
\begin{equation}\label{eq.funcguided}
  \wlk=\bbV_{k,\omega} e_{k,\omega},
\end{equation}
where $e_{k,\omega}\in H^1(\bbR)$ is a  solution of  the following scalar Sturm-Liouville equation:
\begin{equation}\label{eq.Esturmguided}
\displaystyle - \frac{\rmd}{\rmd x}\left( \frac{1}{\mu(\omega,\cdot )}\frac{\rmd e_{k,\omega}}{\rmd x}\right)+\frac{ \mathcal{D}_{k,\omega} }{\mu(\omega,\cdot)}\, e_{k,\omega} =0, \quad  \mbox{ with }
 \quad \mathcal{D}_{k,\omega}(x) :=k^2-\varepsilon(\omega, x)\, \mu(\omega,x)\, \omega^2
\end{equation}
with  $\eps(\omega,x)$ and  $\mu(\omega, x)$ defined by replacing $\bbR_+$ by $(-L,L)$, $\bbR_-$ by  $\bbR\setminus (-L,L)$  and $\bx$ by $x$ in \eqref{eq.defepsmu}.
As \eqref{eq.Esturmguided} is taken  in sense of distributions, it contains the transmission conditions:
\begin{equation}\label{eq.transmissionconditiontierce}
[e_{k,\omega}]_{x=\pm L}=0 \quad  \mbox{ and } \quad  \Big[  \frac{1}{\mu(\omega,\cdot)}\frac{\rmd e_{k,\omega}}{\rmd x}\Big]_{x=\pm L}=0.
\end{equation}
To analyse  the solution of the equation  \eqref{eq.Esturmguided}, we exploit the symmetry of the medium with respect to $x=0$. To this aim, we introduce the following orthogonal decomposition of $H^1(\bbR)$:
\begin{equation}\label{eq.decomp}
H^1(\bbR)=H^{1}_{\mathrm{ev}}(\bbR) \oplus H^{1}_{\mathrm{od}}(\bbR),
\end{equation}
where $H^{1}_{\mathrm{ev}}(\bbR)$ (resp. $H^{1}_{\mathrm{od}}(\bbR) $) is the space $H^1(\bbR)$ of even (resp. odd) functions.
Using the decomposition \eqref{eq.decomp}, one easily checks that for $(k,\omega)$ fixed,  solving the equation \eqref{eq.Esturmguided} in $H^1(\bbR)$ is equivalent to solve it in  $H^{1}_{\mathrm{ev}}(\bbR) $ and $H^{1}_{\mathrm{od}}(\bbR) $ separately. Thus, as the even and odd  solutions of \eqref{eq.Esturmguided}:  $ e_{k,\omega}^{\mathrm{ev}}$ and $ e_{k,\omega}^{\mathrm{od}}$  are smooth functions at $x=0$, it is equivalent by parity to know their restriction on the half line $\bbR^+$ which satisfies  respectively  for $\mathrm{p}\in \{\mathrm{od}, \, \mathrm{ev} \}$:
\begin{equation}\label{eq.Esturmguidedevod}
  \displaystyle - \frac{\rmd}{\rmd x}\left( \frac{1}{\mu(\omega,\cdot )}\frac{\rmd e_{k,\omega}^{\mathrm{p}}}{\rmd x}\right)+\frac{ \mathcal{D}_{k,\omega}}{\mu(\omega,\cdot)}\, e_{k,\omega}^{\mathrm{p}} =0, \
 \displaystyle   \ \displaystyle [e_{k,\omega}^{\mathrm{p}}]_{x= L}=0  \mbox{ and }   \Big[  \frac{1}{\mu(\omega,\cdot)}\frac{\rmd e_{k,\omega}^{\mathrm{p}}}{\rmd x}\Big]_{x= L}=0,
\end{equation}
with $e_{k,\omega}^{\mathrm{ev}}$ and $e_{k,\omega}^{\mathrm{od}}$ satisfying a Neumann and a Dirichlet boundary condition at $x=0$:
\begin{equation}\label{eq.even-odd}
\frac{\rmd e_{k,\omega}^{\mathrm{ev}}}{\rmd x}(0)=0 \quad \mbox{ and } \quad e_{k,\omega}^{\mathrm{od}}(0)=0.
\end{equation}
In this chapter, for length purpose, we voluntary  concentrate on the even case, i.e. when the component  of the electrical field is even. 
One can proceed similarly for the odd case (see \cite{Ros-23}). 


\subsection{Guided waves: localization in the $(k,\omega)$ plane and dispersion equations}\label{sec-disp-rel}

As $\mathcal{D}_{k,\omega} =\mathcal{D}_{|k|,|\omega|}$ for all $(k,\omega) \in \bbR^2$ and $\mu(\omega,\cdot)$ is even in $\omega$, we can restrict our study of the solutions  $e_{k,\omega}^{\mathrm{ev}}$ of \eqref{eq.Esturmguidedevod} and \eqref{eq.even-odd}, to the quadrant $k \geq 0$ and $\omega \geq 0.$ 
Moreover, as we are looking for $H^1(\bbR_+)$ functions on the half-line $\bbR_+$, it imposes that the guided waves are necessarily evanescent in the vacuum.  Thus, the solutions lie in the region where $\mathcal{D}_{k,\omega}^-$, defined by \eqref{eq.defTheta} is negative. This region is defined by 
\begin{equation}\label{def-mathcalN}
\Lambda:=\{  (k,\omega)\in \bbR^+\times( \bbR_{+} \setminus \{ \Om\}) \mid \omega< c \, k\}
\end{equation}
where $\bbR^+:=\bbR_+\cup \{ 0\}$ and $c:=(\varepsilon_0 \,\mu_0)^{-1/2}$ is the speed of light in the vacuum. We  split now  $\Lambda$ in two disjoint sub-regions depending on the sign of  $\mathcal{D}_{k,\omega}^+$ (defined in  \eqref{eq.defTheta}):
\begin{eqnarray*}
&& \Lambda=\Lambda^{-}\cup \,\Lambda^{+} \quad \mbox{ with } \quad  \Lambda^+:= \Lambda\setminus  \Lambda^{-}  \mbox{ and } \\[10pt]
 &&    \Lambda^-:=\left\{(k,\omega) \in \Lambda \mid  0 < \omega< \min(\Oe,\Om) \text{ and } k < k_\scI(\omega) \right\}.
\end{eqnarray*}
We notice  that  $\Lambda^-$ coincides with  $\zEI\cap (\bbR^+\times \bbR_+)$, where $\zEI$  is defined in \ref{sec-geneigen}, see figure \ref{fig.speczones1}.\\[4pt]
\noindent Physically, as  $\mathcal{D}_{k,\omega}^+<0$ in $\Lambda^{-}$,  the associated  guided waves in this region are propagative in the slab. 
 Contrariwise, in $\Lambda^{+}$, as $\mathcal{D}_{k,\omega}^+\geq 0$, the associated  guided waves are exponentials solutions in the slab.  Indeed, we will prove that these waves are localized also in a sub-region of  $\Lambda^{+}$, where the Drude material behave as negative material since  $\mu(\omega,x)$ is negative. Thus, they are  plasmonic waves which are propagating and localized at the vicinity  of the two interfaces $x=\pm L$ between the positive and the negative medium.\\[4pt]
\noindent We are looking for solutions in $H^1(\bbR_+)$ of the Sturm-Liouville problem \eqref{eq.Esturmguidedevod} with the Neumann condition \eqref{eq.even-odd}. These solutions are on the one hand $H^1-$functions, thus they are necessarily exponentially decreasing on $]L, +\infty]$ On the other hand, they need to satisfy a continuity condition at $x=L$  imposes by the first transmission condition of \eqref{eq.Esturmguidedevod}. Therefore, if non-trivial solutions of \eqref{eq.Esturmguidedevod}   exists, the space of solutions is one dimensional and a basis function can be chosen by setting arbitrary $e_{k,\omega}^{\mathrm{ev}}(0)=1$ in the following way:
\begin{equation*}
e_{k,\omega}^{\mathrm{ev}}(x)= \operatorname{cosh}(\xi_{k,\omega}^{+} x) \  \mbox{ in } \ [0,L] \quad  \mbox{ and } \quad e_{k,\omega}^{\mathrm{ev}}(x)= \operatorname{cosh}(\xi_{k,\omega}^{+} L)\ \rme^{-\xi_{k,\omega}^{-} (x-L)} \ \mbox{ in } \  \bbR^+\setminus [0,L] ,
\end{equation*}
 where for $(k,\omega)\in \Lambda$:
\begin{equation}\label{eq.xidisp}
\xi_{k,\omega}^{-}=|\mathcal{D}_{k,\omega}^-|^{\frac{1}{2}}>0 \quad  \mbox{ and } \quad \xi_{k,\omega}^{+}=\left\{ \begin{array}{lll}

&  \rmi \, |\mathcal{D}_{k,\omega}^+|^{\frac{1}{2}} \quad &  \text{ if }   (k,\omega)\in \Lambda^-  , \\[6pt]
\displaystyle 
&  |\mathcal{D}_{k,\omega}^+|^{\frac{1}{2}} \quad & \text{ if }    (k,\omega)\in \Lambda^+.
\end{array}\right.
\end{equation}
The existence of such one dimensional space rely on the fact that the basis function has to satisfy the second transmission of \eqref{eq.Esturmguidedevod}. Thus, it exists if only if $(k, \omega)\in \Lambda^-$  is a solution of the following equation (refereed as the dispersion relation for even electric guided modes):
\begin{equation}\label{eq.dispersioneven}
\displaystyle \xi^{+}_{k,\omega} \operatorname{ tanh }(\xi^{+}_{k,\omega} \, L)=-\frac{\mu_+(\omega)}{\mu_0}\, \xi^{-}_{k,\omega} .
\end{equation}

\subsection{Analysis and existence of guided waves}\label{sec-disp-cas}
We introduce the following sets associated to the solutions of the  dispersion relations  \eqref{eq.dispersioneven} in  the zones $\Lambda^{+}$ and $\Lambda^{-}$:
\begin{equation}\label{eq.dispsets}
\mathcal{D}_{\mathrm{ev}}^{\pm}=\Big\{  (k,\omega)\in \Lambda_{}^{\pm} \mid \xi^{+}_{k,\omega} \operatorname{ tanh }(\xi^{+}_{k,\omega} \, L)=-\frac{\mu_+(\omega)}{\mu_0}\, \xi^{-}_{k,\omega} \Big\}.
\end{equation}
and denote their union by 
$$
\mathcal{D}_{\mathrm{ev}}:=\mathcal{D}_{\mathrm{ev}}^+ \cup \mathcal{D}_{\mathrm{ev} }^-.
$$
\eqref{eq.dispsets} is the precisely the set of ordered pairs $(k,\omega)\in \bbR^+\times (\bbR_+\setminus \{ \Om\})$ for which the Sturm-Liouville equation  \eqref{eq.Esturmguided} admits a non trivial even solution $e_{k,\omega}$.  Thus, at a fixed  $k\in \bbR$, $(| k|, |\omega|)\in \mathcal{D}_{\mathrm{ev}} $ if  only if $\omega \in \sigma_p(\bbA_k)\setminus \{0, \pm \, \Om \}$ and there exists an eigenfunction $ \wlk=\bbV_{k,\omega} e_{k,\omega}$, associated to $\omega$, whose first component $e_{k,\omega}$ (related to the electrical field) is  an even function. Hence, one defines the ``even part'' of the point spectrum $\sigma_p^{\mathrm{ev}}(\bbA_k)\subset \sigma_p(\bbA_k)$ the set of  eigenvalues of $\bbA_k$ in  $\bbR\setminus \{ 0, \pm\, \Om\}$ which admits an eigenfunction $ \wlk$ whose first component $e_{k,\omega}$ is even. Thus one has the following equivalence: 
\begin{equation}\label{eq.equivalnecespectrumdipserif}
\omega \in\sigma_p^{\mathrm{ev}}(\bbA_k)\Longleftrightarrow \, (|k|, |\omega|)\in \mathcal{D}_{\mathrm{ev}}.
\end{equation}
\noindent As the dispersion curves of the dispersive slab have critical points, we introduce the following standard definition.
\begin{Def} 
Let $I$ be an open  interval of $\bbR$  and  $f:I\subset \bbR \to \bbR$ be of function of class $\mathcal{C}^3$ on $I$. Then $f$ admits a critical point if and only if it exists $t_0\in I$ such that $f'(t_0)=0$. Such a critical point at  $t_0$ is a non degenerate maximum (resp. minimum)  if $f''(t_0)<0$ (resp. $f''(t_0)>0$) and a non-degenrate inflection point if  $f''(t_0)=0$ and $f^{(3)}(t_0)\neq 0$.
\end{Def}
\noindent The results presented here are obtained  via a parametrization of the solutions of the dispersion relation of \eqref{eq.dispsets}.  Compared to the case of the dielectric slab  presented in section  \ref{classic-case} and in \cite{Wil-76}, this parametrisation is not explicit. Indeed, the dispersion of the Drude medium complicates  significantly  the analysis of the dispersion curves which is done in details in the PhD Thesis \cite{Ros-23} and will be the object of an  ongoing article.
\begin{Thm}\label{thm.omegan}
One has 
\begin{equation}\label{eq.solutionsetsdisp}
\mathcal{D}_{\mathrm{ev}}^{-}=\bigcup_{n \in \mathbb{N}^*} \mathcal{C}_{n}, 
\end{equation}
where the  dispersion curves $\mathcal{C}_n$ are defined for all integer  $n\geq 1$ by
\begin{equation}\label{eq.kappaseuil}
\mathcal{C}_n=\{ (k, \omega_n(k))  \mid k> \kappa_n \} \ \mbox{ where} \  \kappa_n:= L \, \frac{\Om \Oe}{c^2} \, \big((\pi \, n)^2 + \frac{L^2}{c^2}\Big(\Om^2+ \Oe^2 \big) \Big)^{-\frac{1}{2}}
\end{equation}
is decreasing  to $0$ with  $n$ and $\omega_n :[\kappa_n,+\infty)\to \bbR^+$ for $n\geq 2$ satisfies the following properties:
\begin{enumerate}
\item $\omega_n$ is $\mathcal{C}^{\infty}$ on $(\kappa_n,+\infty)$ and  at least $C^1$ at $\kappa_n$ and $\omega_{n+1}<\omega_n$,
\item $\exists \,  \kappa_{\mathrm{cr},n}>\kappa_n $ such that $\omega_n$ is strictly increasing  on $[\kappa_n, \kappa_{\mathrm{cr},n} )$  (with $\omega_n'>0$), strictly decreasing on $( \kappa_{\mathrm{cr},n} ,+\infty)$  (with $\omega_n'<0$) and  $\omega_n( \kappa_{\mathrm{cr},n})$ is a  non-degenerate  maximum.  
\item $\omega_n(\kappa_n)=c \, \kappa_n\to 0$ as $n\to +\infty$ and $\omega_n'(\kappa_n)=c$. 
\item $\omega_n(k)=  \displaystyle   \Om \, \Oe  \, c^{-1} \, k^{-1} -a_n k^{-3}+o(k^{-3}) \mbox{ as} \, k\to +\infty$,  with $a_n>0$ satisfying $a_n<a_{n+1}$.
\end{enumerate}
 \end{Thm}
\noindent The following Theorem proves that $\mathcal{D}_{\mathrm{ev}}^{+}$ is made of single dispersion curve $\mathcal{C}_0$. To describe the monotonicity of the function $\omega_0$ whose graph defines $\mathcal{C}_0$, one needs to introduce the dimensionless parameters $ \rho$, $\boldsymbol{\Omega}_{\mathrm{m}}$, the function $\ \mathcal{S}_{0,\boldsymbol{\Omega}_{\mathrm{m}}, \rho}:(0,\infty)\to \bbR$ defined by
 \begin{eqnarray*}\label{eq.parameters}
 &&  \rho:=\frac{\Oe}{\Om}, \quad  \boldsymbol{\Omega}_{\mathrm{m}}= \frac{ \Om \, L}{c} \quad \mbox{ and } \quad \mathcal{S}_{0,\boldsymbol{\Omega}_{\mathrm{m}}, \rho}(\tau):=1+\frac{\beta_0(\tau)}{\boldsymbol{\Omega}^2_{\mathrm{m}}} \rho^2 \alpha_0(\tau)^{-\frac{1}{2}},\\
  \mbox{ where  } && \alpha_0(\tau)=\frac{1}{\operatorname{ tanh }(\tau) \, \big(\operatorname{tanh}(\tau)+\tau \big(1-\operatorname{tanh}(\tau)^2\big)} \ \mbox{ and } \ \beta_0(\tau)=\frac{\tau^3 \,  \big(1-\operatorname{tanh}(\tau)^2)}{\operatorname{tanh}(\tau)+\tau \big(1-\operatorname{tanh}(\tau)^2)}.
 \end{eqnarray*}
 One proves  after some tedious computations (done in \cite{Ros-23}) that the  smooth function  $ \mathcal{S}_{0,\boldsymbol{\Omega}_{\mathrm{m}}, \rho}$ admits a unique minimum at $\tau_c=\tau_c( \boldsymbol{\Omega}_{\mathrm{m}},\rho)$ on $(0,\infty)$ and that the existence and number of local extrema of  $\omega_0$ depend on  the sign  of $\rho-1 $ and the sign of  the minimum $ \mathcal{S}_{0,\boldsymbol{\Omega}_{\mathrm{m}}, \rho}(\tau_c)$.
\begin{Thm}\label{thm.omega0}
One has $
\mathcal{D}_{\mathrm{ev}}^{+}=\mathcal{C}_0,$ 
where the  dispersion curve 
$$
\mathcal{C}_0=\{ (k, \omega_n(k))  \mid k> \kappa_0 \}  \mbox{ with } \kappa_0=\kc\  \mbox {(is also given by formula \eqref{eq.kappaseuil} for $n=0$),}
$$
where the function $\omega_0:[\kappa_0, +\infty)\to \bbR^+$  is $C^{\infty}$ on $(k_0,+\infty)$   and al least $\mathcal{C}^1$ at $\kappa_0$ and satisfies 
\begin{eqnarray*}
\omega_0(\kappa_0)&=&c \, \kappa_0, \quad \, \omega_0'(\kappa_0)=c, \ \mbox{ and } \\[4pt]
 \omega_n(k) &= &   \Op + a_{\mathrm{\infty}} \, k^{-2} +O(k^{-4}) \mbox{ as} \, k\to +\infty  \ \mbox{ with } \ a_{\mathrm{\infty}}= \displaystyle \frac{\Om^3 }{8 c^2 \,\sqrt{2}} \, (\rho^2-1) .
\end{eqnarray*}
Moreover,  concerning the monotonicity, we have the following four situations (depending on the parameter $\rho$ and  the sign of $S_{0,\rho,\Omega_m}(\tau_c)$): 
\begin{enumerate}
\item If $\rho\geq 1$, it exists $\kappa^M_{\mathrm{cr}, 0} \in (\kappa_0,+\infty)$ such that $\omega_0$ is strictly increasing  on $[\kappa_0,\kappa_{\mathrm{cr}, 0}^M \ )$ (with $\omega_0'>0$), strictly decreasing on $(\kappa_{\mathrm{cr}, 0}^M\ ,+\infty)$ (with $\omega_0'<0$)  and $\omega_0(\kappa_{\mathrm{cr}^M, 0})$ is a non degenerate maximum of $\omega_0$.
\item If $\rho<1$ and 
\begin{enumerate}
\item if $S_{0,\rho,\Omega_m}(\tau_c)>0$ then   $\omega_0$ is  strictly increasing on $(\kappa_0,+\infty)$ with $\omega_0'>0$.
\item if $S_{0,\rho,\Omega_m}(\tau_c)=0$ then $\omega_0$ is strictly increasing on $(\kappa_0,+\infty)$ and there exists a  unique non-degenerate inflection  point at $\kappa_{\mathrm{cr}, 0}^{\mathrm{i}}$  with $\omega_0'>0$ on $(\kappa_0,+\infty)\setminus \{\kappa_{\mathrm{cr}, 0}^{\mathrm{i}} \}$.
\item  if $S_{0,\rho,\Omega_m}(\tau_c)<0$ then  $\omega_0$ has two critical points: a non degenerate  minimum  $\kappa_{\mathrm{cr}, 0}^{m}$ and a non degenerate maximum  $\kappa_{\mathrm{cr}, 0}^{M}$ satisfying $\kappa_0<\kappa_0^M<\kappa_{\mathrm{cr}, 0}^{m}$ so that $\omega_0$ is strictly increasing (resp. strictly decreasing) on $[0,\infty)\setminus (\kappa_0^M, \kappa_0^m) $ (on $(\kappa_0^M, \kappa_0^m) $).
\end{enumerate}
\end{enumerate}
 \end{Thm}

\subsection{Description the dispersion curves of the dispersive slab}\label{sec-disp-curves}
In this section, we provide various comments on the qualitative properties of the 
dispersion curves given by Theorems  \ref{thm.omegan} and \ref{thm.omega0}  that we illustrate with different figures (obtained via numerical simulations).
\subsubsection{General comments}\label{sec-disp-curves-genral-comments}
The dispersion curves $\mathcal{C}_n=\{(\kappa, \omega_n(k)),\mid k>\kappa_n \}$, $n\geq 0$  never cross each other since   $\omega_{n+1}< \omega_n $ and can be  divided in two subgroups.   On the one hand,  the curves $\mathcal{C}_n$ for $n\ge 1$  lie in the region $\Lambda^-$ and are  associated to guided waves which are propagative in the slab.  On the other hand,  the curve $\mathcal{C}_0$     lies in the region $\Lambda^+$  and is  related to guided modes which are  plasmonic modes (evanescent  in the slab and localized in the $x-$direction near the two interfaces $x=\pm L$). \\[6pt]
\noindent For $n\ge 1$, the curves $\mathcal{C}_n$ are associated (see figure \ref{fig.med-disp-ng2}) to even electric  guided modes $e_{k,\omega}^{\mathrm{ev}}$. Unlike in section \ref{classic-case} where the threshold $\kappa_n \to+ \infty$ as  $n \to +\infty$, here  $\kappa_n$ tends to $0$ as  $n \to +\infty$. Thus,   the even point spectrum $\sigma_p^{\mathrm{ev}}(\bbA_k)$ (defined by \eqref{eq.equivalnecespectrumdipserif} and given  by  Theorems \ref{thm.omegan} and  \ref{thm.omega0}):
$$\sigma_p^{\mathrm{ev}}(\bbA_k)=\{ \pm \,  \omega_n (|k|), \,  n  \in \mathbb{N}  \mid \kappa_n<|k| \}, \quad \forall k\in \bbR^*$$ is an infinite set, since  as $\kappa_n \to 0$, there are an infinite  number of indices $n$ for which  $\kappa_n<|k|$.\\
As in section \ref{classic-case}, the  curve $\mathcal{C}_n$ (see figure \ref{fig.med-disp-ng2}) are tangent to the line $\omega= k \, c$ at $\kappa=\kappa_n$ since $\omega_n'(\kappa_n)=c$. However, unlike in section \ref{classic-case} where $\omega_n$ is strictly increasing and $\omega_n(k)\to +\infty$ as $k\to+\infty$,   $\omega_n$ admits here  a unique global maximum at $\kappa_{\mathrm{cr},n}$ and decay to $0$ (see  Theorem \ref{thm.omegan}) as $k\to+\infty$.  In particular, as $\omega_n'>0$ on $(\kappa_{n}, \kappa_{\mathrm{cr},n})$ and $\omega_n'<0$ on $(\kappa_{\mathrm{cr},n}, +\infty)$, the associated guided modes  (in the slab) are forward modes  for $k\in (\kappa_{n}, \kappa_{\mathrm{cr},n})$ and backward mode  for $k\in (\kappa_{\mathrm{cr},n},+\infty)$ since their phase velocity  $\omega/k>0$ but the sign of their group velocity $ \omega_n'$ changes at $\kappa_{\mathrm{cr},n}$. This phenomenon does not exist for the non-dispersive slab (see section  \ref{classic-case}) where all the guides modes are forward modes   in the slab nor for the bilayered medium (analyzed in section  \ref{sect-transm}) where the modes of the spectral zone $\Lambda^-=\zEI$ are backward modes in the Drude medium (see section \ref{sec-geneigen} or section  3.2.2 of \cite{Cas-Haz-Jol-17} for more details).\\[6pt]
\noindent Concerning the curve $\mathcal{C}_0$ associated to even plasmonic waves  $e_{k,\omega}^{\mathrm{ev}}$,  it is tangent at $\kappa_0=\kappa_c$  to the line $\omega= k \, c$. This geometrical property is  not satisfied by  the curve $\omega_{\scE}$ of the bilayered medium (see the proof of Lemma 26 of  \cite{Cas-Haz-Jol-22}).\\ As for the bi-layered medium in the non-critical case, $\omega_0$ has an horizontal asymptote  in $\Op$ when $k\to +\infty$. For the bilayered medium,   $\omega_{\scE}$  is  either  strictly monotonous for $\Oe\neq \Om$ or flat in the critical case $\Oe=\Om$. Here, the situation is different, since there are four different scenarios  depending   on the parameters $\rho$ and $\boldsymbol{\Omega}_{\mathrm{m}}$.  $\mathcal{C}_0$ is either strictly increasing  or  admits one or two critical points which can be a non degenerate minimum, maximum or   inflection point (see Theorem \ref{thm.omega0} and figures \ref{fig.med-disp-0-case-1-2} and \ref{fig.med-disp-0-case-3-4}).  For the odd case (defined  in  the equations \ref{eq.Esturmguidedevod} and \ref{eq.even-odd}), one proves also that the dispersion  curves are not flat or even locally flat (see \cite{Ros-23}).  This is really different from the  critical  case $\Oe=\Om$ of the bi-layered medium where  the  dispersion curves  are  given by the constant functions  $\pm \, \Op$ for $|\bk|>\kc$ (see figure \ref{fig.speczones1}). This implies that the propagative frequencies  $\pm \, \Op \in \sigma_p(\bbA) $ and are resonances  of the system for  which the fields  blow up linearly in time.  Here, for the dispersive slab,  as $(\bbR\setminus \{0, \pm \, \Om\})\cap \sigma_p(\bbA)=\varnothing$, there is no interface resonance with  linear explosion  in time.

\begin{figure}[h!]
\centering
 \includegraphics[width=0.75\textwidth]{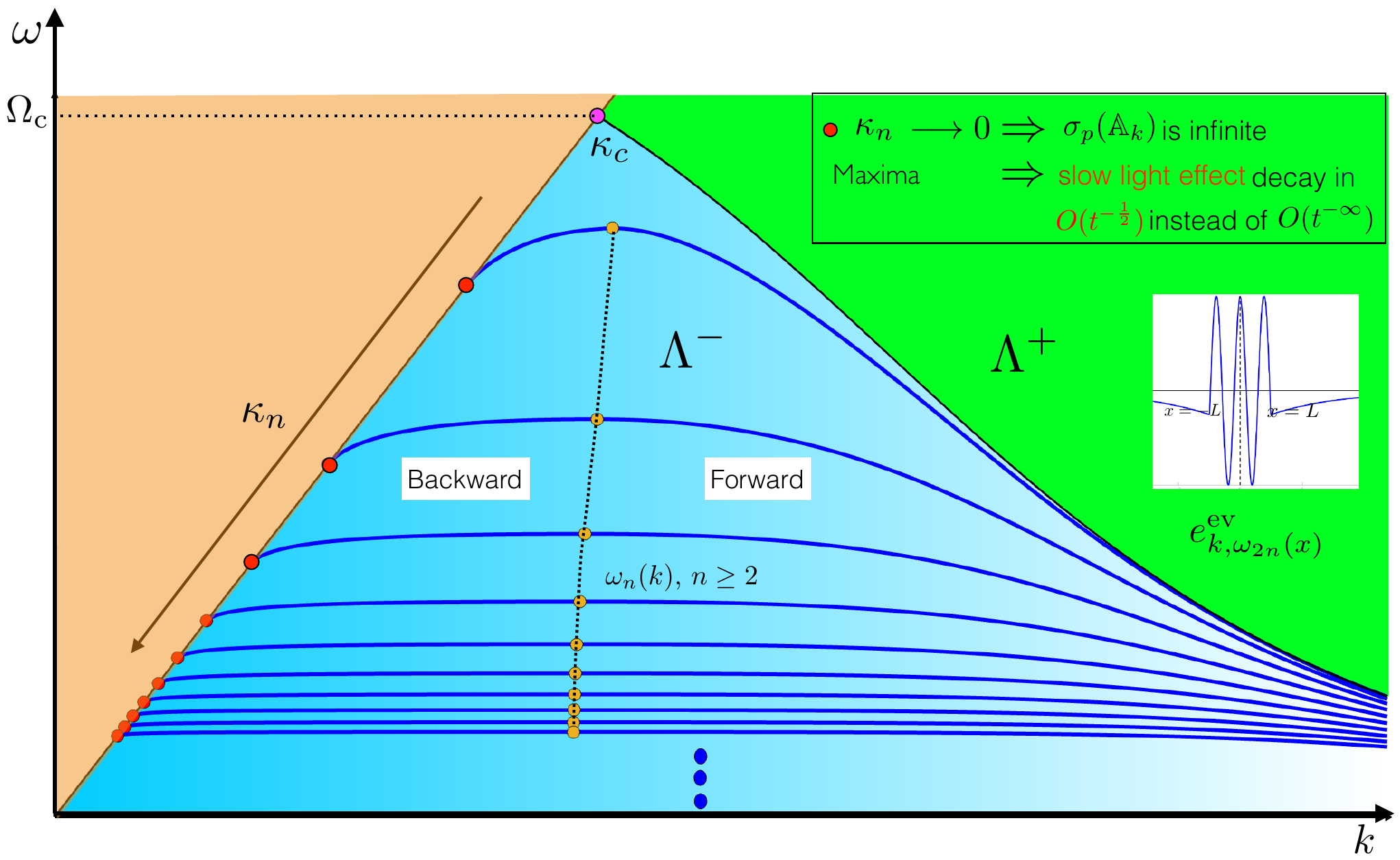}
 \caption{Dispersion curves $\mathcal{C}_n$ for $n\geq 2$ for a slab  of width $2L$ of  dispersive non-dissipative Drude medium  embedded in the vacuum, corresponding to the figure \ref{fig.medslab}.}
 \label{fig.med-disp-ng2}
\end{figure}

\begin{figure}[h!]
\centering
 \includegraphics[width=0.495\textwidth]{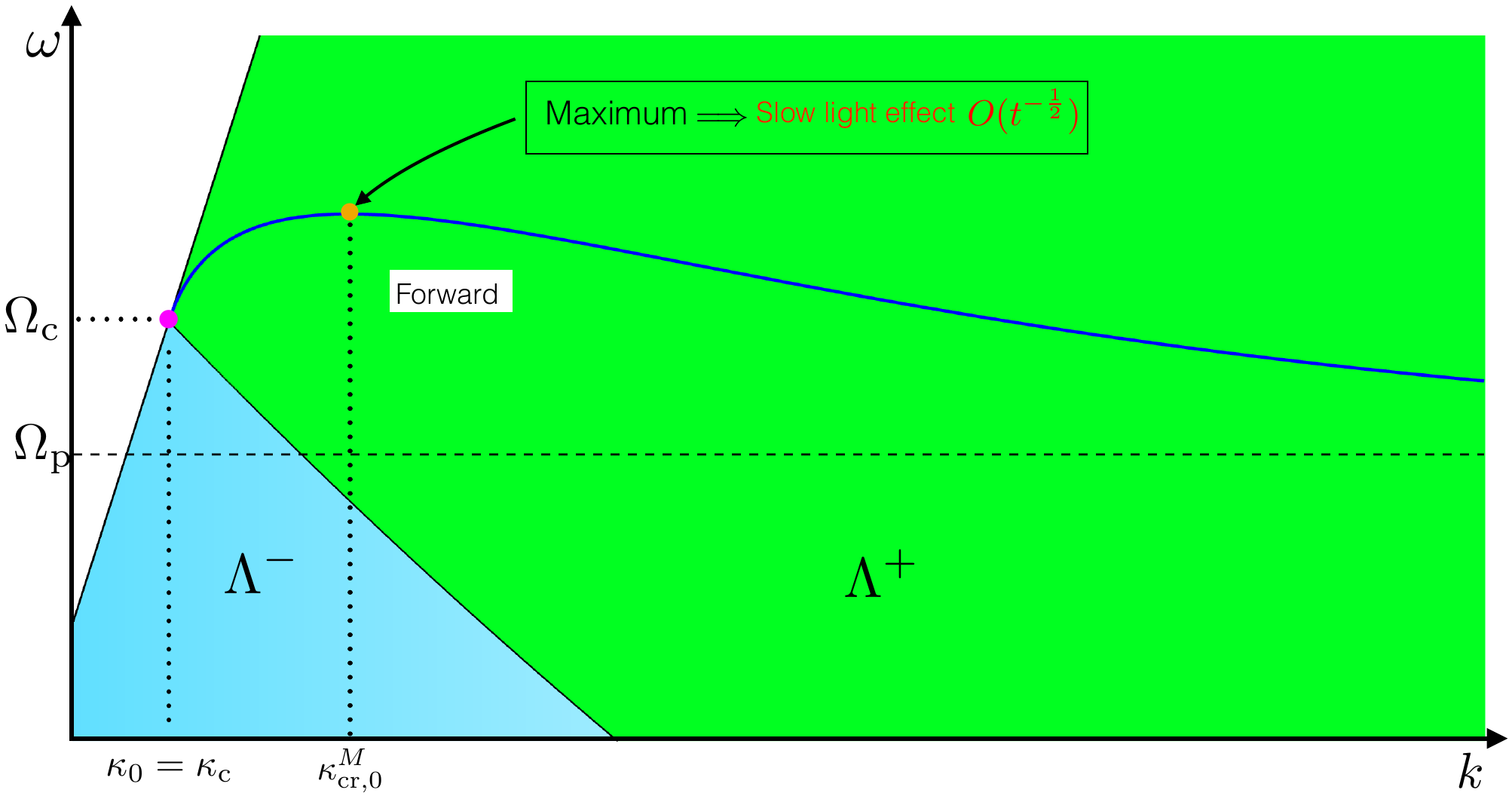}
  \includegraphics[width=0.495\textwidth]{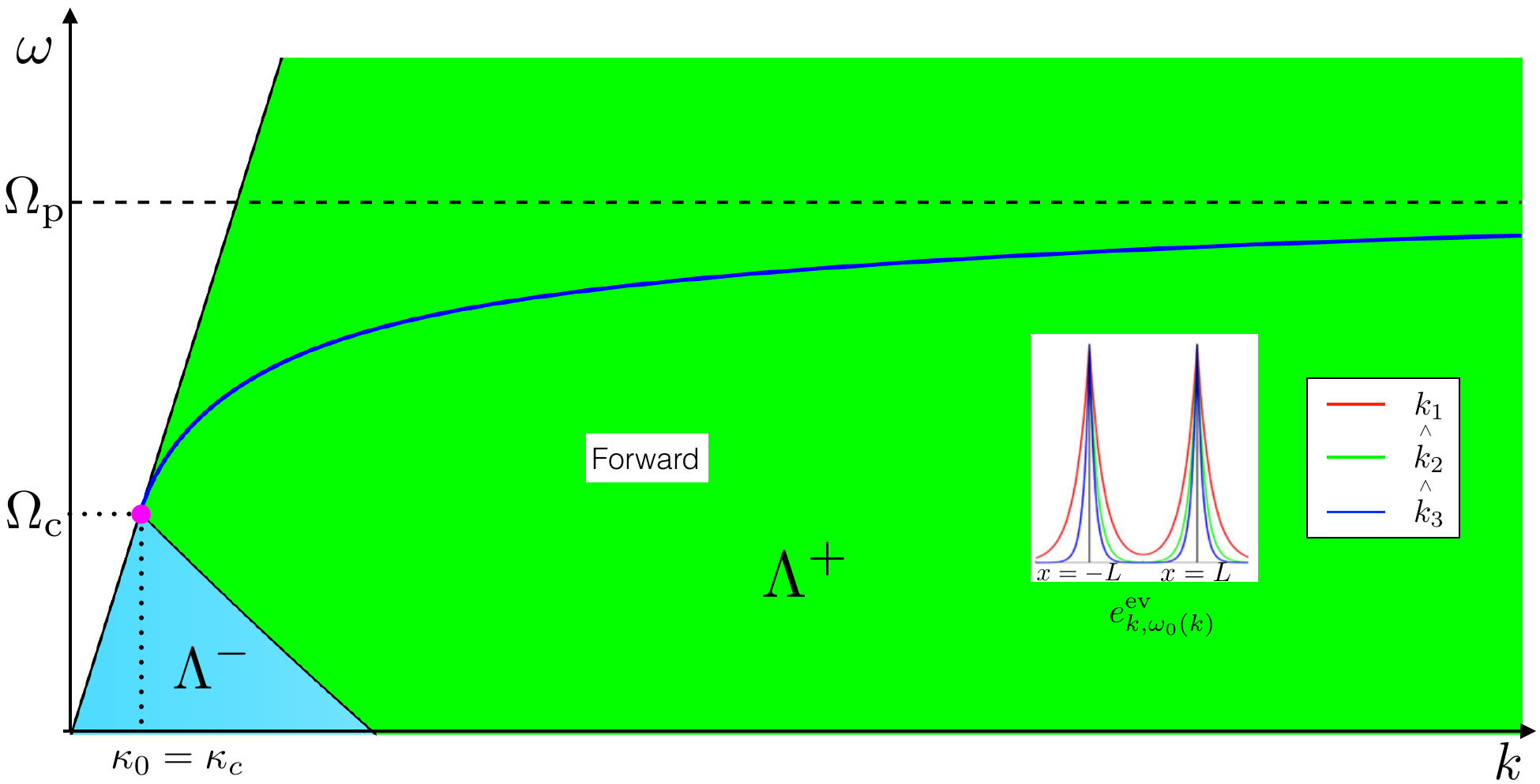}
 \caption{Dispersion curve $\mathcal{C}_0$  when $\rho\geq 1$ (left) and  $\rho<1$ (left) $S_{0,\rho,\Omega_m}(\tau_c)>0$ (right).}
 \label{fig.med-disp-0-case-1-2}
\end{figure}
\begin{figure}[h!]
\centering
 \includegraphics[width=0.495\textwidth]{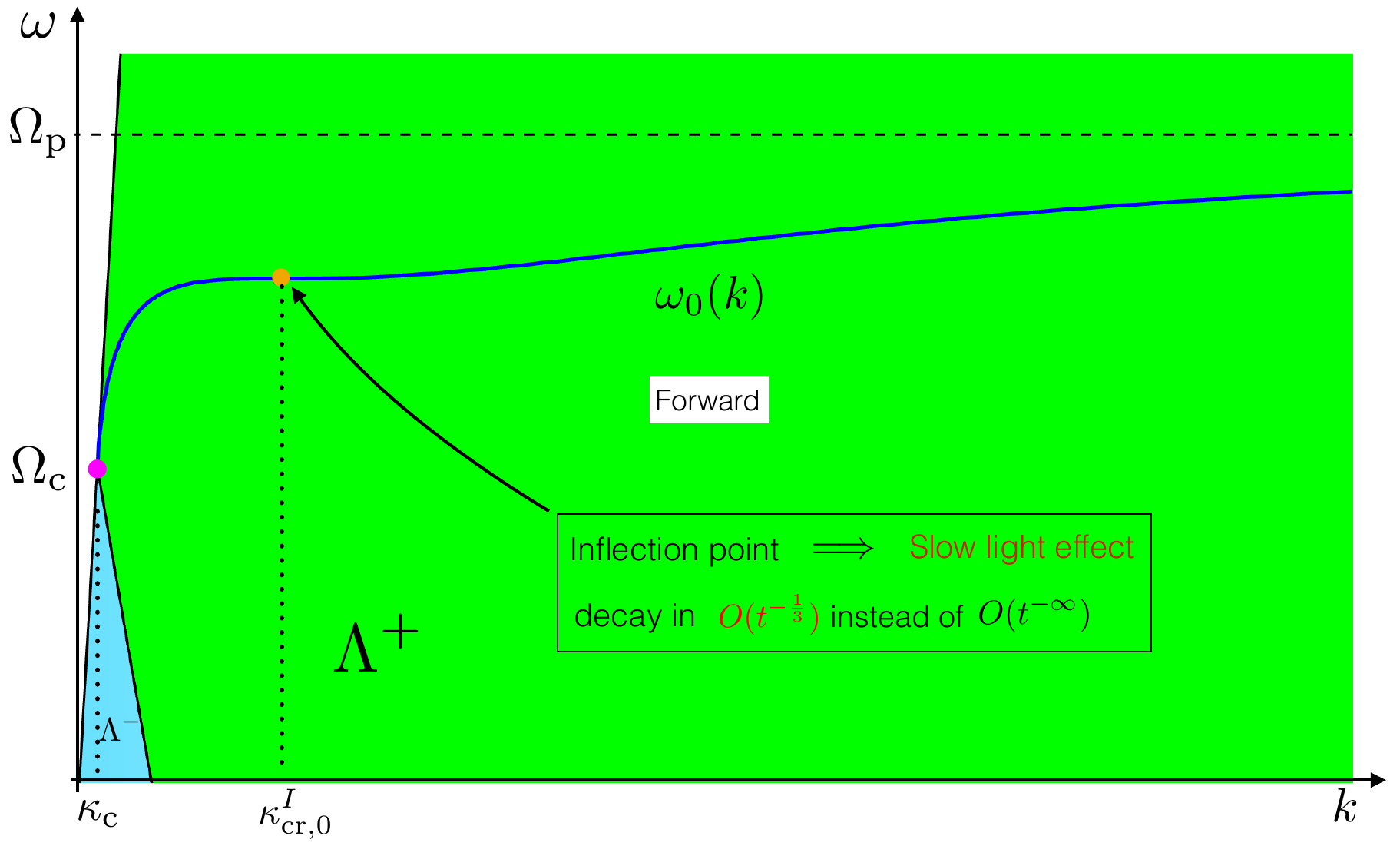}
  \includegraphics[width=0.495\textwidth]{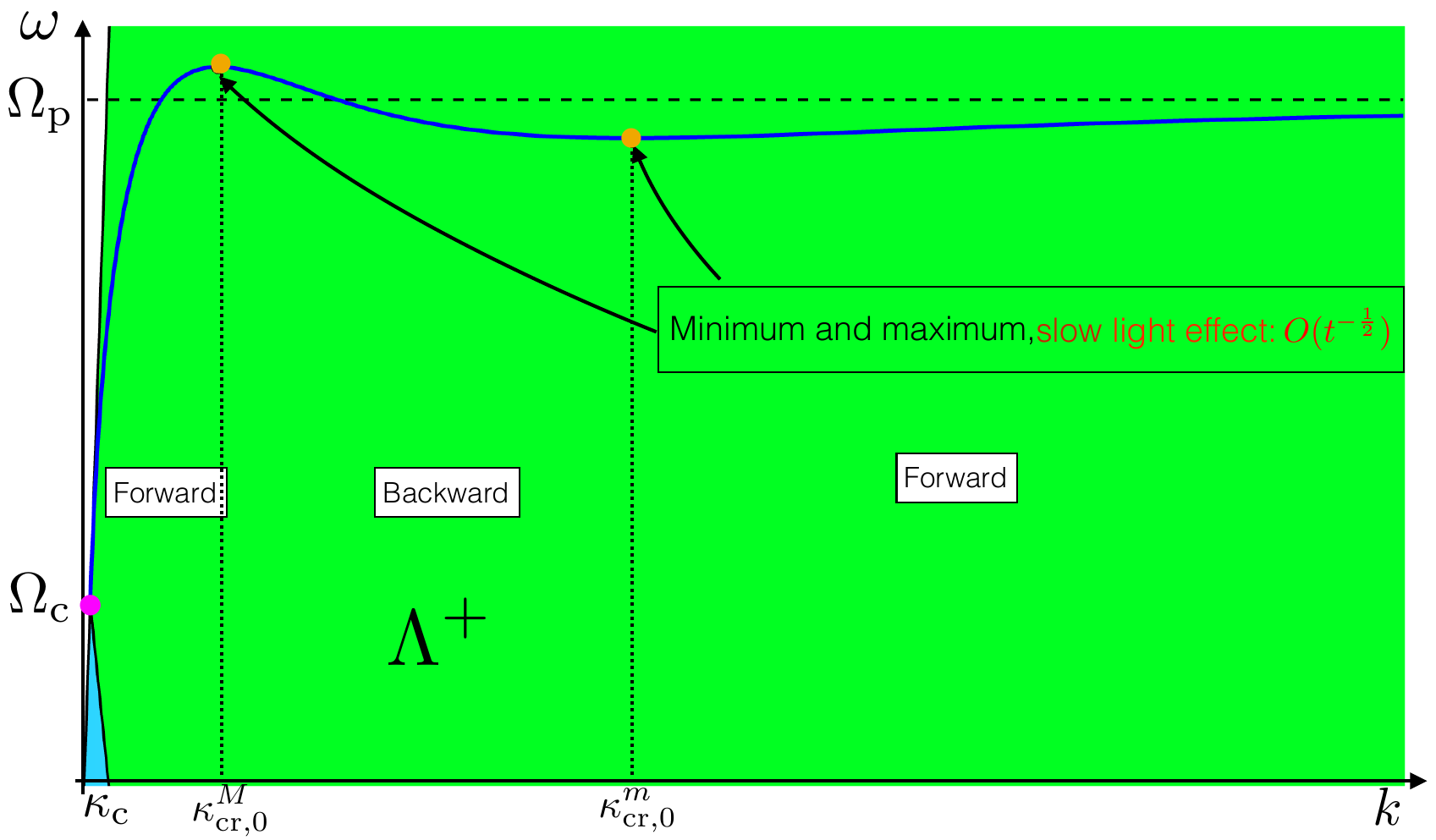}
 \caption{Curve $\mathcal{C}_0$  when $\rho<1$ and $S_{0,\rho,\Omega_m}(\tau_c)=0$ (left) and $\rho<1$ and $S_{0,\rho,\Omega_m}(\tau_c)<0$ (right).}
 \label{fig.med-disp-0-case-3-4}
\end{figure}

\subsubsection{Existence of strong guiding effect}\label{sec-disp-curves-slow-light}
In wave propagation phenomena, the group velocity is related to the  speed of  propagation of wave packets in the  medium.
 A strong particularity of the dispersion curves $\mathcal{C}_{n}$ for $n\geq 0$ of the dispersive slab is the existence of  critical points   $(\kappa_{\mathrm{cr},n},\omega'(\kappa_{\mathrm{cr},n}) )\in \mathcal{C}_n$ where the group velocity $\omega'(\kappa_{\mathrm{cr},n})$ vanishes (see Theorems \ref{thm.omegan} and  \ref{thm.omega0}). 
 Such  points do not exists for  instance for the non-dispersive  dielectric  slab (since $\omega_n'>0$  for $n \in \mathbb{N}$, see  section \ref{classic-case}, Proposition \ref{curves-prop}). \\[6pt]
The existence of critical points $\kappa_{\mathrm{cr},n}$ on the dispersion curve  $\mathcal{C}_n$ is responsible of stronger guiding effect than in  standard situation where they do not exist. This is traduced physically by  a slower decay for large time of the modulus of  the solution associated to a wave packet of guided waves localized in the Fourier space on the  dispersion curve $\mathcal{C}_n$ at the vicinity  the frequency $\omega_n(\kappa_{n,\mathrm{cr}})$. \\[6pt]More precisely,  one considers the Cauchy evolution problem:
\begin{equation}\label{eq.Cauchy}
\frac{\rmd \, \bU}{\rmd\, t} + \rmi\, \bbA \, \bU=0   \quad  \mbox{ with   } \bU(0)=\bU_0  
\end{equation}
where the initial condition is composed of a wave packet of guided waves on the  dispersion  curve $\mathcal{C}_n$  (for some $n\geq 0$)  defined for $(x,y)\in \bbR^2$ by
\begin{equation}\label{eq.IC}
 \bU_0(x,y):=    \int_{k>\kappa_n} \chi_n(k) \, \wlk(x)   \  \rme^{\rmi\,  k   y}\,  \rmd k   \ \mbox{ with } \chi  \geq 0 \mbox{ and }  \chi \in D\big([\kappa_n,+\infty)\big)
\end{equation}
(where $ D\big((\kappa_n,+\infty)\big)$ is the space of $\mathcal{C}^{\infty}-$smooth compactly supported functions in $(\kappa_n,+\infty)$). Moreover the envelope  function $\chi_n$ is chosen such that if the dispersion curve $\mathcal{C}_n$ admits a critical points $\kappa_{\mathrm{cr}, n}$, then  $\chi_n(\kappa_{\mathrm{cr}, n})\neq 0$.  For the particular case where  $n=0$, $\rho<1$  and $\mathcal{S}_{0,\boldsymbol{\Omega}_{\mathrm{m}}, \rho}(\tau_c)<0$, the dispersion curves  $\mathcal{C}_0$ has two critical points   $\kappa_{0}^M$ or  $\kappa_{0}^m$ (see Theorem \ref{thm.omegan}),  thus  one assumes also that the the support of $\chi_0$ contains only one of these  critical  points.\\[6pt]
One can show that the solution of \eqref{eq.IC} is given by 
\begin{equation}\label{eq.sol}
\bU(t,x,y)=\rme^{-\rmi \bbA \, t} \bU_0(x,y)=    \int_{k>\kappa_n} \chi_n(k) \,  \ \wlk(x)      \ \rme^{\rmi\,  (k   y -\omega_n(k) \, t) } \,  \rmd k .
\end{equation}
Thus, one can rewrite $\bU(t,x,y)$ as an oscillatory integral in $t$ of the form:
\begin{equation*}
\bU(t,x,y)= \int_{\kappa_n}^{\infty}  \mathbf{A}(x,y; k)      \ \rme^{- \rmi\,  \phi_n(k) \, t} \,  \rmd k  \mbox{ with }   \mathbf{A}(x,y;k)=  \chi(k) \,  \wlk(x)  \rme^{\rmi\,  k   y}   \  \mbox{ and } \phi_n(k)=\omega_n(k) .
\end{equation*}
Using a stationary phase result (see e.g. \cite{Cri-92} pp. 131 to 133 and  proposition 3 pp 334 of \cite{Stein-93}), one gets  that for $t\to+ \infty$:
\begin{equation*}
\! \bU(t,x,y)=\left\{ \begin{array}{lll}
&O(t^{-k}) \quad   \forall k \in \mathbb{N}, \  \  \mbox{ if }   \forall k \in (\kappa_n,+\infty), \  \omega_n'(k)\neq 0  .\\[10pt]
&  \displaystyle \! B_{2}^{\pm}(x,y) \, \rme^{\rmi  \omega_n(\kappa_{\mathrm{cr},n})\, t }\ t^{-\frac{1}{2}}+O(t^{-1}),    \mbox{ if }  \omega'_n(\kappa_{\mathrm{cr},n})=0, \ \pm \, \omega_n''(\kappa_{\mathrm{cr},n})>0 .\\[10pt]
\displaystyle 
&  \displaystyle \! B_{3}^{\pm}(x,y) \, \rme^{\rmi  \omega_n(\kappa_{\mathrm{cr},n})\, t} \, t^{-\frac{1}{3}}+O(t^{-\frac{2}{3}}),   \mbox{ if } \omega_n'(\kappa_{\mathrm{cr},n})=\omega_n''(\kappa_{\mathrm{cr},n})=0, \,  \pm \, \omega_n^{(3)}(\kappa_{\mathrm{cr},n})>0  .
\end{array}\right.
\end{equation*}
\begin{flalign*}
 \mbox{ where } \qquad \qquad \qquad  &B_2^{\pm}(x,y)=\mathbf{A}(x,y; \kappa_{\mathrm{cr},n}) \left( \frac{2}{|\omega_n^{(2)} ( \kappa_{\mathrm{cr},n}) |}  \right)^{\frac{1}{2}} \sqrt{\pi} \, e^{\pm \,\rmi \frac{\pi}{4} }, &\\
  & B_{3}^{\pm}(x,y)= \mathbf{A}(x,y; \kappa_{\mathrm{cr},n}) \left( \frac{3!}{|\omega_n^{(3)} ( \kappa_{\mathrm{cr},n}) |}  \right)^{\frac{1}{3}} \int_{-\infty}^{\infty}\rme^{\pm \, \rmi x^3} \mathrm{d}x. &
\end{flalign*}
Thus, one observes that if there is no critical point the euclidean norm of the vector $\bU(t,x,y)$ decays more faster than any power of $1/t$ whereas in presence of critical points this decay is much  slower:  in $t^{-\frac{1}{2}}$ for a non degenerate maximum or minimum or  in $ t^{-\frac{1}{3}}$  for a non degenerate inflection point. Thus,  in presence of critical points, the guiding effect is stronger. Therefore, such situations  are refereed in the literature as a  ``slow light''  phenomenon  \cite{Fig-11,Ship-16}. Moreover,  for applications purposes, a particular attention   (see \cite{Fig-11}) is given for  guiding structure where there  exists an inflection point. In our problem, this situation occurs  for the dispersion curve $\mathcal{C}_0$ (when $\rho\leq 1$ and $S_{0,\rho,\Omega_m}(\tau_c)=0$, see  Theorem \ref{thm.omega0}).

%

\appendix
\section{The  case  for a non-dispersive  dielectric slab}\label{classic-case}
We recall here the classical  results obtained  on guided waves when the Drude material in the slab is replaced by a non dispersive dielectric material of  permittivity $\varepsilon_1>0$ and permeability  $\mu_1>0$. These result were first obtained by C. Wilcox in the context of the Pekeris model \cite{Wil-76} for acoustic wave propagation in shallow water. Indeed, the Pekeris model leads to the the same  type of Sturm-Liouville equation as for the analysis of (TE) guided waves propagation in a stratified medium made of a slab of  dielectric embedded in  the vacuum (see e.g. \cite{Wed-91}). Here, as the medium is not dispersive $P=0$ and $\bM=0$ and thus the associated Hilbert space is  $\Hxy:=L^2(\bbR^2)\times L^2(\bbR^2)^2$  with the following  inner product 
$$
(\bU, \bU')_{\mathcal{H}}=\int_{\bbR^2}\eps(\bx)\, E \cdot  \overline{E'}+  \mu(\bx) \, \bH \cdot \overline{ \bH'} \, \rmd \bx,  \quad \forall \bU=(E, \bH) \in \Hxy \mbox{ and  } \  \bU'=(E', \bH')\in \Hxy,
$$
where  $\eps$ and $\mu$ are piece-wise constant function given by  $\varepsilon_1$ and $\mu_1$ in the slab $\mathcal{L}$ and $\varepsilon_0$ and $\mu_0$ in $\bbR^2\setminus \overline{\mathcal{L}}$.
The  self-adjoint Maxwell operator  $\bbA:  \rmD(\bbA)\subset \Hxy \to \Hxy$ is given  by 
\begin{align}\label{eq.kerA}
\bbA & := \ \rmi\, \begin{pmatrix}
0 &\eps(\bx)^{-1}\,\curl \\
- \mu(\bx)^{-1}\,\bcurl& 0 
\end{pmatrix}  \mbox{ with } \ \left\{ \begin{array}{lll} 
& \rmD(\bbA):= &H^{1}(\bbR^2) \times \bH_{\!\curl}(\bbR^2), \\[6pt]
\displaystyle 
&  \ker(\bbA)=&\{ (0, \nabla \phi) \mbox{ with } \phi\in  W^1(\bbR^2)\}. 
\end{array}\right.  
\end{align}
The set of  non-propagative frequencies  reduces here to $\{ 0\}$ and
the  expression of the reduced self-adjoint  operators $(\bbA_{k})_{k\in \bbR}$ (obtained by decomposing $\bbA$ with $\mathcal{F}$) is easily deduced from $\bbA$ by replacing $\curl$ and $\bcurl$ by $\curl_k$ and $\bcurl_k$.\\[6pt]
\noindent The analysis of the above  sections still holds in this setting. The dispersion relation  \eqref{eq.dispersioneven} (related to the even solution of the Sturm-Liouville equation \eqref{eq.Esturmguided} with the new functions $\eps(\cdot)$ and $\mu(\cdot)$ instead of $\eps(\omega,\cdot)$ and $\mu(\omega,\cdot)$)   becomes  (by replacing $\mu_+(\omega)$ by $\mu_1$) 
\begin{equation}\label{eq.dispersiodielec}
\displaystyle \xi^{+}_{k,\omega} \operatorname{ tanh }(\xi^{+}_{k,\omega} \, L)=-\frac{\mu_1}{\mu_0}\, \xi^{-}_{k,\omega}  \quad  \mbox{ for }  \ (k,\omega)\in  \Lambda_\mathrm{nd}=\{  (k,\omega)\in \bbR^+\times \bbR_+ \mid \omega< c_0 \, k\},
\end{equation}
where the waves numbers in the $x$ direction  are given by the following  square roots of $\mathcal{D}^{\pm}_{k,\omega}$:
$$\xi_{k,\omega}^{-}=|\mathcal{D}_{k,\omega}^-|^{1/2}>0  \mbox { with }  \mathcal{D}_{k,\omega}^-=k^2-\omega^2 \varepsilon_0\, \mu_0   \ \mbox{ and } \  \xi_{k,\omega}^{+}= \rmi |\mathcal{D}_{k,\omega}^+|^{1/2} \mbox { with } \mathcal{D}_{k,\omega}^+=k^2-\omega^2 \varepsilon_1\, \mu_1 .$$
The guided waves are here necessary evanescent  in the vacuum (thus $\mathcal{D}_{k,\omega}^->0$). In the slab, one shows that they are necessarily propagative  (i.e. $\mathcal{D}_{k,\omega}^+<0$ ). Indeed, the fact that $\mu_1>0$  imposes, in a similar manner,  as we saw  in the end of the previous paragraph when $\mu_+(\omega)\geq 0$, that they are no evanescent solutions  in the slab (i.e. with $ \xi_{k,\omega}^{+}\geq 0$)  of the equation \eqref{eq.dispersiodielec}.\\[4pt]
Thus we are looking for solutions of the dispersion relations \eqref{eq.dispersiodielec} in  a sub-region of  the positive  quadrant where $\mathcal{D}_{k,\omega}^->0$ and   $\mathcal{D}_{k,\omega}^+<0$. This  region $\Lambda_{\mathrm{nd}}^-\subset \Lambda_{\mathrm{nd}}$  defined by
$$
\Lambda_{\mathrm{nd}}^-= \{ (k, \omega )\in \bbR^+ \times \bbR_+  \mid  c_1  \,k < \omega < c \, k \}
$$
is non-empty if only if $c_1<c$ 
where   $c_1:=(\varepsilon_1 \,\mu_1)^{-1/2}$  is  the speed of light  in the dielectric. \\[4pt]
To put it in a nutshell, a ``even'' guided waves at the frequency $\omega\in \bbR_+$ with wave number $k\in \bbR^+$ exists if only if $(k, \omega)\in \Lambda_{\mathrm{nd}}^-$  is a  solution of the dispersion relations \eqref{eq.dispersiodielec}. Therefore, we introduce the set
\begin{equation*}
\mathcal{D}_{\mathrm{ev}, \mathrm{nd} }=\{  (k,\omega)\in \Lambda_{\mathrm{nd}}^-\mid  \displaystyle \xi^{+}_{k,\omega} \operatorname{ tanh }(\xi^{+}_{k,\omega} \, L)=-\frac{\mu_1}{\mu_0}\, \xi^{-}_{k,\omega}  \}.
\end{equation*}
So that, the following  equivalence holds:
\begin{equation}\label{eq.equivalnecespectrumdispersionrelation}
\omega \in\sigma_p^{\mathrm{ev}}(\bbA_k) \Longleftrightarrow \, ( |k|, |\omega|)\in \mathcal{D}_{\mathrm{ev}, \mathrm{nd} }.
\end{equation}
An explicit  parametrization of   the solutions   of the equation \eqref{eq.dispersiodielec} (see \cite{Wil-76} for the details) yields the following Proposition on the characterization of the set $\mathcal{D}_{\mathrm{ev},\mathrm{nd}}$.
\begin{Pro}\label{curves-prop}
If  $c_1<c$, then
\begin{equation}\label{eq.ev-odN-}
\mathcal{D}_{\mathrm{ev}, \mathrm{nd} }=\bigcup_{n \in \mathbb{N}} \mathcal{C}_{n} 
\end{equation}
where the  curves $\mathcal{C}_n$ (referred as the dispersion curves) are defined for all $n\in \mathbb{N}$ by
$$
\mathcal{C}_n=\{ (k, \omega_n(k))  \mid k> \kappa_n \} \ \mbox{ where } \  \kappa_n:=n \, c_1\, (c^2-c_1^2)^{-\frac{1}{2}} \frac{\pi}{L} \mbox{ is increasing to $+\infty$ with $n$}, 
$$
and $\omega_n :[\kappa_n,+\infty)\to \bbR^+$ is an  analytic strictly increasing function satisfying $\omega_n <\omega_{n+1}$ and 
$$
\omega_n(\kappa_n)=c \, \kappa_n \to +\infty \mbox{ as } n \to +\infty ,\quad   \omega_n'(\kappa_n)=c ,  \ \mbox{ and } \quad   \omega_{n}(k)=c_1 k +o(k)  \mbox{ as } k \to +\infty .
$$
 \end{Pro}
\noindent Thus,  by Proposition \ref{curves-prop}, the even part of the point spectrum $\sigma_p^{\mathrm{ev}}(\bbA_k)$ (defined by \eqref{eq.equivalnecespectrumdispersionrelation}): 
$$
\sigma^{\mathrm{ev}}_p(\bbA_k)= \{ \pm \, \omega_n (|k|), \,  n  \in \mathbb{N}  \mid \kappa_n<|k| \}
$$
is a finite set since $\kappa_n \to +\infty$, as $n\to +\infty$.
\begin{figure}[h!]
\centering
 \includegraphics[width=0.84\textwidth]{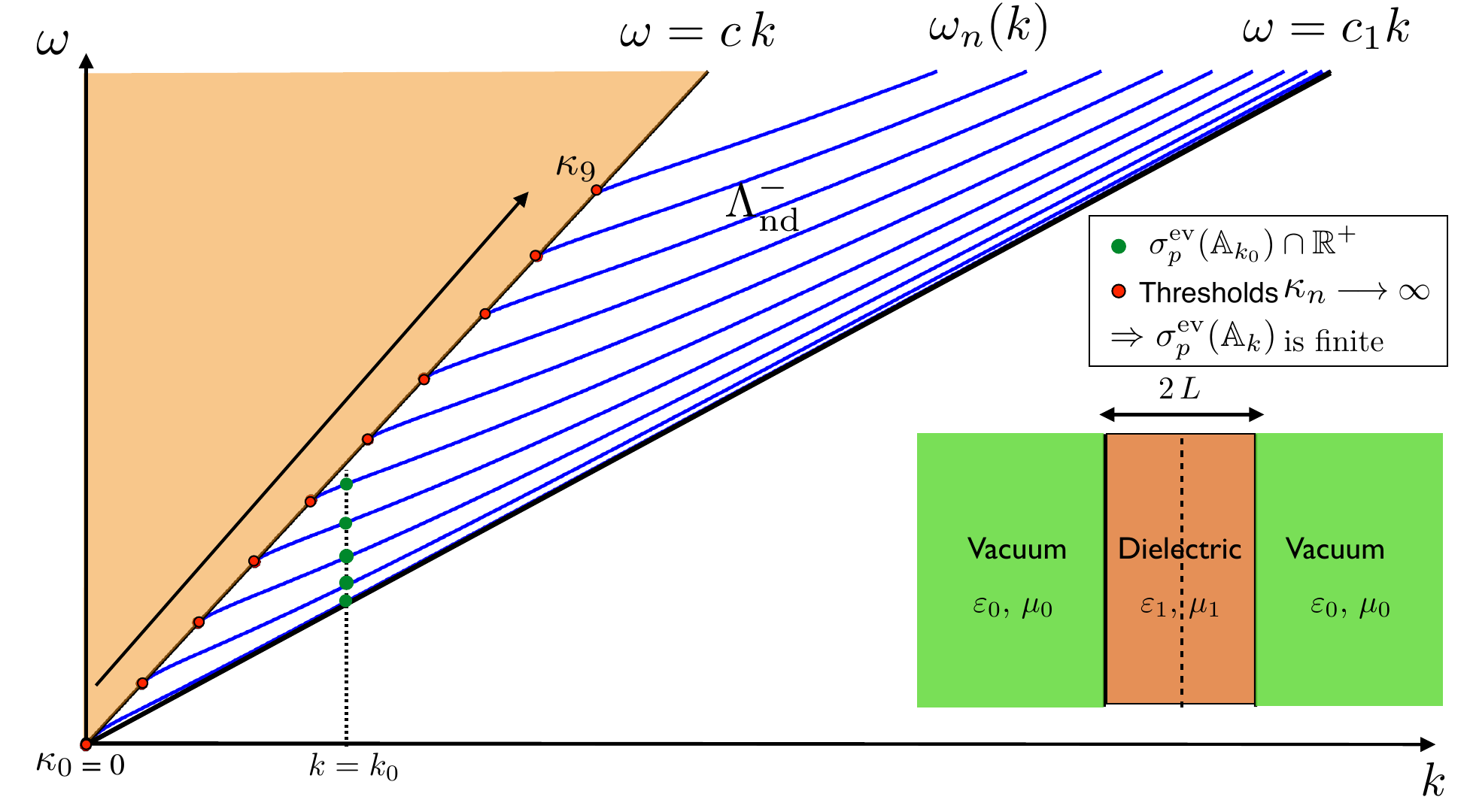}
 \caption{Dispersion curves of guided waves  (even case) for a slab  of width $2L$ made of  a non-dispersive dielectric  embedded in the vacuum for which $c_1=(\varepsilon_1 \,\mu_1)^{-1/2}<c=(\varepsilon_0 \,\mu_0)^{-1/2}$.}
 \label{fig.med-dielectric}
\end{figure}

\subsection*{Acknowledgements}
The authors would like to thank their colleague  Christophe Hazard for his collaboration on the two articles  \cite{Cas-Haz-Jol-17} and \cite{Cas-Haz-Jol-22} whose topic corresponds to the subject treated in section \ref{sect-transm}. The authors are also  grateful to their former Ph. D. student  Luis Alejandro Rosas Mart\'inez  \cite{Ros-23} for his collaboration on the work presented in section  \ref{sec-slab}.



\end{document}